\documentclass[11pt,letterpaper]{article}

\usepackage[verbose=true,letterpaper,margin=1in]{geometry}

\usepackage[utf8]{inputenc}
\usepackage[T1]{fontenc}
\usepackage{lmodern}
\usepackage{authblk}

\usepackage{booktabs}
\usepackage{threeparttable}

\usepackage{amsfonts}
\usepackage{amsmath}
\usepackage{amssymb}

\usepackage{nicefrac}
\usepackage{xcolor}
\usepackage{graphicx}

\usepackage{float} 

\usepackage[natbibapa]{apacite}

\AtBeginDocument{%
  \def\doi#1{\url{https://doi.org/#1}}}

\AtBeginDocument{
  
}

\usepackage{amsthm}
\newtheorem{lemma}{Lemma}
\newtheorem{definition}{Definition}
\newtheorem{theorem}{Theorem}
\newtheorem{corollary}{Corollary}[theorem]

\usepackage[hidelinks]{hyperref}
\usepackage{url}

\DeclareMathOperator{\tr}{tr}
\newcommand{\ci}{\perp\!\!\!\perp}
\newcommand{\ssep}{\mathop{;}\,}

\newcommand\blfootnote[1]{%
  \begingroup
  \renewcommand\thefootnote{}\footnote{#1}%
  \addtocounter{footnote}{-1}%
  \endgroup
}

\usepackage{setspace}

\title{{\textbf{A Restricted Latent Class Hidden Markov Model for Polytomous Responses, Polytomous Attributes, and Covariates: Identifiability and Application}}}

\author[1]{Eric Alan Wayman}
\author[2]{Steven Andrew Culpepper}
\author[2]{Jeff Douglas}
\author[1]{Jesse Bowers}
\affil[1]{Independent scholar \protect\\
\texttt{ericwaymanpublications@mailworks.org},  \texttt{bowers.jesse+publications@gmail.com}}
\affil[2]{Department of Statistics, University of Illinois Urbana-Champaign \protect\\
\texttt{sculpepp@illinois.edu}, \texttt{jeffdoug@illinois.edu}}

\date{}

\begin{document}

\maketitle

\begin{abstract}
  We introduce a restricted latent class exploratory model for longitudinal data with ordinal attributes and respondent-specific covariates. Responses follow a time inhomogeneous hidden Markov model where the probability of a respondent's latent state at the current time point is conditional on the respondent's latent state at the previous time point as well as the respondent's covariates at the current time point. We prove that the model is identifiable, state a Bayesian formulation, and demonstrate its efficacy in a variety of scenarios through two simulation studies. We apply the model to response data from a mathematics examination, comparing the results to a previously published confirmatory analysis, and also apply it to emotional state response data which was measured over a several-day period.
\end{abstract}

\section{Introduction}

\noindent Latent class models \citep{goodman1974exploratory} are an exploratory method of classifying respondents into a finite set of latent or unobserved classes, under the assumption that responses are conditionally independent given the class membership of respondents. Restricted latent class models, or RLCMs \citep{haertel1984application, haertel1990continuous, vermunt2001use}, are a type of latent class model where the restrictions on certain parameters allow for an interpretation of the relationship between the response data and the latent state of the respondents. The discrete aspect of these latent states makes RLCMs useful in settings where classification of the latent state can serve the purpose of finding values of latent attributes with an interpretation that is of interest to the researcher and respondents. \blfootnote{This article was published in the \emph{Journal of Educational and Behavioral Statistics} (February, 2026): \url{https://doi.org/10.3102/10769986251415569}}

A major type of model where latent class membership is modeled over time is the hidden Markov model \citep{wiggins1955mathematical,baum1966statistical}. In the literature, a specific type of hidden Markov model that has been used to model latent class membership over time is referred to by the name ``latent transition analysis'' \citep{collins1992latent,hagenaars1990categorical,poulsen1983latent,van1990mixed}. A hidden Markov model has three components: the emission probabilities, which are the probabilities of response values conditional on the value of the latent state, the transition probabilities between time points, which are the probabilities of latent state values at a time point conditional on the latent state value at the previous time point, and the marginal probability of the latent state at the initial time point \citep{miller2016advanced}. Several latent class models for which the transition model is a hidden Markov model have been introduced. \cite{wang2018tracking} introduced a model with binary attributes and binary attributes that models the transition matrix as a function of covariates. \cite{chen2018hidden} presented a model that uses a categorical distribution over the number of possible trajectories. In the model of \cite{zhang2020multilevel}, interventions are related to the changes in the latent state. \cite{li2016latent} and \cite{kaya2017assessing} use the deterministic-input noisy-and-gate (DINA) model and both the DINA and the deterministic input noisy-or-gate model (DINO) models respectively as measurement models in a latent transition analysis framework. The \texttt{R} package \texttt{TDCM} \citep{madison2025tdcm} provides functionality for fitting the transition diagnostic classification model \citep{madison2018assessing}, a confirmatory model (where the latent structure is prespecified) with binary attributes and binary data.

Two relevant longitudinal models which do not fit into the hidden Markov model framework are that of \citet{chen2020multivariate}, for which the probability of each latent state is conditional on shared parameters and where the latent state is modeled by a multivariate probit specification, and that of \citet{bartolucci2009multivariate}, which is a longitudinal model for polytomous data with random effects that take on discrete values.

We consider latent class models in which the latent classes arise from vectors where each component is a level of a latent attribute. Much of the previous work of this type of RLCM utilized binary latent attributes. When performing a diagnosis, attributes with multiple levels (polytomous attributes) allow for a richer description of a condition or knowledge state. A concept utilized by some models of this type is the \(Q\)-matrix \citep{rupp2010diagnostic}. \(Q\)-matrices specify a relationship between the latent state and the response values: they enter models in a manner similar to variable selection matrices. The particular form the \(Q\)-matrix takes depends on the model being specified.

There has been work on models that utilize polytomous attributes. \cite{chen2013general} introduced a model where the attributes are expert-defined. In the model of \citet{sun2013polytomous}, the attributes have particular levels relating to a polytomous \(Q\)-matrix. The models of \citet{he2023sparse} and \citet{wayman2025restricted} both introduced RLCMs which are similar to the model presented in this manuscript, but are only for the cross-sectional setting. The model presented in Chapter 5 of \citet{bartolucci2012latent} is similar to the model we introduce here, but  our model uses a multivariate probit specification for the transition model which incorporates covariates through the mean of the underlying continuous random vector, a parameterization not included there.

Regarding identifiability, \citet{liu2013theory} established identifability for a cross-sectional model with binary attributes and binary responses which uses a \(Q\)-matrix. \citet{xu2017identifiability} and \citet{xu2018identifying} established identifiability for a cross-sectional model for with binary attributes, binary responses, a particular monotonicity condition, and a \(Q\)-matrix that satisfies certain conditions. More recently, \citet{liu2023identifiability} established the identifiability of a longitudinal model with binary responses, binary attributes whose transition probabilities follow a hidden Markov model, and a \(Q\)-matrix satisfying certain conditions. 

This paper introduces an RLCM where latent states consisting of polytomous attributes change over time and where covariates can help explain transitions amongst components of the latent attribute vectors. We do this by extending to the longitudinal setting the RLCM of \citet{wayman2025restricted} for cross-sectional data where respondent-specific covariates are related to the respondent's latent state through a multivariate probit model. Our model uses a restricted hidden Markov model to uncover structure in the attribute change process: observed responses occur with a probability conditional on the value of a latent variable (the emissions probabilities) and the transition probabilities for each latent state value are conditional only on the previous time point. Just as exploratory restricted latent class models have proven useful for diagnostic models by finding structure in general finite mixture models in the cross-sectional setting, when we move to the longitudinal setting, we have an analogous benefit in finding greater structure in hidden Markov models.

Compared to some previous models, ours is exploratory rather than confirmatory. We demonstrate in our education application that this exploratory model improves upon the best confirmatory model \citep{tang2021does} for a particular dataset \citep{zhan2021data} measuring performance on a mathematics assessment and the effectiveness of a particular intervention. That model, the sLong-DINA, has a unidimensional higher order factor, whereas our model utilizes a multivariate probit whose correlation matrix can capture more associations amongst the latent attributes. Combined with the fact that the \(Q\)-matrix is not pre-specified, this leads to a more accurate representation of the underlying attributes and how they relate to performance.

The structure of this paper is as follows. We first outline the main components of the model, and then state a Bayesian formulation of the model from whose posterior we can sample. We then show that the model is identifiable. Next, we describe the sampling algorithm, which makes use of parameter expansion. We display simulation results which demonstrate the accuracy of parameter estimation in a variety of scenarios. We then apply the model to two different sets of longitudinal data: data that were gathered as part of an education study, and data that were gathered to study respondents' emotional states over time. We conclude with a discussion.

\section{Methodology}

In the following, for any \(S \in \mathbb{N}\), let \([S]\) denote the set \(\{1, 2, \ldots, S\}\), and for any matrix \(A\) let \(A^\prime\) denote the transpose of \(A\). We observe the responses of \(N\) respondents at \(T\) time points to the same questionnaire of \(J\) questions with ordinal responses (``items''); we denote the response of respondent \(n\) at time point \(t\) as \(Y_n^t = (Y_{n1}^t, \ldots, Y_{nJ}^t)\), where for each \(j \in [J]\) we have \(Y_{nj}^t \in \{0, 1, \ldots, M_j - 1\}\). We also observe a respondent-specific vector of \(D\) covariates, denoted \(X_n^t \in \mathbb{R}^{1 \times D}\). We assume each respondent has a \(K\)-dimensional latent state (also referred to as the ``attribute profile'') at each time point \(t\), which we denote as \(\alpha_n^t = (\alpha_{n1}^t, \ldots, \alpha_{nK}^t) \in \{0, \ldots, L - 1\}^K =: A_L\), where \(L\) is a fixed natural number.

As a time inhomogeneous \citep{seneta2006non} hidden Markov model, our model includes both an emissions matrix, here a matrix of latent state-conditional item response probabilities which is the same for all respondents and all time points, as well as a set of transition matrices between all time points and a vector of probabilities for the latent state at the initial time point, both of which can vary across time points and across respondents. We denote the \((\sum_{j=1}^J M_j) \times L^K\) emissions matrix by \(B = \left(p(Y_{n}^t \mid \alpha_n^t, \theta_m)\right)\) (this matrix is the same for all \(n\) and all \(t\)). This component of the model is the measurement model and we thus denote its parameters as \(\theta_m\).

We write the \(L^K \times L^K\) transition matrix for respondent \(n\) from time \(t-1\) to time \(t\) as \(U_{n, t, t-1} = \left(p(\alpha_n^t \mid \alpha_n^{t-1}, \theta_s)\right)\), and the vector of marginal probabilities (of dimension \(L^K\)) for the latent state of respondent \(n\) at time \(t\) as \(\pi_n^t = (p(\alpha_n^t \mid \theta_s))\). The set of transition matrices for all \(n\) for \(t \in \{2, \ldots, T\}\) and the initial latent state probabilities for all \(n\) at \(t = 1\) form the structural model, so we denote the parameters relevant to these components as \(\theta_s\).

\begin{figure*}
  \centering
  \includegraphics[width=0.4\textwidth]{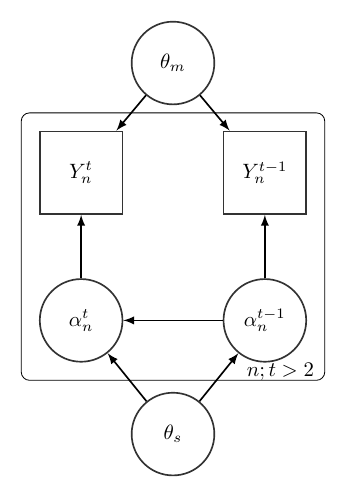}
  \caption{\label{fig:simplifiedmodel} Simplified version of model in directed graphical model form (for alt text, see Supplementary Material A)}
\end{figure*}

We now describe the components of the model, namely the measurement model, structural model, and a monotonicity condition regarding the latent state and the response data. A simplified version of the model is shown in Figure \ref{fig:simplifiedmodel} in directed graphical model form \citep{murphy2012machine}.

\subsection{Measurement Model}

The measurement model forms the elements of the emissions matrix \(B = (p(Y_{n}^t \mid \alpha_n^t, \theta_m))\). It is a cumulative probit link model \citep{agresti2015foundations}, namely, for all \(n \in [N]\), \(t \in [T]\), \(j \in [J]\) and \(m \in \left\{0, \ldots, M_j - 1\right\}\), \begin{equation}
  p(Y_{nj}^t = m \mid \alpha_n^t, \theta_m) = p(Y_{nj}^t = m \mid \alpha_n^t, \beta_j, \kappa_j) = \Phi(\kappa_{j, m + 1} - d_n^t \beta_{j}) - \Phi(\kappa_{j, m} - d_n^t \beta_{j})
\end{equation}

\noindent where \(\theta_m = (\kappa, \beta)\), \(\Phi\) denotes the cdf of the standard normal distribution, where for each \(j\), \(\kappa_{j0} < \kappa_{j1} < \cdots < \kappa_{j M_j}\) and \(\kappa_{j0} = -\infty, \kappa_{j1} = 0, \kappa_{j M_j} = \infty\) for identifiability reasons, and where \(d_n^t\) is a function of \(\alpha_n^t\) and is called the design vector of respondent \(n\).

We use a cumulative coding \citep{he2023sparse,wayman2025restricted} of \(\alpha_n^t\) in order to relate the effect of potentially each level and dimension of \(\alpha_n^t\) to the observed response \(Y_{nj}^t\). Specifically, for \(\alpha_n^t = (\alpha_{n1}^t, \ldots, \alpha_{nK}^t)\) we define, for all \(k \in [K]\), the functions \(d_k: A_L \to \{0, 1\}^L\) by \(d_k(\alpha_n^t) = (I(\alpha_{nk}^t \geq 0), I(\alpha_{nk}^t \geq 1), \ldots, I(\alpha_{nk}^t \geq L - 1))\). We define the function \begin{equation}d: A_L \to \{0, 1\}^{\prod_{k=1}^K L}\end{equation} by \(d(\alpha_n^t) = d_1(\alpha_n^t) \otimes d_2(\alpha_n^t) \otimes \cdots \otimes d_K(\alpha_n^t)\), and we write the value of \(d\) evaluated at \(\alpha_n^t\) as \(d_n^t\). Our definition for \(d_n^t\) corresponds with the saturated model that includes all main effects and interaction terms. We may fit reduced models by only using a subset of components of the design vector, which includes only the interactions we desire. For example, reduced models might only include main effects (order \(1\)), or main effects and two-way interactions (order \(2\)), up to a saturated model of order \(K\).

\subsection{Monotonicity Condition}

For two arbitrary \(u, v \in A_L\), \(u = (u_1, \ldots, u_K)\) and \(v = (v_1, \ldots, v_K)\), we write \(u \geq v\) if for all \(k \in [K]\) we have \(u_k \geq v_k\).

So that our ordinal latent state vectors have an interpretation related to observable ordinal quantities, we impose a monotonicity condition used in earlier models \citep{wayman2025restricted,culpepper2019exploratory}: for all \(t \in [T]\), \(n \in [N]\), and \(t \in [T]\),
\begin{equation}\label{monotoncond}
\forall\, u, v \in A_L \quad u \geq v \implies p(Y_{nj}^t > m \mid u, \beta_j, \kappa_j) \geq p(Y_{nj}^t > m \mid v, \beta_j, \kappa_j),
\end{equation}

\noindent equivalently, for all \(u, v \in A_L, u \geq v \implies d_u \beta_j \geq d_v \beta_j\) (where \(d_u\) is the design vector associated with \(u\)). This monotonicity condition restricts the parameter space for \(\beta_j\).

\subsection{Structural Model}

The structural model forms the transition matrices, i.e. for \(t \in \{2, 3, \ldots, T\}\), \(U_{n, t, t-1} = (p(\alpha_n^t \mid \alpha_n^{t-1}, \theta_s))\), as well as the initial latent state probabilities \(\pi_n^t = (p(\alpha_n^t \mid \theta_s))\) for \(t = 1\). The structural model is a multivariate probit model \citep{mcdonald1967nonlinear,ashford1970multi,christoffersson1975factor,muthen1978contributions} where for \(t \in \{2, \ldots, T\}\), the latent state at time point \(t\), \(\alpha_n^t\), is related to its value at time point \(t - 1\) as well as the covariates \(X_n^t\) through the mean of an underlying multivariate normal random variable, namely \begin{align}\label{mvprobit2T}
  p(\alpha_n^t \mid \alpha_n^{t - 1}, \theta_s) & = p(\alpha_n^t \mid \alpha_n^{t - 1}, \gamma, \lambda, \xi, R) \notag\\
  & = \int_{\gamma_{K\alpha_{nK}^t}}^{\gamma_{K, \alpha_{nK}^t + 1}} \ldots \int_{\gamma_{1\alpha_{n1}^t}}^{\gamma_{1, \alpha_{n1}^t + 1}} \phi_K(\alpha_n^{\ast, t}; X_n^t \lambda + d_{n, \text{otr}}^{t - 1} \xi, R) d\alpha_n^{\ast, t},
\end{align}

\noindent where \(\theta_s = (\gamma, \lambda, \xi, R)\), where \(\phi_K(x; a, b)\) is the density function of a multivariate normal random variable of dimension \(K\) in which \(x\) is the variable (row vector), \(a\) is the mean, and \(b\) is the covariance, and where \(d_{n, \text{otr}}^t\) is the design vector for the structural model with order otr (where \text{otr} stands for ``order of transition model''). For each \(k \in [K]\), we assume \(\gamma_{k0} < \gamma_{k1} < \cdots < \gamma_{k L}\), where we set \(\gamma_{k0} = -\infty\), \(\gamma_{k1} = 0\), and \(\gamma_{k L} = \infty\) for identifiability.

For \(t = 1\), we assume \begin{equation}\label{mvprobit1}
p(\alpha_n^1 \mid \theta_s) = p(\alpha_n^1 \mid \gamma, \lambda, R) = \int_{\gamma_{K\alpha_{nK}^1}}^{\gamma_{K, \alpha_{nK}^1 + 1}} \ldots \int_{\gamma_{1\alpha_{n1}^1}}^{\gamma_{1, \alpha_{n1}^1 + 1}} \phi_K(\alpha_n^{\ast, 1}; x_n^1 \lambda, R) d\alpha_n^{\ast, 1}
\end{equation}

\noindent We choose a correlation structure rather than a covariance structure for identifiability reasons (see Section 4).

We now detail the Bayesian model we implement that reflects the above components.

\section{Bayesian Model}

\begin{table}
\small
\begin{center}
\caption{\label{tab:paramslist} Variable and parameter descriptions for Bayesian model}
\vspace{0.5\baselineskip}
\begin{threeparttable}[t]
\begin{tabular}{lll}
\toprule
     & Symbol & Description \\
\midrule
\((Y, \alpha)\) & & \\
& \(Y\) & responses \\
& \(\alpha\) & latent states \\
 \midrule
     \(\theta_m\) & & \\
& \(\kappa\) & thresholds for measurement model \\
& \(\beta\) & slope parameters relating latent states to responses \\
  \midrule
     \(\theta_s\) & & \\
& \(\gamma\) & thresholds for multivariate probit specification \\
& \(\lambda\) & slope parameters relating covariates to latent states \\
& \(\xi\) & slope parameters relating latent states between time points \\
& \(R\) & underlying correlation matrix for latent states \\
  \midrule
     \(\theta_o\) & & \\
& \(Y^{\ast}\) & augmented variables for responses \\
& \(\alpha^{\ast}\) & augmented variables for latent states \\
& \(\delta\) & activation indicator variables for \(\beta\) \\
& \(\omega\) & part of prior for \(\delta\) \\
& \(V\) & diagonal of covariance matrix \\
\bottomrule
\end{tabular}
  \begin{tablenotes}
    \footnotesize
    \item \(\theta_m\) refers to measurement model parameters, \(\theta_s\) to structural model parameters, and \(\theta_o\) to other parameters.
  \end{tablenotes}
\end{threeparttable}
\end{center}
\end{table}

\begin{figure*}[t]
  \centering
  \includegraphics[width=0.8\textwidth]{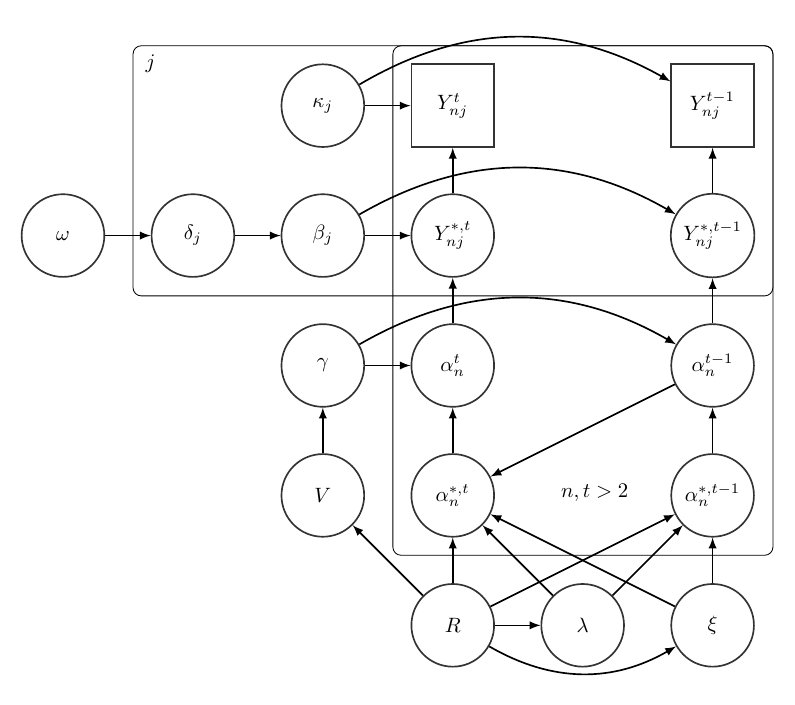}
  \caption{\label{fig:dgm-f1} Bayesian model in directed graphical model form, part one (for alt text, see Supplementary Material A)}
\end{figure*}

\begin{figure*}[t]
  \centering
  \includegraphics[width=0.3\textwidth]{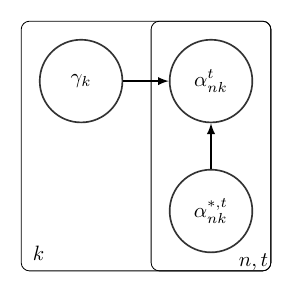}
  \caption{\label{fig:dgm-f2} Bayesian model in directed graphical model form, part two (for alt text, see Supplementary Material A)}
\end{figure*}

Our Bayesian model is formulated as a directed graphical model \citep{murphy2012machine}, the graph \(G\) for which is displayed in Figures \ref{fig:dgm-f1} and \ref{fig:dgm-f2} using plate notation \citep{murphy2012machine}. We label the set of vertices in \(G\) as \(Z = (Y, \alpha, \theta, \theta_o)\), where \(\theta = (\theta_m, \theta_s)\) is the set of measurement model parameters \(\theta_m\) and structural model parameters \(\theta_s\), and \(\theta_o\) denotes other parameters introduced for computational purposes; the variables and parameters are listed in Table \ref{tab:paramslist}. Parameters with an asterisk superscript are present for computational purposes.

Since our model \(p(Z)\) is a directed graphical model \citep{murphy2012machine} it ``admits a recursive factorization according to \(G\)'' \citep{lauritzen1996graphical}, i.e. \(p(Z) = \prod_{z \in Z} p(z \mid \text{pa}(z))\) where \(\text{pa}(z)\) refers to the parents of vertex \(z\).

From the recursive factorization, we have \begin{align}
  p(Z) & = p(Y \mid Y^\ast, \kappa) \cdot p(Y^\ast \mid \beta, \alpha) \cdot p(\kappa) \cdot p(\beta \mid \delta) \cdot p(\delta \mid \omega) \cdot p(\omega) \notag\\
  & \quad\quad \cdot \left[\prod_{t=2}^T p(\alpha^t \mid \alpha^{\ast, t}, \gamma) \cdot p(\alpha^{\ast, t} \mid \alpha^{t-1}, \lambda, \xi, R)\right] \cdot p(\alpha^1 \mid \alpha^{\ast, 1}, \gamma) \cdot p(\alpha^{\ast, 1} \mid \lambda, R) \notag\\
& \quad\quad \cdot p(\gamma \mid V) \cdot p(\lambda \mid R) \cdot p(\xi \mid R) \cdot p(V, R) \label{origmodel}
\end{align}

\noindent where we have taken the additional step of writing \(p(V \mid R) \cdot p(R) = p(V, R)\). Note that we have grouped across subscripts.

Since the model is a directed graphical model, it obeys the directed local Markov property \citep{lauritzen1996graphical} relative to \(G\), i.e. for any variable \(z \in Z\) we have \(z \ci \text{nd}(z) \mid \text{pa}(z)\) where \(\text{nd}(z)\) denotes the non-descendants of \(z\). Note that this implies \(z \ci (\text{nd}(z) \setminus \text{pa}(z)) \mid \text{pa}(z)\).

We now specify assumed relationships (likelihood, priors) for factors that appear in \eqref{origmodel}.

\subsection{Likelihood and Data Augmentation Prior for Observed Data}

We assume for all \(n \in [N]\), \(t \in [T]\), and \(j \in [J]\), \begin{equation}
  p(Y_{nj}^{\ast, t} \mid \alpha_n^t, \beta_j) = \phi(Y_{nj}^{\ast, t};\, d_n^t \beta_j, 1)
\end{equation}

\noindent where \(\phi(x; a, b)\) is the normal density with variable \(x\), mean \(a\), and variance \(b\), and \begin{equation}
  p(Y_{nj}^t \mid Y_{nj}^{\ast, t}, \kappa_j) = \sum_{m=0}^{M_j - 1} I(Y_{nj}^t = m) \cdot I(\kappa_{j m} < Y_{nj}^{\ast, t} \leq \kappa_{j, m + 1})
\end{equation}

\noindent which yields (see Supplementary Material B) \begin{equation}
p(Y_{nj}^t \mid \alpha_n^t, \beta_j, \kappa_j) = \int_{\kappa_{j Y_{nj}^t}}^{\kappa_{j, Y_{nj}^t + 1}} \phi(Y_{nj}^{\ast, t};\, d_n^t \beta_j, 1) dY_{nj}^{\ast, t}. \label{classcondresponseprob1}
\end{equation}

\noindent We assume \(p(\kappa_j) = I(-\infty = \kappa_{j0} < \kappa_{j1} < \cdots < \kappa_{j Mj} = \infty)\).

\subsection{Priors for Measurement Model Coefficients and Related Parameters}

We use the Dirac spike and normal slab prior for variable selection \citep{kuo1998variable} for each \(\beta_j\), namely, for all \(j \in [J]\):

\begin{equation}p(\beta_j \mid \delta_j) = c_j(\delta_j) \cdot I(\beta_j \in \mathcal{R}_j) \cdot \prod_{h=1}^H p(\beta_{hj} \mid \delta_{hj}), \quad p(\delta_j \mid \omega) = \prod_{h=1}^H p(\delta_{hj} \mid \omega)\end{equation} \begin{equation}p(\beta_{hj} \mid \delta_{j}) = p(\beta_{hj} \mid \delta_{hj}) = I(\delta_{hj} = 0) \cdot \mathrm{\Delta}(\beta_{hj}) + I(\delta_{hj} = 1) \cdot \phi(\beta_{hj}; 0, \sigma_{\beta}^2)\end{equation}

\begin{equation}
\mathcal{R}_j := \Big\{\beta_j \in \mathbb{R}^H : \forall\, u, v \in A_L \quad u \geq v \implies d_u \beta_j \geq d_v \beta_j \Big\}
\end{equation}

\noindent where \(\delta_{hj} \mid \omega \sim \text{Bernoulli}(\omega)\), \(\Delta\) is the Dirac delta generalized function, and \(\mathcal{R}_j\) is the region reflecting the monotonicity condition \eqref{monotoncond}. As in \citet{chen2020sparse} and \citet{wayman2025restricted}, we let \(\omega \sim \text{Beta}(\omega_0, \omega_1)\) where \(\omega_0\) and \(\omega_1\) are hyper-parameters.

\subsection{Priors for Structural Model Parameters}

For all \(n \in [N]\) and all \(t \in [T]\), we assume \begin{equation}\label{prioralpha}
p(\alpha_n^t \mid \alpha_n^{\ast, t}, \gamma_n^t) = \prod_{k=1}^K p(\alpha_{nk}^t \mid \alpha_{nk}^{\ast, t}, \gamma_k) = \prod_{k=1}^K I(\alpha_{nk}^{\ast, t} \in (\gamma_{k \alpha_{nk}^t}, \gamma_{k \alpha_{nk}^t + 1}])
\end{equation}

\noindent For all \(n \in [N]\) and all \(t \in \{2, 3, \ldots, T\}\), we assume \begin{equation}\label{prioralphastar2T}
\alpha_n^{\ast, t} \mid \lambda, \alpha_n^{t - 1}, \xi, R \sim N_K(X_n^t \lambda + d_{n, \text{otr}}^{t - 1} \xi, R)
\end{equation}

\noindent and for \(t = 1\), \begin{equation}\label{prioralphastar1}
\alpha_n^{\ast, t} \mid \lambda, R \sim N_K(X_n^t \lambda, R).
\end{equation}

\noindent These give (see Supplementary Material B) equations \eqref{mvprobit2T} and \eqref{mvprobit1}. We write \(\zeta = (\lambda^\prime, \xi^\prime)^\prime \in \mathbb{R}^{(D + H_{\text{otr}}) \times K}\) so \eqref{prioralphastar2T} and \eqref{prioralphastar1} can be written more simply as \begin{equation}
\alpha_n^{\ast, t} \mid \lambda, \alpha_n^{t - 1}, \xi, R \sim N_K(W_n^t \zeta, R)
\end{equation}

\noindent where \(W_n^t = (X_n^t, d_{n, \text{otr}}^{t - 1})\) for \(t \in \{2, 3, \ldots, T\}\) and \(W_n^t = (X_n^t, O)\) for \(t = 1\). Writing \(\alpha^{\ast, t} = ({\alpha_1^{\ast, t}}^\prime, \ldots, {\alpha_N^{\ast, t}}^\prime)^\prime\) and \(\alpha^\ast = ({\alpha^{\ast, 1}}^\prime, \ldots, {\alpha^{\ast, T}}^\prime)^\prime\), as well as \(W^t = ({W_1^t}^\prime, \ldots, {W_n^t}^\prime)^\prime\) and \(W = ({W^1}^\prime, \ldots, {W^T}^\prime)^\prime\)  gives \begin{equation}
  \alpha^\ast \mid \alpha^{1, \ldots, T-1}, \lambda, \xi, R \sim N_{TN, K}(W \zeta, I_{TN} \otimes R)
\end{equation}

\noindent (see Supplementary Material C).

We decompose a positive definite covariance matrix \(\Sigma\) as \(V^{1/2} R V^{1/2}\), where \(V = \text{diag}(v_1, \ldots, v_K) \in \mathbb{R}^{K \times K}\) and for all \(k \in [K]\), \(v_k > 0\). For \(p(R, V)\), we use a prior from \citet{wayman2025restricted}: \begin{equation}
  p(R, V) \propto (\det{R})^{-\frac{1}{2}(v_0 + K + 1)} \cdot \prod_{k \in [K]} \left[ \exp\left(-\frac{1}{2} v_k^{-1} A_{kk}\right) \cdot v_k^{-\frac{1}{2}(v_0 + 2)}\right]
\end{equation}

\noindent where for all \(k \in [K]\), \(A_{kk}\) are the diagonal elements of \(R^{-1}\) and \(v_0\) is a hyperparameter.

For each \(\gamma_k\), we use a prior introduced in \citet{wayman2025restricted} for latent variable models with a discrete latent state, namely, for each level \(l \in \{2, 3, \ldots, L - 1\}\), we assume \(\gamma_{kl} \ci \gamma_{k,l-2}, \gamma_{k, l-3}, \ldots, \gamma_{k3}, \gamma_{k2} \mid \gamma_{k, l-1}, v_k\) so that \begin{equation}
  p(\gamma_k \mid v_k) = p(\gamma_{k, L-1} \mid \gamma_{k, L-2}, v_k) \cdot p(\gamma_{k, L-2} \mid \gamma_{k, L-3}, v_k) \cdot \cdots \cdot p(\gamma_{k3} \mid \gamma_{k2}, v_k) \cdot p(\gamma_{k2} \mid v_k)
\end{equation}

\noindent and for each \(l \in \{2, 3, \ldots, L - 1\}\) we assume \(\gamma_{kl} \mid \gamma_{k,l-1}, v_k\) follows a left-truncated exponential whose rate parameter \(a\) is to be chosen such that the density is relatively flat: \begin{equation}
  p(\gamma_{kl} \mid \gamma_{k,l-1}, v_k) = I\left(\gamma_{kl} \in (\gamma_{k, l-1}, \infty)\right) \cdot a v_k^{1/2} \cdot \exp\left[-a v_k^{1/2} \cdot (\gamma_{kl} - \gamma_{k,l-1})\right].
\end{equation}

We let \(\lambda \mid R \sim N_{D, K}(0, I_D \otimes R)\), and \(\xi \mid R \sim N_{H_\text{otr}, K}(0, I_{H_\text{otr}} \otimes R)\), where \(N_{N, K}(A, B \otimes C)\) indicates a matrix variate normal distribution (see Supplementary Material D) with mean \(A\) and covariance matrix \(B \otimes C\). It follows that \(\zeta \mid R \sim N_{D + H_{\text{otr}}, K}(0, I_{D + H_{\text{otr}}} \otimes R)\).

\subsection{Integrability of the Model with Respect to the Variance Parameter}

\noindent By the directed local Markov property, \(\gamma \ci R \mid V\), so \(p(\gamma \mid V) \cdot p(R, V) = p(\gamma, R, V)\). If we use this in \eqref{origmodel}, we observe that \(V\) does not appear on the right-hand side of any conditional bars, so integration with respect to \(V\) is straightforward and gives \begin{align}
  p(Z \setminus V) & = p(Y \mid Y^\ast, \kappa) \cdot p(Y^\ast \mid \beta, \alpha) \cdot p(\kappa) \cdot p(\beta \mid \delta) \cdot p(\delta \mid \omega) \cdot p(\omega) \notag\\
  & \quad\quad \cdot \left[\prod_{t=2}^T p(\alpha^t \mid \alpha^{\ast, t}, \gamma) \cdot p(\alpha^{\ast, t} \mid \alpha^{t-1}, \lambda, \xi, R)\right] \cdot p(\alpha^1 \mid \alpha^{\ast, 1}, \gamma) \cdot p(\alpha^{\ast, 1} \mid \lambda, R) \notag\\
& \quad\quad \cdot p(\gamma, R) \cdot p(\lambda \mid R) \cdot p(\xi \mid R). \label{originalmodel}
\end{align}

\section{Model Identifiability}\label{sec-modelident}

\begin{definition}\label{identifdefgeneral}
For sets \(A, B\) and an equivalence relation \(\sim_E\) on \(A\), we say that \(A\) is identifiable from \(B\) up to \(\sim_E\) if and only if there exists a surjective function \(g: A \to B\) such that \begin{equation}
  \forall\, a, \widetilde{a} \in A \quad [g(a) = g(\widetilde{a}) \implies a \sim_E \widetilde{a}]
\end{equation}
\end{definition}

Often \(B\) is a family of density or likelihood functions and \(A\) is a parameter space. In this situation, for readability we will often write \(\{f(Y \mid \theta)\}\) to represent \(\{f_Y(\bullet \mid \theta) ; \theta \in \Theta\}\). We state the definitions of two types of identifiability specifically for this situation.

\begin{definition}\label{identifdefstrict}
  For a discrete random variable \(Y\) taking values on data space \(\mathcal{Y}\) and with \(\{f(Y \mid \theta)\}\) a family of likelihoods or densities for \(Y\), \(\Theta\) is strictly identifiable from \(\{f(Y \mid \theta)\}\) up to an equivalence relation \(\sim_E\) if and only if \begin{equation}
  \forall\, \theta, \widetilde{\theta} \in \Theta \quad [\forall y \in \mathcal{Y} \quad f_Y(y \mid \theta) = f_Y(y \mid \widetilde{\theta})] \implies \theta \sim_E \widetilde{\theta}.
\end{equation}
\end{definition}

\begin{definition}\label{identifdefgeneric}
Let \(\mathcal{N}_{\Lambda}\) denote the family of all \(\Lambda\)-null sets on the parameter space \(\Theta\), where \(\Lambda\) denotes the Lebesgue measure. For a discrete random variable \(Y\) taking values on data space \(\mathcal{Y}\) and with \(\{f(Y \mid \theta)\}\) a family of likelihoods or densities for \(Y\), \(\Theta\) is generically identifiable from \(\{f(Y \mid \theta)\}\) up to an equivalence relation \(\sim_E\) if and only if \begin{equation}
  \exists\, N \in \mathcal{N}_{\Lambda} \quad \forall\, \theta, \widetilde{\theta} \in \Theta \setminus N \quad [\forall y \in \mathcal{Y} \quad f_Y(y \mid \theta) = f_Y(y \mid \widetilde{\theta})] \implies \theta \sim_E \widetilde{\theta}.
\end{equation}
\end{definition}

\noindent Note that in Definitions \ref{identifdefstrict} and \ref{identifdefgeneric}, the surjective function of Definition \ref{identifdefgeneral} is the parameterization map.

When we say a parametric model for a variable \(Y\) and parameter space \(\Theta\) is strictly (generically) identifiable up to an equivalence relation, we mean that \(\Theta\) is strictly (generically) identifiable from \(\{f(Y \mid \theta)\}\) up to an equivalence relation. Also when applying either of the definitions, if the specific equivalence relation is not mentioned, it should be assumed that the relation is the equality relation. 

We will prove that if certain conditions hold, \(\Theta = (\Theta_s, \Theta_m)\) is strictly identifiable up to label swapping from \(\{f(Y \mid \theta)\}\), and if a certain subset of those conditions holds, \(\Theta\) is generically identifiable up to label swapping from \(\{f(Y \mid \theta)\}\). The label swapping equivalence relation on \(\Theta\) referred to here is the permutation of the rows and/or columns of all relevant vector or matrix parameters in \(\Theta\) that correspond to a permutation of the dimensions of the latent state vector (for example, if the labels of \(\alpha_n^t\) are permuted, the rows of \(\beta\) would have to be permuted as well since each row of \(\beta\) corresponds to a particular interaction effect of the dimensions of \(\alpha_n^t\), which have been permuted).

The formal statement of this result is as follows: \begin{theorem}
\label{identiffinalresult}
Define the following conditions: 
\begin{enumerate}
\item[(C1)] For all \(n \in [N]\) and \(t \in [T]\), all elements of \(\pi_n^t\) are greater than zero (this implies that \(\text{rank}(\text{diag}(\pi_n^t)) = L^K\))
\item[(C2)] \(\sum_{j=1}^J M_j \geq L^K\) (the total number of item responses is greater than or equal to the number of possible latent states)
\item[(C3)] there exist subsets \(J_1, J_2, J_3\) of items such that the ranks of the three resulting block matrices comprising the emissions matrix \(B\) are all of rank \(L^K\) (this implies that \(\text{rank}(B) = L^K\), i.e. \(B\) has full column rank)
\item[(C4)] for all \(n \in [N]\) and \(t \in \{2, 3, \ldots, T\}\), \(\text{rank}(U_{n, t, t-1}) = L^K\) (i.e. \(U_{n, t, t-1}\) has full rank)
\item[(C5)] \(N \geq D + H_{\text{otr}}\)
\item[(C6)] for all \(t \in [T]\), \(\text{rank}(W^t) = D + H_{\text{otr}}\) (i.e. \(W^t\) has full column rank)
\end{enumerate}

\noindent In addition, define two conditions on the \(\delta\) matrix:
\begin{enumerate}
\item[(D1)] \(\delta\) has for each attribute an active main effect on all levels of the attribute for at least two items 
\item[(D2)] \(\delta\) has no interaction effects active
\end{enumerate}

\noindent If conditions (C1) through (C6) and (D1) hold, then \(\{p(Y \mid \theta)\}\) is generically identifiable up to label swapping, and if conditions (C1) through (C6) and (D1) and (D2) hold, then \(\{p(Y \mid \theta)\}\) is strictly identifiable up to label swapping.
\end{theorem}

See Supplementary Material E for the proof. We note that as part of the proof of Theorem \ref{identiffinalresult}, we prove that the single time-point multivariate probit model is strictly identifiable.

\section{Parameter Expansion and Algorithm}

In order to produce a model from which we can easily sample using a multiple-block Metropolis-Hastings algorithm \citep{chib2011introduction}, we perform a transformation of \(Z\) to \(\widetilde{Z}\); our algorithm will sample from the model \(p_{\widetilde{Z}}(\widetilde{z}) = p_Z(g^{-1}(\widetilde{z})) \cdot |\det{J_{g^{-1}}(\widetilde{z})}|\). This methodology fits into the category of ``parameter expansion'' \citep{liu1999parameter} or ``conditional augmentation'' \citep{meng1999seeking}.

We transform \(Z\) to \(\widetilde{Z}\), where \(\widetilde{\alpha^\ast} = \alpha^\ast V^{1/2}\), \(\widetilde{\gamma} = \gamma V^{1/2}\), \(\widetilde{\zeta} = \zeta V^{1/2}\), and \(\Sigma = V^{1/2} R V^{1/2}\), similar to the cross-sectional model we are extending \citep{wayman2025restricted} (see Supplementary Material F for details). Our Metropolis-within-Gibbs algorithm samples from \(p(\widetilde{Z})\), and transforms each sampled value using the inverse of the transformation to produce a sample from the original model. The sampling steps are shown in Supplementary Material G.

\section{Software}

A Python package \citep{wayman2025probitlcmlongit} containing an implementation of the procedures for simulation, data analysis, and model selection described in this manuscript has been released under a FLO (free/libre/open) license and is available at \url{https://github.com/ericwayman01/probitlcmlongit}. We note that the sampling algorithm is mostly implemented using the Armadillo C++ library \citep{sanderson2025armadillo,sanderson2019practical}. The data generating parameters for the simulations and relevant files for the data analyses are available at \url{https://github.com/ericwayman01/probitlcmlongit_related}. For the data sets that were used for the data analyses, please refer to the Data Availability Statement.

\section{Simulation Studies}

To investigate the efficacy of the model in a variety of scenarios, we performed two simulation studies: simulation study one has larger sample sizes (\(N = 250, 500, 1500, 3000\)) and a smaller number of time points (\(T = 3\)), and simulation study two uses smaller sample sizes (\(N = 125, 250, 500\)) and a larger number of time points (\(T = 30\)). Simulation study two also includes an examination of the performance of the model when a percentage of time points of data are missing. These choices for the two simulation studies roughly correspond to the two data applications we present in Section \ref{sec:applications}.

The simulations are set up as follows. For a given combination of \(K\) and \(L\), we create items in sets of five such that each set satisfies one of the following conditions: (1) for each dimension of the latent state, a set of items is related to that dimension alone, and (2) for each pair of dimensions of the latent state, a set of items is related to the pair of dimensions. For \(K = 2\) this gives 15 items, for \(K = 3\), 25 items, and for \(K = 4\), 45 items. The \(\delta\) and \(\beta\) matrices are chosen such that these relationships hold.

In order for the algorithm to be able to classify respondents correctly in terms of their latent class values, there need to be appreciably different patterns of item response class-conditional probabilities across classes: values of \(\kappa\) are generated such that this is the case. Our simulation contains two covariates, age and sex, and we generate \(\lambda\) such that there is some contribution of both covariates to the value of \(\alpha_n^t\). \(\xi\) is chosen such that each dimension of \(\alpha_n^{t-1}\) affects that same dimension of \(\alpha_n^t\), but no dimensions act in any combination.

We have 15 combinations of values of \(K, L\) and \(\rho\) (where \(\rho\) is the value of the off-diagonal elements of the correlation matrix \(R\)), and thus 15 sets of parameters as described above. We evaluate how the model performs in terms of parameter recovery for multiple sample sizes: we generate data from the model, run the model, obtain parameter estimates, and compare these estimates to the data-generating parameter values (this last step of the procedure is described in a later paragraph).

In the simulations, we set hyper-parameters to the following values for all scenarios: \(\sigma_{\beta}^2 = 2.0\), \(\omega_0 = 0.5\), and \(\omega_1 = 0.5\), and \(a = 1000^{-1}\). We set \(v_0 = K + 1\) so that we have uniform priors for the correlations \citep{barnard2000modeling,gelman2007data}. We tuned \(\sigma_{\kappa}^2\) for each sample size so that its acceptance rate was roughly 40\%. For all scenarios, the order of the measurement model was set to 2, and the order of the transition model was set to 1. For both simulation studies, for each replication we use a burn-in period of 6,000 draws and a post burn-in phase of 10,000 draws.

To evaluate convergence of the model, for all scenarios and all replications we ran the Geweke test \citep{geweke1992evaluating} on each element of each matrix parameter. For every scenario, it is the case that for every element, fewer than 2\% of the replications yield a test statistic that falls outside the 95\% confidence interval, and for most parameters the percentage is either zero or close to it. In addition, we examined integrated autocorrelation time, or IAT  \citep{liu2008monte,plummer2006coda,mcgibbon2019pyhmc} for a representative replication for each scenario. We found that most matrix and vector parameters had an average IAT of less than 10 (corresponding to effective sample sizes of greater than 1,000); for some scenarios \(\beta\) had a higher IAT. We also examined trace plots for various elements of matrix parameters to confirm convergence visually as we settled on a sufficient chain length.

In simulation study two, initial values for missing data rows were set as follows. In this simulation, each respondent \(n\) has a known vector of time points for which \(Y_n^t\) is missing, which is denoted \((t_n^1, \ldots, t_n^i)\), \(i \in [T]\). For each respondent \(n\), if \(t_n^1 = 1\) we use the first value of data that is not missing which follows \(t_n^1\). Then, for any \(t > 1\) for respondent \(n\),  we proceed sequentially and for each \(t\) use the first value of data that is not missing which precedes \(t\).

We evaluate parameter recovery on an element-wise basis. In addition to the parameters of the model, we report recovery of a additional ``parameter,'' \(\eta\), which is the set of class-conditional item response probabilities. For a given parameter (e.g. \(\beta\)), denote an element (e.g. \(\beta_{hj}\) for some \(h\) and \(j\)) for the moment by \(\theta\), and denote by \(\theta^{(s, r)}\) the \(s\)th draw of \(\theta\) from the Markov chain for the \(r\)th replication in the post-burn-in phase. Letting \(S\) be the number of draws in the post-burn-in phase, for elements of all parameters except \(\delta\), we use the estimate of the posterior mean \(\widehat{\theta}^{(r)} = (1 / S) \sum_{s=1}^S \theta^{(s, r)}\); for elements of \(\delta\), we use the estimate of the posterior mode \(\widehat{\theta}^{(r)} = I\left((1/S) \sum_{s=1}^S \theta^{(s, r)} > 0.5\right)\). To measure how well parameters were recovered, for each replication we calculated a recovery metric: for elements of \(\gamma, \eta, R, \lambda, \xi\) and \(\beta\), we calculate the absolute error of estimation for each element for replication \(r\), namely for a scalar \(\theta\), \(\text{AE}_r(\theta, \widehat{\theta}^{r}) = |\theta - \widehat{\theta}^{r}|\). We then take the average of absolute error across all replications to arrive at mean absolute error (MAE) for that element. For elements of \(\delta\), we use correctness of estimation, namely for a scalar \(\theta\), \(\text{CE}_r(\theta, \widehat{\theta}^{r}) = I(\theta = \widehat{\theta}^{(r)})\). We then take the average of correctness of estimation across all replications to arrive at recovery accuracy for that element. Averaging these values across all elements of a matrix or vector parameter gives us respectively the average mean absolute error and average recovery accuracy for the matrix or vector parameter, which are the values reported.

The results of the simulation studies are reported in Supplementary Material H and are summarized here. Recovery of \(\gamma\), \(\eta\), \(R\), \(\lambda\), \(\xi\), \(\beta\), and \(\delta\) are evaluated. We also report recovery metrics for inactive and active coefficients of \(\beta\) and the corresponding \(\delta\) elements: for \(i \in {0, 1}\), average recovery accuracy across all \(\delta_{hj}\) for which \(\delta_{hj} = i\) is reported under the header \(\delta^i\) and average MAE of the corresponding elements of \(\beta\) is reported under \(\beta^i\). We also report the per-replication run time in minutes for each scenario for both simulation studies. The timing was performed using 4 cores of an 11th Gen Intel(R) Core(TM) i7-1160G7 processor (released in the year 2020) on a machine with 16 GB of RAM.

We observe that for most combinations of \(J, K, L\) and \(\rho\), as sample size increases parameter recovery improves.

\section{Applications}\label{sec:applications}

\subsection{Education Application}

\begin{table}
\small
\begin{center}
\caption{\label{tab:tangattributes} Attributes representing knowledge of rational number manipulations}
\vspace{0.5\baselineskip}
\begin{tabular}{rr}
\toprule
Attribute number & Attribute description \\
\midrule
1 & Rational numbers \\
2 & Related concepts of rational numbers \\
3 & Axis \\
4 & Addition and subtraction of rational numbers \\
5 & Multiplication and division of rational numbers \\
6 & Mixed operation of rational numbers \\
\bottomrule
\end{tabular}
\end{center}
\end{table}

We apply the model to the dataset \citep{zhan2021data} used in \citet{tang2021does}, a study which aimed to measure the effectiveness of two types of feedback for math test takers: CDF (cognitive diagnostic feedback) and CIRF (correct-incorrect response feedback). The test had 18 items, 12 of which were multiple-choice and 6 of which were calculations. The items were designed to diagnose whether or not students had mastered six latent attributes related to rational number operations \citep{tang2020development}. The six latent binary attributes are displayed in Table \ref{tab:tangattributes}. We note that in their study, \citet{tang2021does} designed a \(Q\)-matrix which reflects their assumptions of which latent attributes (representing mathematics skills) would need to be mastered in order to score correctly on each of the items. As noted in the Introduction, the particular definition of the \(Q\)-matrix depends on the model being specified. In the case of \citet{tang2021does}, their sLong-DINA model has no interaction terms, so a 1 for a particular attribute-item pair in the \(Q\)-matrix indicates that the attribute can enter into the equation for the latent state-conditional response probability for that item.

The dataset consists of item response data for 276 respondents. The respondents are grouped into almost equal size groupings: the diagnosis group, the traditional group, and the control group. Respondents took a math test three times, and thus the item response data consists of binary values indicating a correct or incorrect answer for each of the 18 items observed at three time points. The protocol was as follows: all respondents took the test once, and 24 hours afterward, CDF (cognitive diagnostic feedback) was provided to the diagnosis group and CIRF (correct-incorrect response feedback) was provided to the traditional group. The control group received no feedback. One week after the first test, the respondents took the test a second time, and feedback was once again provided to the different groups as above. Finally, all respondents took the test a third time after one week had passed.

We fit our longitudinal model, which is exploratory, to this dataset for values of \(K\) ranging from 2 through 6, with a measurement model order of 2 (i.e. including one-way main effects and two-way interactions) for interpretability. We also fit a confirmatory version of our model which uses a fixed \(\delta\) matrix corresponding to the fixed \(Q\)-matrix of \citet{tang2021does} rather than estimating \(\delta\) from the data. Dummy variables indicating the three groupings of respondents were used as covariates. We set hyperparameter values to \(\sigma_{\beta}^2 = 2.0\), \(\omega_0 = 0.5\), and \(\omega_1 = 0.5\), \(a = 1000^{-1}\), and \(v_0 = K + 1\). We performed hyperparameter tuning on \(\sigma_{\kappa}^2\) and chose its value such that its acceptance rate is roughly 40\%. We specify the order of the transition model to be 1 (only main effects). We ran the model using a burn-in period of 10,000 draws and a post-burn-in period of 20,000 draws. We consider two diagnostics for convergence, the Geweke test and the IAT. We observe that fewer than 2\% of the Geweke test statistics fall outside the 95\% confidence interval, and that on average for each parameter the effective sample size (number of draws divided by IAT) is greater than 200.

We utilize the WAIC \citep{watanabe2010asymptotic} to evaluate the choice of model (i.e. the choice of various possible values of \(K\) and \(L\)). The WAIC is an estimate of the generalization loss of a Bayesian model; the smaller is generalization loss, the smaller is the Kullback-Leibler distance from the true distribution to the posterior predictive distribution obtained from the selected model and the observed data. Here we use class-conditional values of the likelihood to calculate the WAIC \citep{merkle2019bayesian}. We store values of the conditional likelihood \begin{equation}
  p(y_n \mid \theta, \alpha_n) = \prod_{t=1}^T \prod_{j=1}^J \left[\Phi(\kappa_{y_{nj}^t + 1} - d_n^t \beta_j) - \Phi(\kappa_{y_{nj}^t} - d_n^t \beta_j)\right]
\end{equation}

\noindent evaluated at each \((\theta^{(s)}, \alpha^{(s)})\) in the post-burn-in phase of our sampling, using a thinning interval of 10 (see Supplementary Material I for derivation). These values were evaluated using the \texttt{waic} function of the \texttt{R} package \texttt{loo} \citep{vehtari2024loo}. The resulting WAIC values are displayed in Table \ref{tab:modelselresults}.

\begin{table}
\small
\begin{center}
\caption{\label{tab:modelselresults} Class-conditional WAIC for models, education application}
\vspace{0.5\baselineskip}
\begin{tabular}{rrrr}
\toprule
\(K\) & \(L\) & Notes & WAIC \\
\midrule
2 & 2 & & 12,425.73 \\
3 & 2 & & 11,576.91  \\
4 & 2 & & 10,984.37 \\
5 & 2 & & 10,269.64 \\
6 & 2 & & 9,523.10 \\
6 & 2 & Tang and Zhan delta & 11,437.60 \\
\bottomrule
\end{tabular}
\end{center}
\end{table}

We observe that our model, with its exploratory \(\delta\) matrix, has a better fit than the model of \citet{tang2021does} for all values of \(K\) greater than 3. We consider the parameter estimates for the \(K = 6\) case since \citet{tang2021does} assumed six attributes for their model and \(K = 6\) had the lowest WAIC value, indicating the best fit. The run time for this model with the above chain length was 12.194 minutes, utilizing 4 cores of an 11th Gen Intel(R) Core(TM) i7-1160G7 processor on a machine with 16 GB of RAM.

\begin{table}
\small
\begin{center}
\caption{\label{tab:applicbeta1} Measurement model main effects, education application}
\vspace{0.5\baselineskip}
\begin{tabular}{lcccccc}
\toprule
& \multicolumn{6}{c}{Effect} \\
Item & Intercept & \(\beta_6\) & \(\beta_5\) & \(\beta_4\) & \(\beta_3\) & \(\beta_2\) \\
\midrule
1 & -0.41 & 1.51 &  & 2.15 &  & 1.39 \\
2 & -0.57 & 1.49 &  &  &  & 3.10     \\
3 & -1.37 &  &  &  &  & 1.51         \\
4 & -0.65 &  &  &  &  &              \\
5 & -1.27 &  & 0.98 &  &  &          \\
6 & -1.43 &  &  &  & 0.95 &          \\
7 & -2.07 & 1.70 &  & 1.38 &  &      \\
8 & -1.14 & 1.40 &  &  &  & 1.50     \\
9 & -1.57 &  & 1.74 &  &  &          \\
10 & -1.69 &  &  & 1.84 &  &         \\
11 & -1.56 &  & 1.61 &  &  &         \\
12 & -1.29 &  &  &  &  &             \\
13 & -2.21 & 2.32 & 3.19 &  &  &     \\
14 & -2.14 &  & 2.82 &  & 3.44 &     \\
15 & -2.73 & 1.73 & 1.29 &  & 1.40 & \\
16 & -2.93 & 1.94 &  &  &  &         \\
17 & -2.40 & 1.54 &  &  &  &         \\
18 & -3.31 & 1.83 &  &  &  &         \\
\bottomrule
\end{tabular}
\end{center}
\end{table}

\begin{table}
\small
\begin{center}
\caption{\label{tab:applicbeta2} Measurement model interaction effects, education application}
\vspace{0.5\baselineskip}
\begin{tabular}{lccccccccccccc}
\toprule
& \multicolumn{13}{c}{Effect} \\
Item & \(\beta_{5, 6}\) & \(\beta_{4, 6}\) & \(\beta_{3, 6}\) & \(\beta_{3, 5}\) & \(\beta_{3, 4}\) & \(\beta_{2, 6}\) & \(\beta_{2, 5}\) & \(\beta_{2, 4}\) & \(\beta_{2, 3}\) & \(\beta_{1, 6}\) & \(\beta_{1, 5}\) & \(\beta_{1, 4}\) & \(\beta_{1, 3}\) \\
\midrule
1 &  &  &  &  &  &  &  &  &  &  &  &  \\
2 &  &  &  &  &  &  &  &  &  &  &  &  \\
3 & 1.80 & & 1.51 &  &  &  & 1.89 &  &  &  &  &  \\
4 &  & 2.34 &  &  &  &  &  &  &  &  &  &  \\
5 &  &  &  & 0.77 &  &  &  &  & 2.03 &  & 2.60 &  \\
6 &  &  &  & 0.75 &  &  &  &  &  &  &  &  \\
7 &  &  & 1.83 &  &  & 2.61 &  &  &  &  &  &  & 2.48 \\
8 &  &  &  &  &  &  &  &  &  &  & 2.97 &  & 1.95 \\
9 &  & 1.22 &  &  &  &  &  &  & 1.62 &  &  & 2.44 \\
10 &  &  &  &  &  &  &  &  &  &  &  &  \\
11 & 1.22 &  &  &  &  &  &  &  &  &  &  &  & 1.20 \\
12 &  &  &  &  &  &  &  &  & 1.76 & 1.96 &  &  \\
13 &  &  &  &  &  &  &  & 1.68 &  &  &  &  \\
14 &  &  &  &  &  &  &  &  & 3.83 &  &  &  \\
15 &  &  &  &  &  &  &  &  & 2.60 &  &  &  \\
16 &  &  &  &  & 1.27 &  &  &  &  &  & 1.20 &  \\
17 &  &  &  &  &  &  &  &  &  & 1.67 &  &  \\
18 &  &  &  &  &  &  & 0.89 &  &  &  &  &  \\
\bottomrule
\end{tabular}
\end{center}
\end{table}

\begin{table}
\small
\begin{center}
\caption{\label{tab:modelqmatrix} Model-implied \(Q\)-matrix, education application}
\vspace{0.5\baselineskip}
\begin{tabular}{lcccccccccccccccccc}
\toprule
 & \multicolumn{18}{c}{Item} \\
Attribute & 1 & 2 & 3 & 4 & 5 & 6 & 7 & 8 & 9 & 10 & 11 & 12 & 13 & 14 & 15 & 16 & 17 & 18 \\
\midrule
1 &   &   &   &   & 1 &   & 1 & 1 & 1 &    &  1 &  1 &    &    &    &  1 &  1 &    \\
2 & 1 & 1 & 1 &   & 1 &   & 1 & 1 & 1 &    &    &  1 &  1 &  1 &  1 &    &    &  1 \\
3 &   &   & 1 &   & 1 & 1 & 1 & 1 & 1 &    &  1 &  1 &    &  1 &  1 &  1 &    &    \\
4 & 1 &   &   & 1 &   &   & 1 &   & 1 &  1 &    &    &  1 &    &    &  1 &    &    \\
5 &   &   & 1 &   & 1 & 1 &   & 1 & 1 &    &  1 &    &  1 &  1 &  1 &  1 &    &  1 \\
6 & 1 & 1 & 1 & 1 &   &   & 1 & 1 & 1 &    &  1 &  1 &  1 &    &  1 &  1 &  1 &  1 \\
\bottomrule
\end{tabular}
\end{center}
\end{table}

The estimate of the \(\beta\) matrix (the average of draws of \(\beta\)) is shown in Table \ref{tab:applicbeta1} (which shows the main effects) and Table  \ref{tab:applicbeta2} (which shows the interaction effects). \(\beta_k\) indicates the the main effects for attribute \(k\), and \(\beta_{k,l}\) indicates the effects of the interaction of attributes \(k\) and \(l\). We have applied a sparsity criterion, namely that any element of \(\beta\) for which 0 falls into the 95\% equal-tail credible interval is deemed inactive. Attributes which had no significant effects were excluded from the tables. Attribute 1 has no significant main effects, and all other attributes load onto a different combination of items. Items 4 and 12 correspond to no main effects. Interaction effects are present for all items other than 1, 2 and 10. We observe a significant amount of sparsity for main effects, where most pairs of attributes are related to one, two or three items through interaction effects (one pair of attributes, attributes 1 and 6, have an interaction effect which loads on five items).

We examine the \(Q\)-matrix implied by the significant measurement model coefficients, which is shown in Table \ref{tab:modelqmatrix}. This \(Q\)-matrix is significantly denser than the one hypothesized by Tang and Zhan: it shows many items being related to at least three attributes and some related to more, with only one item being related to a single attribute.

\begin{table}
\small
\begin{center}
\caption{\label{tab:appliclambda} \(\lambda\) coefficients estimates, education application}
\vspace{0.5\baselineskip}
\begin{tabular}{lrrr}
\toprule
Attribute & Intercept & Diagnosis & Traditional \\
\midrule
1 & 0.94  & -0.66  & -0.71 \\
2 & -0.25 & 0.30   & 0.28   \\
3 & -1.83 & 0.66   & 0.50   \\
4 & -1.04 & 0.37   &   \\
5 &       & 0.29   &   \\
6 & -0.53 & 0.32   &   \\
\bottomrule
\end{tabular}
\end{center}
\end{table}

The estimates of the slope coefficient relating the latent state to covariates, \(\lambda\), is shown in Table \ref{tab:appliclambda}. We see that the diagnosis intervention (CDF) has a positive effect on every attribute except attribute 1, for which the effect was negative. We see that the traditional intervention (CIRF) has positive effects on attributes 2 and 3 and a negative effect on attribute 1 (attributes 4 through 6 were not significant). These results correspond to the conclusions of \citet{tang2021does}, which are that CDF is more effective than no feedback, and CDF is more effective than CIRF. For the 95\% equal-tail credible intervals for each coefficient of \(\lambda\), see Supplementary Material J.

\begin{table}
\small
\begin{center}
\caption{\label{tab:applicxi} \(\xi\) coefficients estimates, education application}
\vspace{0.5\baselineskip}
\begin{tabular}{lcccccc}
\toprule
& \multicolumn{6}{c}{Attributes at time \(t\)} \\
Attr. at \(t - 1\) & 1 & 2 & 3 & 4 & 5 & 6 \\
\midrule
1 & 0.79  &  &  &  & 0.67   & 0.39 \\
2 &       & 1.31  &  &  & 0.84   & \\
3 &       &  & 1.22  & &   & 0.44 \\
4 &       & &  & 1.12  & 0.46   & 0.49 \\
5 & -0.49 &  & 0.61  & 0.41  & 2.53 & 1.13 \\
6 & -0.51 &  &  & 0.46  & 1.38 & 2.13 \\
Intercept & -0.57 &  & 0.77  &  & -2.63 & -1.89 \\
\bottomrule
\end{tabular}
\end{center}
\end{table}

The matrix of coefficients \(\xi\) indicates how a respondent's attributes at a time point are affected by the respondent's attributes at the previous time point. The estimate of \(\xi\) is shown in Table \ref{tab:applicxi}. We see that there is a large positive number for each entry of the main diagonal of the table, which shows that if attribute \(k\) is present at time \(t - 1\), it is likely to be present at \(t\) as well (in the context of this application, this means that skills once mastered are maintained across time). Most relationships between attributes between sequential time points are positive. For the 95\% equal-tail credible intervals for each coefficient of \(\xi\), see Supplementary Material J.

\begin{table}
\small
\begin{center}
\caption{\label{tab:applicRmat} \(R\) (underlying correlation matrix) estimate, education application}
\vspace{0.5\baselineskip}
\begin{tabular}{lrrrrrr}
\toprule
 & 1 & 2 & 3 & 4 & 5 & 6 \\
\midrule
1 & 1.00 & -0.44 & 0.24 & -0.40 & -0.41 & -0.29 \\
2 & -0.44 & 1.00 & -0.08 & -0.08 & 0.41 & -0.50 \\
3 & 0.24 & -0.08 & 1.00 & -0.44 & -0.25 & -0.22 \\
4 & -0.40 & -0.08 & -0.44 & 1.00 & -0.03 & 0.24 \\
5 & -0.41 & 0.41 & -0.25 & -0.03 & 1.00 & 0.15 \\
6 & -0.29 & -0.50 & -0.22 & 0.24 & 0.15 & 1.00 \\
\bottomrule
\end{tabular}
\end{center}
\end{table}

Table \ref{tab:applicRmat} shows the estimates of the underlying correlation matrix relating the six attributes of the latent state. We observe almost no correlation between attribute pairs (2, 3), (2, 4), and (4, 5). We observe several correlations close to negative and positive 0.5.

\subsection{Emotional State Application}

We apply the model to response data \citep{shui2020datasetsynapse} collected from 140 respondents over a period of five days \citep{shui2021dataset} (two respondents from the original dataset with entire days of data missing were excluded). Data were collected by a device distributed to all participants which sent a request to the participant at six varying time points per day with a minimal interval of 90 minutes between requests \citep{shui2021dataset}. The response data was missing some time points for some respondents; we assume that these data points were missing completely at random \citep{marini1980maximum,little2021missing}. The missing data vectors are treated as a parameter with the same conditional independence and dependence assumptions in the graphical model as \(Y\). We sample from the augmented posterior: the sampling algorithm remains the same with one additional step of sampling the various \(Y_n^t\) as follows. For each respondent \(n\), the known vector of time points for which \(Y_n^t\) is missing is denoted \((t_n^1, \ldots, t_n^i)\), \(i \in [T]\). For respondent \(n\), for each time point \(t \in (t_n^1, \ldots, t_n^i)\) of missing data, the conditional of \(Y_{nj}^t\) collapsed on \(Y_{nj}^{t, \ast}\) is a categorical distribution with \begin{equation}\label{3-missingdatamodel}
  p(Y_{nj}^t = m \mid \alpha_n^t, \beta_j, \kappa_j) = \Phi\left(\kappa_{j, m + 1} - d_n^t \beta_{j}\right) - \Phi\left(\kappa_{j, m} - d_n^t \beta_{j}\right).
\end{equation}

\noindent Details on how time points for missingness and data initialization were handled are in Supplementary Material K.

The response data consisted of seventeen items, listed in the Item column of Table \ref{3-tab:itemdesc}: five items described as the TIPI-C, or Ten-Item Personality Inventory in China \citep{shui2021dataset,shui2020datasetsynapse,lu2020disentangling}, the ten-item PANAS, or Positive and Negative Affect Schedule \citep{watson1988development}, emotional valence, and emotional arousal. The TIPIC-C items range from 0 through 6, the PANAS items range from 0 through 4, and the valence and arousal items range from 0 through 4.

\begin{table}
\small
\begin{center}
\caption{\label{3-tab:itemdesc} Description of items used for response data, emotional state application}
\vspace{0.5\baselineskip}
\begin{tabular}{lll}
\toprule
Item & Description\\
\midrule
TIPIC-C 1 & extraversion (outgoing/energetic to solitary/reserved)\\
TIPIC-C 2 & agreeableness (friendly/compassionate to challenging/callous)\\
TIPIC-C 3 & conscientiousness (efficient/organized to extravagant/careless)\\
TIPIC-C 4 & openness to experience (inventive/curious to consistent/cautious)\\
TIPIC-C 5 & emotional stability (calm, stable to anxious)\\
PANAS 1 & upset\\
PANAS 2 & hostile\\
PANAS 3 & alert\\
PANAS 4 & ashamed\\
PANAS 5 & inspired\\
PANAS 6 & nervous\\
PANAS 7 & determined\\
PANAS 8 & attentive\\
PANAS 9 & afraid\\
PANAS 10 & active\\
Emotional valence & extremely negative to extremely positive\\
Emotional arousal & extremely calm to extremely excited\\
\bottomrule
\end{tabular}
\end{center}
\end{table}

We used four covariates for the analysis. The first two were dummy variables indicating time range of measurement, namely afternoon (12:30 - 18:29) and evening (18:30 - 23:59) as opposed to a baseline of morning (07:00 - 12:29). The other two are from the pre-test measurements, specifically from the Meaning of Life Questionnaire, or MLQ \citep{steger2006meaning}: we computed the two subscales for presence and search and use the \(z\)-scores of these as our other two covariates.

We fit five different models of increasing complexity, four of which are displayed in Table \ref{3-tab:waic} along with class-conditional WAIC values calculated as they were in the education application. For one of the five models, \(K = 4, L = 2\), we observed near collinearity between two latent attributes so we excluded this model from our consideration, and took this as evidence that the latent space has under four dimensions. We selected the model with the lowest WAIC value, namely \(K = 3, L = 3\). The chains had a burn-in period of 20,000 draws and a post-burn-in phase of 20,000 draws. The run time for this model and chain length was 47.880 minutes, utilizing 4 cores of an 11th Gen Intel(R) Core(TM) i7-1160G7 processor on a machine with 16 GB of RAM. As in the previous application, we use both the Geweke test and the IAT to evaluate convergence. The Geweke test statistic for every single one of the parameters falls within the 95\% acceptance region; on average for each parameter the effective sample size is always greater than 100.

\begin{table}
\small
\begin{center}
\caption{\label{3-tab:waic} Class-conditional WAIC for models, emotional state application}
\vspace{0.5\baselineskip}
\begin{tabular}{lll}
\toprule
  K & L & WAIC \\
\midrule
2 & 2 & 160,803.7\\
3 & 2 & 155,267.4\\
2 & 3 & 154,092.4\\
3 & 3 & 149,782.1\\
\bottomrule
\end{tabular}
\end{center}
\end{table}

\begin{table}
\small
\begin{center}
\caption{\label{3-tab:betacoeffs1} \(\beta\) coefficients estimates, main effects, emotional state application}
\vspace{0.5\baselineskip}
\begin{tabular}{lrrrrrrr}
\toprule
& \multicolumn{6}{c}{Attribute} \\
& [0 0 0] & [0 0 1] & [0 0 2] & [0 1 0] & [0 2 0] & [1 0 0] & [2 0 0] \\
\midrule
extraversion & 0.63 &  &  &  &  & 0.73 & 0.99 \\
agreeableness & 1.66 &  &  & 0.93 & 1.47 & 0.48 & 0.68 \\
conscientiousness & 1.00 &  &  & 0.94 &  & 0.95 &  \\
openness & 0.72 &  &  & 0.70 &  & 0.92 & 0.59 \\
stability & 1.24 &  &  & 1.17 & 1.41 & 0.70 & 0.25 \\
upset & -1.02 & 1.97 & 1.58 &  &  &  &  \\
hostile & -1.90 & 1.94 & 1.13 &  &  &  &  \\
alert & -1.27 & 1.42 & 1.33 &  &  &  &  \\
ashamed & -1.22 & 1.64 & 0.94 &  &  &  &  \\
inspired & -1.03 & 0.94 & 0.43 & 0.64 &  & 0.94 & 1.16 \\
nervous & -1.10 & 1.82 & 1.50 &  &  &  &  \\
determined & -0.60 & 0.72 & 0.64 & 0.79 &  & 1.14 & 0.76 \\
attentive & -0.26 & 0.52 & 0.57 & 1.11 &  & 1.34 & 0.62 \\
afraid & -1.50 & 1.93 & 1.81 &  &  &  &  \\
active & -0.54 & 0.82 & 0.68 & 0.37 & 0.89 & 1.49 & 1.69 \\
valence & 0.96 &  &  & 1.29 & 1.14 & 1.31 & 1.08 \\
arousal & 0.68 &  & 0.68 &  &  &  & 0.68 \\
\bottomrule
\end{tabular}
\end{center}
\end{table}

\begin{table}
\small
\begin{center}
\caption{\label{3-tab:betacoeffs2} \(\beta\) coefficients estimates, interaction effects, emotional state application}
\vspace{0.5\baselineskip}
\begin{tabular}{lrrrrrrr}
\toprule
& \multicolumn{7}{c}{Attribute} \\
& [0 1 2] & [0 2 1] & [0 2 2] & [1 1 0] & [2 0 1] & [2 1 0] & [2 2 0] \\
\midrule
extraversion & 0.74 & 1.71 &  &  &  & 0.74 & 1.33 \\
aggreeableness &  &  &  &  &  &  &  \\
conscientiousness &  & 1.52 &  &  &  &  & 0.92 \\
openness &  &  &  &  &  & 0.48 &  \\
stability &  &  &  &  &  &  &  \\
upset &  &  &  &  &  &  &  \\
hostile &  &  &  &  &  &  &  \\
alert &  &  &  &  &  &  &  \\
ashamed &  &  &  &  &  &  &  \\
inspired &  &  &  & 0.67 &  &  &  \\
nervous &  &  &  & 0.32 &  &  &  \\
determined &  &  &  & 0.60 &  &  & 0.91 \\
attentive &  &  &  &  &  &  &  \\
afraid &  &  & 1.75 &  &  &  &  \\
active &  & 0.84 &  &  &  & 0.67 &  \\
valence &  &  &  &  & 0.61 &  &  \\
arousal &  & 0.94 &  &  &  & 0.70 & 1.15 \\
\bottomrule
\end{tabular}
\end{center}
\end{table}

Tables \ref{3-tab:betacoeffs1} and \ref{3-tab:betacoeffs2} display the model's estimates of \(\beta\). We see in Table \ref{3-tab:betacoeffs1} that some groups of items are related to only one dimension of the latent state: the TIPIC-C (personality) items are related to attributes 1 and 2 only, while the seven of the ten PANAS (affect) items are only related to attribute 3. In Table \ref{3-tab:betacoeffs2}, we see that six of the items are involved in no interaction effects. For most of the items for which there are interaction effects, taking into account the attributes involved in both the main and interaction effects leads us to conclude that such items are involved in some way with all three attributes. Supplementary Material K contains further details of the data analysis.

\section{Discussion}

In this paper, we introduced a longitudinal extension of a cross-sectional RLCM with polytomous attributes and covariates. The model has similarities to a previous model introduced by \citet{bartolucci2012latent} but provides a new structure, namely a multivariate probit specification for the transition model which incorporates covariates through the mean of the continuous random vector underlying the discrete latent state vector. The modeling technique presented here lends itself more to diagnosis than do latent trait models such as item response theory or factor analysis. In educational studies, it is often of interest to see how skill mastery can change over time, in particular for measuring the effects of interventions. Here, a Markov modeling approach was presented with covariates with quantifiable effects on transitioning among the latent states. A flexible feature of this model is that the relationship between the item response probabilities and the latent state need not be specified explicitly beforehand, and can be discovered accurately through an exploratory Bayesian approach. In fact, the results from our education application provide evidence that our exploratory model provided improved fit to an educational intervention study in comparison to a confirmatory RLCM as described in previous research. One implication is that our methods can be used to validate expert knowledge about the underlying \(Q\)-matrix and provide a more precise framework for evaluating intervention effects.

A Bayesian modeling approach and corresponding estimation algorithm were presented and shown through simulation to perform well under a variety of settings. In addition to being able to recover the underlying latent structure, our parameter expansion algorithm provides an efficient approach for inferring the parameters of the multivariate probit. Specifically, sampling the multivariate probit correlation matrix and thresholds are notoriously difficult. Our simulation results provide evidence that our novel Bayesian formulation was effective in the thresholds and attribute correlations.

As demonstrated in the emotional state application, we have extended the algorithm to handle intermittent missing data by treating the unobserved responses like other parameters of the model. This is a more efficient alternative to multiple imputation techniques that are common in the literature, and gives a simple approach to addressing missing data which is always an issue in studies such as the one we have presented. One avenue for future work would be to apply the model in a mental health or medical setting where patients transition between states: the model would help researchers evaluate the impact of cognitive therapies or medical treatments.

\section*{Statements and Declarations}

\subsection*{Declaration of Conflicting Interests}

The authors declared no potential conflicts of interest with respect to the research, authorship, and/or publication of this article.

\subsection*{Funding}

This work was partially supported by the U.S. National Science Foundation, under award numbers SES 2150628 and SES 1951057.

\subsection*{Data Availability}

The data \citep{zhan2021data} analyzed in the education application are available in the openICPSR repository, \url{https://doi.org/10.3886/E153061V1}.

The data \citep{shui2021dataset} analyzed in the emotional state application are available in the Synapse repository, \url{https://doi.org/10.7303/syn22418021}.

\bibliographystyle{apacite}
\bibliography{refs}

\newpage

\noindent {\textbf{\LARGE Supplementary Materials}}

\section*{Supplementary Material A}

Supplementary Material A contains the alt text for figures.

\subsection*{Figure 1}

The figure is a directed acyclic graph. There is one plate, which contains a subset of the vertices of the graph: \(Y_n^t\), \(Y_n^{t - 1}\), \(\alpha_n^t\), \(\alpha_n^{t - 1}\).

In addition to the above vertices, there are two vertices not contained in the plate: \(\theta_m\), \(\theta_s\).

The directed edges between vertices are described in the following list, where for example the list item ``\(a\) to \(b\)'' indicates a directed edge from vertex \(a\) to vertex \(b\):

\begin{itemize}
  \item \(\alpha_n^t\) to \(Y_n^t\)
  \item \(\alpha_n^{t - 1}\) to \(Y_n^{t - 1}\)
  \item \(\alpha_n^{t - 1}\) to \(\alpha_n^t\)
  \item \(\theta_m\) to \(Y_n^t\)
  \item \(\theta_m\) to \(Y_n^{t - 1}\)
  \item \(\theta_s\) to \(\alpha_n^t\)
  \item \(\theta_s\) to \(\alpha_n^{t - 1}\)
\end{itemize}

\subsection*{Figure 2}

The figure is a directed acyclic graph. There are two plates, each of which contains a subset of the vertices of the graph: \begin{itemize}
  \item Plate \(j\) contains: \(\delta_j\), \(\beta_j\), \(\kappa_j\), \(Y_{nj}^{\ast, t}\), \(Y_{nj}^t\), \(Y_{nj}^{\ast, t - 1}\), \(Y_{nj}^{t - 1}\)
  \item Plate \(n, t > 2\) contains: \(Y_{nj}^{\ast, t}\), \(Y_{nj}^t\), \(Y_{nj}^{\ast, t - 1}\), \(Y_{nj}^{t - 1}\), \(\alpha_n^t\), \(\alpha_n^{t - 1}\), \(\alpha_n^{\ast, t}\), \(\alpha_n^{\ast, t - 1}\)
\end{itemize}

In addition to the above vertices, there are six vertices not contained in any plate: \(\omega\), \(\gamma\), \(V\), \(R\), \(\lambda\), \(\xi\).

The directed edges between vertices are described in the following list, where for example the list item ``\(a\) to \(b\)'' indicates a directed edge from vertex \(a\) to vertex \(b\):

\begin{itemize}
\item \(\omega\) to \(\delta_j\)
\item \(\delta_j\) to \(\beta_j\)
\item \(\beta_j\) to \(Y_{nj}^{\ast, t}\)
\item \(\beta_j\) to \(Y_{nj}^{\ast, t - 1}\)
\item \(\kappa_j\) to \(Y_{nj}^t\)
\item \(\kappa_j\) to \(Y_{nj}^{t - 1}\)
\item \(\alpha_{n}^t\) to \(Y_{nj}^{\ast, t}\)
\item \(\alpha_{n}^{t - 1}\) to \(Y_{nj}^{\ast, t - 1}\)
\item \(\alpha_{n}^{\ast, t}\) to \(\alpha_{n}^t\)
\item \(\alpha_{n}^{\ast, t - 1}\) to \(\alpha_{n}^{t - 1}\)
\item \(\gamma\) to \(\alpha_{n}^t\)
\item \(\gamma\) to \(\alpha_{n}^{t - 1}\)
\item \(V\) to \(\gamma\)
\item \(R\) to \(V\)
\item \(R\) to \(\alpha_{n}^{\ast, t}\)
\item \(R\) to \(\alpha_{n}^{\ast, t - 1}\)
\item \(R\) to \(\lambda\)
\item \(\lambda\) to \(\alpha_{n}^{\ast, t}\)
\item \(\lambda\) to \(\alpha_{n}^{\ast, t - 1}\)
\item \(\xi\) to \(\alpha_{n}^{\ast, t}\)
\item \(\xi\) to \(\alpha_{n}^{\ast, t - 1}\)
\end{itemize}

\subsection*{Figure 3}

The figure is a directed acyclic graph. There are two plates, each of which contains a subset of the vertices of the graph: \begin{itemize}
  \item Plate \(k\) contains: \(\gamma_k\), \(\alpha_{nk}^t\), \(\alpha_{nk}^{\ast, t}\)
  \item Plate \(n, t\) contains: \(\alpha_{nk}^t\), \(\alpha_{nk}^{\ast, t}\)
\end{itemize}

The directed edges between vertices are described in the following list, where for example the list item ``\(a\) to \(b\)'' indicates a directed edge from vertex \(a\) to vertex \(b\):

\begin{itemize}
\item \(\gamma_k\) to \(\alpha_{nk}^t\)
\item \(\alpha_{nk}^{\ast, t}\) to \(\alpha_{nk}^t\)
\end{itemize}

\section*{Supplementary Material B}

Supplementary Material B describes the data augmentation procedure for the observed data and latent states.

\subsection*{Observed Data}

\noindent The introduction of \(Y_{nj}^{\ast, t}\) is a form of ``data augmentation'' \citep{tanner2010data}. Using the directed local Markov property, we have for all \(n \in [N]\), all \(t \in [T]\), and all \(j \in [J]\), \citep{albert1993bayesian} \begin{align}
  p(Y_{nj}^t, Y_{nj}^{\ast, t} \mid \alpha_n^t, \beta_j, \kappa_j) & = p(Y_{nj}^t \mid Y_{nj}^{\ast, t}, \alpha_n^t, \beta_j, \kappa_j) \cdot p(Y_{nj}^{\ast, t} \mid \alpha_n^t, \beta_j, \kappa_j) \notag\\
  & = p(Y_{nj}^t \mid Y_{nj}^{\ast, t}, \kappa_j) \cdot p(Y_{nj}^{\ast, t} \mid \alpha_n^t, \beta_j).
\end{align}

\noindent Therefore \begin{equation}
  p(Y_{nj}^t \mid \alpha_n^t, \beta_j, \kappa_j) = \int_{\kappa_{j Y_{nj}}}^{\kappa_{j(Y_{nj} + 1)}} \phi(Y_{nj}^{\ast, t};\, d_n^t \beta_j, 1) dY_{nj}^{\ast, t}.
\end{equation}

\subsection*{Latent State}

For all \(n \in [N]\), for \(t \in \{2, 3, \ldots, T\}\), making use of the directed local Markov property, \begin{align}
  p(\alpha_n^t, \alpha_n^{\ast, t} \mid \gamma, \lambda, \alpha_n^{t - 1}, \xi, R) & = p(\alpha_n^t \mid \alpha_n^{\ast, t}, \gamma, \lambda, \alpha_n^{t - 1}, \xi, R) \cdot p(\alpha_n^{\ast, t} \mid \gamma, \lambda, \alpha_n^{t - 1}, \xi, R) \notag\\
  & = p(\alpha_n^t \mid \alpha_n^{\ast, t}, \gamma) \cdot p(\alpha_n^{\ast, t} \mid \lambda, \alpha_n^{t - 1}, \xi, R).
\end{align}

\noindent For \(t = 1\), \begin{align}
  p(\alpha_n^t, \alpha_n^{\ast, t} \mid \gamma, \lambda, R) & = p(\alpha_n^t \mid \alpha_n^t, \gamma, \lambda, R) \cdot p(\alpha_n^{\ast, t} \mid \gamma, \lambda, R) \notag\\
  & = p(\alpha_n^t \mid \alpha_n^{\ast, t}, \gamma) \cdot p(\alpha_n^{\ast, t} \mid \lambda, R).
\end{align}

\noindent Therefore, for \(t \in \{2, 3, \ldots, T\}\), \begin{align}
  p(\alpha_n^t \mid \gamma, \lambda, \alpha_n^{t - 1}, \xi, R) & = \int p(\alpha_n^t, \alpha_n^{\ast, t} \mid \gamma, \lambda, \alpha_n^{t - 1}, \xi, R) d\alpha_n^{\ast, t} \notag\\
  & = \int_{\gamma_{K\alpha_{nK}^t}}^{\gamma_{K, \alpha_{nK}^t + 1}} \ldots \int_{\gamma_{1\alpha_{n1}^t}}^{\gamma_{1, \alpha_{n1}^t + 1}} \phi_K(\alpha_n^{\ast, t}; X_n^t \lambda + d_{n, \text{otr}}^{t - 1} \xi, R) d\alpha_n^{\ast, t}.
\end{align}

\noindent For \(t = 1\), \begin{align}
  p(\alpha_n^t \mid \gamma, \lambda, R) & = \int p(\alpha_n^t, \alpha_n^{\ast, t} \mid \gamma, \lambda, R) d\alpha_n^{\ast, t} \notag\\
  & = \int_{\gamma_{K\alpha_{nK}^t}}^{\gamma_{K, \alpha_{nK}^t + 1}} \ldots \int_{\gamma_{1\alpha_{n1}^t}}^{\gamma_{1, \alpha_{n1}^t + 1}} \phi_K(\alpha_n^{\ast, t}; X_n^t \lambda, R) d\alpha_n^{\ast, t}.
\end{align}

\section*{Supplementary Material C}\label{sec-amalgamating}

Supplementary Material C derives the distribution of \(\alpha^\ast \mid \alpha^{1, \ldots, T-1}, \lambda, \xi, R\).

From assumed relationship \(\alpha_n^{\ast, t} \mid \alpha_n^{t-1}, \lambda, \xi, R \sim N_K(X_n^t \lambda + d_{n, \text{otr}}^{t - 1} \xi, R)\) and the applicable conditional independencies, letting \(\alpha^{\ast, t} = ({\alpha_1^{\ast, t}}^\prime, \ldots, {\alpha_N^{\ast, t}}^\prime)^\prime \in \mathbb{R}^{N \times K}\) and \(d_{\text{otr}}^t = ({d_{1, \text{otr}}^t}^\prime, \ldots, {d_{N, \text{otr}}^t}^\prime)^\prime\), an \(N \times H\) matrix, observe that \begin{equation}
  \alpha^{\ast, t} \mid \alpha^{t-1}, \lambda, \xi, R \sim N_{N, K}(X^t \lambda + d_{\text{otr}}^{t - 1} \xi, I_N \otimes R).
\end{equation}

\noindent Further, considering the expression \(\prod_{t=2}^T p(\alpha^{\ast, t} \mid \alpha^{t-1}, \lambda, \xi, R)\) that appears in the model's recursive factorization, letting \(\alpha^\ast = ({\alpha^{\ast, 1}}^\prime, \ldots, {\alpha^{\ast, T}}^\prime)^\prime \in \mathbb{R}^{NT \times K}\) and \(d_{\text{otr}} = ({d_{\text{otr}}^1}^\prime, \ldots, {d_{\text{otr}}^T}^\prime)^\prime\) a \(TN \times H\) matrix, we can write for the term inside the exponential \begin{align}
& \sum_{t=2}^{T} (\alpha^{\ast, t} - X^t \lambda - d_{\text{otr}}^{t - 1} \xi)^\prime (\alpha^{\ast, t} - X^t \lambda - d_{\text{otr}}^{t - 1} \xi) + (\alpha^{\ast, 1} - X^1 \lambda)^\prime (\alpha^{\ast, 1} - X^1 \lambda) \notag\\
  & = \sum_{t=2}^{T} (\alpha^{\ast, t} - X^t \lambda - d_{\text{otr}}^{t - 1} \xi)^\prime (\alpha^{\ast, t} - X^t \lambda - d_{\text{otr}}^{t - 1} \xi) + (\alpha^{\ast, 1} - X^1 \lambda - O \xi)^\prime (\alpha^{\ast, 1} - X^1 \lambda - O \xi) \notag\\
  & = (\alpha^{\ast} - W \zeta)^\prime (\alpha^{\ast} - W \zeta). \label{bg5intermed1}
\end{align}

\noindent Thus we obtain \begin{align}
  \alpha^\ast \mid \alpha^{1, \ldots, T-1}, \lambda, \xi, R & \sim N_{TN, K}(X \lambda + (O^\prime, {d_{\text{otr}}^{1, \ldots, T-1}}^\prime)^\prime \xi, I_{TN} \otimes R) \notag\\
& = N_{TN, K}(W \zeta, I_{TN} \otimes R).
\end{align}

\section*{Supplementary Material D}

Supplementary Material D provides information on the matrix variate normal distribution.

The following definitions and theorems are quoted nearly verbatim from \cite{gupta2000matrix}, and use the same numbering.

\vspace{\baselineskip}

\noindent \textbf{Definition 2.2.1} (Gupta and Nagar): The random matrix \(X (p \times n)\) is said to have a matrix variate normal distribution with mean matrix \(M (p \times n)\) and covariance matrix \(\Sigma \otimes \Psi\) where \(\Sigma (p \times p) > 0\) and \(\Psi (n \times n) > 0\), if \(\text{vec}(X^\prime) \sim N_{pn}(\text{vec}(M^\prime), \Sigma \otimes \Psi)\), a multivariate normal distribution. For such an \(X\), we write \(X \sim N_{p,n}(M, \Sigma \otimes \Psi)\).

\vspace{\baselineskip}

\noindent \textbf{Theorem 2.2.1} (Gupta and Nagar): If \(X \sim N_{p,n}(M, \Sigma \otimes \Psi)\), then the p.d.f. of \(X\) is given by \begin{equation*}
    (2\pi)^{-\frac{1}{2}np} (\det(\Sigma))^{-\frac{1}{2}n} (\det(\Psi))^{-\frac{1}{2}p} \text{etr}\left\{-\frac{1}{2} \Sigma^{-1} (X - M) \Psi^{-1} (X - M)^\prime\right\}
  \end{equation*}
  \noindent where \text{etr} denotes the expected value of the trace and where of course \(X \in \mathbb{R}^{p \times n}\) and \(M \in \mathbb{R}^{p \times n}\).

\vspace{\baselineskip}

\noindent \textbf{Theorem 2.3.1} (Gupta and Nagar): If \(X \sim N_{p,n}(M, \Sigma \otimes \Psi)\), then \(X^\prime \sim N_{n,p}(M^\prime, \Psi \otimes \Sigma)\).

\section*{Supplementary Material E}

Supplementary Material E contains the proof of Theorem 1 as well as proofs of other results on which the proof of Theorem 1 depends.

Theorem 1 is proved in section E.1. The proof of Theorem 1 relies on Theorem \ref{identiflikfinalident}, which is proved in section E.2.

\subsection*{E.1 \quad Proof of Theorem 1 (Identifiability of Longitudinal Model)}

When we say ``up to label swapping of dimensions of \(\alpha_n^t\)'' or simply ``up to label swapping,'' we mean that for each relevant vector or matrix parameter \(a\) of \(\Theta\), that \(\widetilde{a} = g(a)\) where \(g\) is a transformation that permutes (1) the rows, (2) the columns, or (3) both the rows and the columns of \(a\) on the dimension(s) for which the latent state value varies. Let \(R\) be the matrix that permutes the rows of a matrix or vector when applied on the left. Recall that \(R^\prime\) is then the permutation matrix that performs the corresponding permutation of the columns of a matrix when applied on the right.

Note that for readability, we abstract away indexing issues by sometimes writing vectors equal to scalars when we define the entries of matrices. 

Let \(U = \{U_{n, t, t-1} \ssep n \in [N], t \in \{2, 3, \ldots, T\}\}\), namely the set of all transition matrices for all respondents, and let \(P = \{\pi_n^t \ssep n \in [N], t \in [T]\}\), namely the set of all marginal latent state probabilities for all respondents for all time points. Write \(\{B, U, P\}\) for the set of all possible triples of \(B\), \(U\), and \(P\). Also, when we write for example \(\{U\}\) we mean the set of all possible \(U\).

The steps of the proof are as follows. We assume conditions (C1) through (C6) hold. We will show

\begin{enumerate}
\item[(1)] (Corollary \ref{identiflongit04}) \(\{B, U, P\}\) is strictly identifiable up to label swapping from \(\{p(Y \mid \theta)\}\)
\item[(2)] (Lemma \ref{identiflongit2a01}) \(\forall\, U, \widetilde{U} \in \{U\} \enspace \forall\, \theta_s, \widetilde{\theta_s} \in \Theta_s \quad U \sim_E \widetilde{U} \implies \theta_s \sim_E \widetilde{\theta_s}\) (this relies on Theorem \ref{identiflikfinalident})
\end{enumerate}

By the result of \citet{he2023sparse}, \begin{enumerate}
\item[(3)] \(\Theta_m\) is generically identifiable up to label swapping from \(\{B\}\) if condition (D1) holds, and \(\Theta_m\) is strictly identifiable up to label swapping from \(\{B\}\) if both conditions (D1) and (D2) hold.
  \end{enumerate}

Putting (2) and (3) together, we have that if both (D1) and (D2) hold, then \((B, U) \sim_E (\widetilde{B}, \widetilde{U}) \implies (\theta_m, \theta_s) \sim_E (\widetilde{\theta_m}, \widetilde{\theta_s})\). If only (D1) holds, then for all all \((\widetilde{\theta_m}, \widetilde{\theta_s})\) other than those on a measure zero subset we have \((B, U) \sim_E (\widetilde{B}, \widetilde{U}) \implies (\theta_m, \theta_s) \sim_E (\widetilde{\theta_m}, \widetilde{\theta_s})\).

From (1) through (3) we have that if condition (D1) holds, then \(\Theta\) is generically identifiable up to label swapping from \(\{p(Y \mid \theta)\}\), and if both conditions (D1) and (D2) hold then \(\Theta\) is strictly identifiable up to label swapping from \(\{p(Y \mid \theta)\}\).

We now show results (1) and (2).

\subsubsection*{E.1.1 \quad Showing Result (1)}

\begin{lemma}\label{identiflongit01}
  If there exist subsets \(J_1, J_2, J_3\) of the items such that partitioning the emissions matrix \((p(Y_n^{t-1} \mid \alpha_n^{t-1}, \theta))\) into matrices \((p(Y_{n, J_1}^{t-1} \mid \alpha_n^{t-1}, \theta))\), \((p(Y_{n, J_2}^{t-1} \mid \alpha_n^{t-1}, \theta))\), and \((p(Y_{n, J_3}^{t-1} \mid \alpha_n^{t-1}, \theta))\) whose ranks are all greater than or equal to \(L^K\), then \(\{\pi_n^{t-1}\}\) is strictly identifiable up to label swapping from \(\{p(Y_n^{t-1} \mid \theta)\}\).
\end{lemma}

\begin{proof}
Consider \(\{p_{Y_n^{t-1}}(\bullet \mid \theta)\}\). Split the vector \(Y_n^{t-1}\) into \(Y_n^{t-1} = (Y_{n, J_1}^{t-1}, Y_{n, J_2}^{t-1}, Y_{n, J_3}^{t-1})\). Note that by the conditional independencies, for all values of \(Y_n^{t-1}\), \(p(Y_n^{t-1} \mid \alpha_n^{t-1}, \theta) = p(Y_{n, J_1}^{t-1} \mid \alpha_n^{t-1}, \theta) \cdot p(Y_{n, J_2}^{t-1} \mid \alpha_n^{t-1}, \theta) \cdot p(Y_{n, J_3}^{t-1} \mid \alpha_n^{t-1}, \theta)\). Arranging \((p(Y_n^{t-1} \mid \theta))\) as a three-way array and letting \(\otimes\) denote the outer product, observe that we can write \begin{align}
  & (p(Y_n^{t-1} \mid \theta)) \notag\\
  & = \sum_{l=1}^{L^K} p(\alpha_n^{t-1} = l \mid \theta) \cdot (p(Y_{n, J_1}^{t-1} \mid \alpha_n^{t-1} = l, \theta)) \otimes (p(Y_{n, J_2}^{t-1} \mid \alpha_n^{t-1} = l, \theta)) \otimes (p(Y_{n, J_3}^{t-1} \mid \alpha_n^{t-1} = l, \theta)).
\end{align}

\noindent By Theorem 3 of \citet{bonhomme2016estimating}, it follows that \(\{(p(Y_{n, J_1}^{t-1} \mid \alpha_n^{t-1}, \theta)), \\ (p(Y_{n, J_2}^{t-1} \mid \alpha_n^{t-1}, \theta)), (p(Y_{n, J_3}^{t-1} \mid \alpha_n^{t-1}, \theta)), \pi_n^{t-1}\}\) is strictly identifiable up to label swapping from \(\{p(Y_n^{t-1} \mid \theta)\}\).
\end{proof}

\begin{lemma}\label{identiflongit02}
  Let \(Y_n^{t-1}, Y_n^t, Y_n^{t+1}\) be discrete random vectors. Let \(\theta \in \Theta\). For \(l \in [L^K]\) define the row vector \((m_{n, l}^{t_1, t_2})_i = p(Y_n^{t_1} = i \mid \alpha_n^{t_2} = l, \theta)\), and define \(m_n^{t_1, t_2} = ({m_{n,1}^{t_1, t_2}}^\prime, \ldots, {m_{n,L^K}^{t_1, t_2}}^\prime)^\prime\) (the matrix whose rows consist of vectors \(m_{n, l}^{t_1, t_2}\)). It holds that \(\{m_n^{t-1, t}, m_n^{t, t}, m_n^{t+1, t}, \pi_n^t\}\) is strictly identifiable up to label swapping from \(\{p(Y_n^{t-1}, Y_n^t, Y_n^{t+1} \mid \theta)\}\).
\end{lemma}

\begin{proof}
  We first relate the matrices \(m_n^{t-1, t}\), \(m_n^{t, t}\), and \(m_n^{t+1, t}\) to matrices whose ranks are known. First, observe that \(m_n^{t, t} = B\). Second, we find \(m_n^{t+1, t}\): \begin{align}
  \left(B {U_{n,t+1, t}}^\prime\right)_{ij} & = \sum_{k=1}^{L^K} p(Y_n^{t+1} = i \mid \alpha_n^{t+1} = k, \theta) \, p(\alpha_n^{t+1} = k \mid \alpha_n^t = j, \theta) \notag\\
  & = p(Y_n^{t+1} = i \mid \alpha_n^t = j, \theta) = \left(m_n^{t+1, t}\right)_{ij}
\end{align}

Third, we find \(m_n^{t-1, t}\). Note that \begin{align}
  \left(\text{diag}(\pi_n^{t-1})\, U_{n, t, t-1}\, (\text{diag}(\pi_n^t))^{-1}\right)_{ij} & = p(\alpha_n^{t-1} = i \mid \theta) \, p(\alpha_n^t = j \mid \alpha_n^{t-1} = i, \theta) \notag\\
  & \qquad\qquad p(\alpha_n^t = j \mid \theta)^{-1} \notag\\
  & = p(\alpha_n^t = j, \alpha_n^{t-1} = i \mid \theta) \, p(\alpha_n^t = j \mid \theta)^{-1} \notag\\
  & = p(\alpha_n^{t-1} = i \mid \alpha_n^t = j, \theta)
\end{align}

\noindent and therefore \begin{align}
  & \left(B\, \text{diag}(\pi_n^{t-1})\, U_{n, t, t-1}\, (\text{diag}(\pi_n^t))^{-1}\right)_{ij} \notag\\
  & = \sum_{k=1}^{L^K} p(Y_n^{t-1} = i \mid \alpha_n^{t-1} = k, \theta) \, p(\alpha_n^{t-1} = k \mid \alpha_n^t = j, \theta) \notag\\
  & = p(Y_n^{t-1} = i \mid \alpha_n^t = j, \theta) = \left(m_n^{t-1, t}\right)_{ij}
\end{align}

Summarizing, we have that \[\begin{array}{lcl}
  m_n^{t-1, t} & = & B \, \text{diag}(\pi_n^{t-1})\, U_{n, t, t-1}\, (\text{diag}(\pi_n^t))^{-1} \\
  m_n^{t, t} & = & B \\
  m_n^{t+1, t} & = & B {U_{n,t+1, t}}^\prime
\end{array}\]

We now show that \(\text{rank}(m_n^{t-1, t}) = \text{rank}(m_n^{t, t}) = \text{rank}(m_n^{t+1, t}) = L^K\). First, \(\text{rank}(m_n^{t, t}) = \text{rank}(B) = L^K\). For \(m_n^{t+1, t}\), observe that since \(m_n^{t+1, t} = B {U_{n,t+1, t}}^\prime\), we have \(\text{rank}(m_n^{t+1, t}) \leq \text{rank}(B)\) and that \begin{equation}
  \text{rank}(B) = \text{rank}(B {U_{n,t+1, t}}^\prime ({U_{n,t+1, t}}^\prime)^{-1}) \leq \text{rank}(B {U_{n,t+1, t}}^\prime)
\end{equation}

\noindent so \(\text{rank}(m_n^{t+1, t}) = \text{rank}(B) = L^K\). Finally, since \(m_n^{t-1, t} = B \, \text{diag}(\pi_n^{t-1})\, U_{n, t, t-1}\, (\text{diag}(\pi_n^t))^{-1}\), we first have \(\text{rank}(m_n^{t-1, t}) \leq \text{rank}(B)\). Since \(\text{diag}(\pi_n^{t-1})\), \(U_{n, t, t-1}\), and \((\text{diag}(\pi_n^t))^{-1}\) are all of dimension \(L^K \times L^K\) and are all of rank \(L^K\), each is invertible and thus their product is invertible. For a moment, let \(A = \text{diag}(\pi_n^{t-1})\, U_{n, t, t-1}\, (\text{diag}(\pi_n^t))^{-1}\). Observe that \begin{equation}
  \text{rank}(B) = \text{rank}(B A A^{-1}) \leq \text{rank}(BA) = \text{rank}(m_n^{t-1, t})
\end{equation}

\noindent and thus \(\text{rank}(m_n^{t-1, t}) = \text{rank}(B) = L^K\).

Observe that we can write \begin{equation}
  (p(Y_n^{t-1}, Y_n^t, Y_n^{t+1} \mid \theta)) = \sum_{l=1}^{L^K} p(\alpha_n^t = l \mid \theta) \cdot m_{n, l}^{t-1, t} \otimes m_{n, l}^{t, t} \otimes m_{n, l}^{t+1, t}
\end{equation}

\noindent By Theorem 3 of \citet{bonhomme2016estimating}, since \(m_{n, l}^{t-1, t}\), \(m_{n, l}^{t, t}\), and \(m_{n, l}^{t+1, t}\) are all full (column) rank, we have that \(\{m_n^{t-1, t}, m_n^{t, t}, m_n^{t+1, t}, \pi_n^t\}\) is strictly identifiable up to label swapping from \(\{p(Y_n^{t-1}, Y_n^t, Y_n^{t+1} \mid \theta)\}\).
\end{proof}

\begin{theorem}\label{identiflongit03}
  \(\{U_{n, t, t-1}, U_{n, t+1, t}, B, \pi_n^{t-1}, \pi_n^t\}\) is strictly identifiable up to label swapping from \\\(\{p(Y_n^{t-1}, Y_n^t, Y_n^{t+1} \mid \theta)\}\).
\end{theorem}

\begin{proof}
  To prove this theorem, we will show that \begin{align}
    & \{m_n^{t-1, t}, m_n^{t, t}, m_n^{t+1, t}, \pi_n^t\} \sim_E \{\widetilde{m}_n^{t-1, t}, \widetilde{m}_n^{t, t}, \widetilde{m}_n^{t+1, t}, \widetilde{\pi}_n^t\} \notag\\
    & \implies \{U_{n, t, t-1}, U_{n, t+1, t}, B, \pi_n^{t-1}, \pi_n^t\} \sim_E \{\widetilde{U}_{n, t, t-1}, \widetilde{U}_{n, t+1, t}, \widetilde{B}, \widetilde{\pi}_n^{t-1}, \widetilde{\pi}_n^t\}.\label{identiflongit03fact}
  \end{align}

\noindent Combining this result with Lemma \ref{identiflongit02} will result in the desired conclusion. We now show that \eqref{identiflongit03fact} holds.

Assume \(\{m_n^{t-1, t}, m_n^{t, t}, m_n^{t+1, t}, \pi_n^t\} \sim_E \{\widetilde{m}_n^{t-1, t}, \widetilde{m}_n^{t, t}, \widetilde{m}_n^{t+1, t}, \widetilde{\pi}_n^t\}\).

First, \(\pi_n^t \sim_E \widetilde{\pi}_n^t\), which holds if and only if \(\widetilde{\pi}_n^t = R \pi_n^t\) (and thus \(\text{diag}(\widetilde{\pi}_n^t) = R \, \text{diag}(\pi_n^t) \, R^\prime\)).

Second, \(m_n^{t, t} \sim_E \widetilde{m}_n^{t, t}\), which holds if and only if \[\begin{array}{llcl}
  & \widetilde{m}_n^{t, t} & = & m_n^{t, t} R^\prime \\
\therefore\quad & \widetilde{B} & = & B R^\prime
\end{array}\]

\noindent and thus \(\widetilde{B} \sim_E B\).

Third, \(m_n^{t+1, t} \sim_E \widetilde{m}_n^{t+1, t}\), which holds if and only if \[\def\arraystretch{1.2}
\begin{array}{rcl}
  & \widetilde{m}_n^{t+1, t} & = m_n^{t+1, t} R^\prime \\
\therefore\quad & \widetilde{B} \, \widetilde{U}_{n,t+1, t}^\prime & = (B \, U_{n,t+1, t}^\prime) \, R^\prime \\
\therefore\quad & (B R^\prime) \, \widetilde{U}_{n,t+1, t}^\prime & = B \, U_{n,t+1, t}^\prime \, R^\prime \\
\therefore\quad & \widetilde{U}_{n,t+1, t} \, (R B^\prime) & = R \, U_{n,t+1, t} \, B^\prime \\
\therefore\quad & \widetilde{U}_{n,t+1, t} \, (R B^\prime) \, B & = R \, U_{n,t+1, t} \, B^\prime \, B\\
\equiv\quad & \widetilde{U}_{n,t+1, t} \, R & = R \, U_{n,t+1, t}\\
\therefore\quad & \widetilde{U}_{n,t+1, t} \, R \, R^\prime & = R \, U_{n,t+1, t} \, R^\prime \\
\equiv\quad & \widetilde{U}_{n,t+1, t} & = R \, U_{n,t+1, t} \, R^\prime
\end{array}\]

\noindent where we have used (C3). Thus we have shown that \(\widetilde{U}_{n,t+1, t} \sim_E U_{n,t+1, t}\).

Fourth, \(m_n^{t-1, t} \sim_E \widetilde{m}_n^{t-1, t}\), which holds if and only if \[\def\arraystretch{1.2}
\begin{array}{rcl}
  & \widetilde{m}_n^{t-1, t} & = m_n^{t-1, t} R^\prime \\
  \therefore\quad & \widetilde{B} \, \text{diag}(\widetilde{\pi}_n^{t-1})\, \widetilde{U}_{n, t, t-1}\, (\text{diag}(\widetilde{\pi}_n^t))^{-1} & = B \, \text{diag}(\pi_n^{t-1})\, U_{n, t, t-1}\, (\text{diag}(\pi_n^t))^{-1} R^\prime\\
  \therefore\quad & (B R^\prime) (R \, \text{diag}(\pi_n^{t-1}) \, R^\prime) \widetilde{U}_{n, t, t-1} \, (R \, (\text{diag}(\pi_n^t))^{-1} \, R^\prime) & = B \, \text{diag}(\pi_n^{t-1})\, U_{n, t, t-1}\, (\text{diag}(\pi_n^t))^{-1} R^\prime\\
  \therefore\quad & B \, \text{diag}(\pi_n^{t-1}) \, R^\prime \, \widetilde{U}_{n, t, t-1} \, (R \, (\text{diag}(\pi_n^t))^{-1} \, R^\prime) & = B \, \text{diag}(\pi_n^{t-1})\, U_{n, t, t-1}\, (\text{diag}(\pi_n^t))^{-1} R^\prime\\
  \therefore\quad & B \, \text{diag}(\pi_n^{t-1}) \, R^\prime \, \widetilde{U}_{n, t, t-1} \, R \, (\text{diag}(\pi_n^t))^{-1} & = B \, \text{diag}(\pi_n^{t-1})\, U_{n, t, t-1}\, (\text{diag}(\pi_n^t))^{-1} \\
  \therefore\quad & \text{diag}(\pi_n^{t-1}) \, R^\prime \, \widetilde{U}_{n, t, t-1} \, R \, \text{diag}(\pi_n^t))^{-1} & = \text{diag}(\pi_n^{t-1})\, U_{n, t, t-1}\, (\text{diag}(\pi_n^t))^{-1} \\
  \therefore\quad & R^\prime \, \widetilde{U}_{n, t, t-1} \, R & = U_{n, t, t-1} \\
  \therefore\quad & \widetilde{U}_{n, t, t-1} & = R \, U_{n, t, t-1}\, R^\prime \\
\end{array}\]

\noindent where in the above derivation we have used Lemma \ref{identiflongit01} and assumptions (C1), (C3) and (C4). Thus we have shown  \(\widetilde{U}_{n, t, t-1} \sim_E U_{n, t, t-1}\).

Thus we have shown that \(\{m_n^{t-1, t}, m_n^{t, t}, m_n^{t+1, t}, \pi_n^t\} \sim_E \{\widetilde{m}_n^{t-1, t}, \widetilde{m}_n^{t, t}, \widetilde{m}_n^{t+1, t}, \widetilde{\pi}_n^t\}\) implies that \(\{U_{n, t, t-1}, U_{n, t+1, t}, B, \pi_n^{t-1}, \pi_n^t\} \sim_E \{\widetilde{U}_{n, t, t-1}, \widetilde{U}_{n, t+1, t}, \widetilde{B}, \widetilde{\pi}_n^{t-1}, \widetilde{\pi}_n^t\}\). By this fact and Lemma \ref{identiflongit02}, we have shown that \(\{U_{n, t, t-1}, U_{n, t+1, t}, B, \pi_n^{t-1}, \pi_n^t\}\) is strictly identifiable up to label swapping from \(\{p(Y_n^{t-1}, Y_n^t, Y_n^{t+1} \mid \theta)\}\).
\end{proof}

\begin{corollary}\label{identiflongit04}
\(\{B, U, P\}\) is strictly identifiable up to label swapping from \(\{p(Y \mid \theta)\}\).
\end{corollary}

\begin{proof}
  Extending Theorem \ref{identiflongit03} across all \(t \in \{2, \ldots, T\}\) and all \(n \in [N]\) gives the result.
\end{proof}

\subsubsection*{E.1.2 \quad Showing Result (2)}

\begin{lemma}\label{identiflongit2a01}
  For all \(U, \widetilde{U} \in \{U\}\), \(U \sim_E \widetilde{U}\) implies that \(\theta_s \sim_E \widetilde{\theta_s}\).
\end{lemma}

\begin{proof}
  Let \(\sigma\) be the permutation of the dimensions of \(\theta_s\) as well as the permutation of the columns and rows of any transition matrix \(U_{n,t, t-1}\), which is induced by permuting the dimensions of \(\alpha^t\) (strictly speaking separate symbols should be used, but that would hinder readability).

  The likelihood \(p(\alpha^t \mid \alpha^{t-1}, \theta_s)\) (where we have again used shorthand; this expression refers to the likelihood values for all possible values of \(\alpha^t\)) forms a column of the matrix \(U_{n, t, t-1}\); allowing \(\alpha^{t-1}\) to range across all possible values yields the matrix \(U_{n, t, t-1}\). In this proof, denote the \(U_{n, t, t-1}\) that results from a particular value \(\theta \in \Theta_s\) as \(f_{\theta}\).

  \(U \sim_E \widetilde{U}\) implies that \(f_{\theta_2} = \sigma(f_{\theta_1})\), where \(\sigma(f_{\theta_1})\) is the matrix resulting from permuting the relevant columns and rows of \(f_{\theta_1}\). We observe that \(\sigma(f_{\theta_1}) = f_{\sigma(\theta_1)}\), i.e. permuting the columns and rows of \(f_{\theta_1}\) according to the dimension reordering gives the same matrix as first permuting the dimensions of \(\alpha^{t-1}\) and \(\theta\) and then writing down the elements of the likelihood vectors forming the transition matrix in the proper row order and placing the vectors themselves in the proper column order.

  Let \(g\) be the parameterization map \(\theta \mapsto f_{\theta}\). By the above, we have \(g(\theta_2) = f_{\theta_2} = \sigma(f_{\theta_1}) = f_{\sigma(\theta_1)} = g(\sigma(\theta_1))\). By Theorem \ref{identiflikfinalident}, \(g\) is injective, so we have \(\theta_2 = \sigma(\theta_1)\), i.e. \(\theta_1 \sim_E \theta_2\). We have thus shown that \(f_{\theta_1} \sim_E f_{\theta_2}\) implies \(\theta_1 \sim_E \theta_2\) (i.e. \(U \sim_E \widetilde{U}\) implies that \(\theta_s \sim_E \widetilde{\theta_s}\)).
\end{proof}

\subsection*{E.2 \quad Proof of Strict Identifiability of Multivariate Probit Model}

Given a data matrix \(W \in \mathbb{R}^{N \times (D + H_{\text{otr}})}\), where \(N \geq D + H_{\text{otr}}\), write \(\mathcal{M} = \mathcal{M}_1 \times \cdots \times \mathcal{M}_N\), where \(\mathcal{M}_n = \{M_n \ssep M_n = W_n \zeta, \zeta \in \mathbb{R}^{(D+H_{\text{otr}}) \times K}\}\) is the \(n\)th row of \(M\) and \(W_n\) is the \(n\)th row of \(W\), so \(\mathcal{M} = \{M; M = W \zeta,\, \zeta \in \mathbb{R}^{(D + H_{\text{otr}}) \times K}\} \subseteq \mathbb{R}^{N \times K}\).

We consider the family of densities for a sample \((\alpha_1, \ldots, \alpha_n) \in \times_{n=1}^N A_L\), namely \(\{p_{\alpha}(\bullet \mid \omega) \ssep \omega \in \Omega\}\), where \(\Omega = \mathcal{M} \times \mathcal{G} \times \mathcal{R}\), \(\mathcal{G} = \mathcal{G}_1 \times \cdots \times \mathcal{G}_K\) and where for all \(k \in [K]\), \begin{align}
  \mathcal{G}_k & = \{(\gamma_{k0}, \gamma_{k1}, \ldots, \gamma_{kL}) \ssep  \gamma_{k0} = -\infty, \gamma_{k1} = 0, \notag\\
  & \quad\quad \gamma_{k1} < \gamma_{k2} < \gamma_{k3} < \ldots, < \gamma_{k,L-1} < \infty, \gamma_{k,L} = \infty\},
\end{align}

\noindent where \(\mathcal{R} = \{R \in \mathbb{R}^{K \times K} \ssep R \text{ is positive definite, } \text{diag}(R) = (1, \ldots, 1)\}\), and where  \begin{align}
  p(\alpha \mid \omega) & = \prod_{n=1}^N p(\alpha_n \mid M_n, \gamma, R) \notag\\
  & = \prod_{n=1}^N \int_{\gamma_{K\alpha_{nK}}}^{\gamma_{K, \alpha_{nK} + 1}} \ldots \int_{\gamma_{1\alpha_{n1}}}^{\gamma_{1, \alpha_{n1} + 1}} \phi_K(\alpha_n^{\ast};\, M_n, R) d\alpha_n^{\ast}\label{mvpdensity}
\end{align}

\noindent where \(M_n \in \mathcal{M}_n\) for each \(n\). We further define \(\Theta_s = \mathcal{Z} \times \mathcal{G} \times \mathcal{R}\), where \(\mathcal{Z} = \mathbb{R}^{(D + H_{\text{otr}}) \times K}\), the parameter space for \(\zeta\). Note that the actual parameter space for this multivariate probit model is \(\Theta_s\); we use \(\Omega\) as a starting point.

The outline of the proof is as follows. We will first show that the family of densities \(\{p(\alpha \mid \omega)\}\) is strictly identifiable. We then show that if \(\text{rank}(W) = D + H_{\text{otr}}\), then \(\Theta_s\) and \(\Omega\) are isomorphic. We conclude that if \(\text{rank}(W) = D + H_{\text{otr}}\), then \(\{p(\alpha \mid \theta_s)\}\) is strictly identifiable.

\subsubsection*{E.2.1 \quad Showing the Strict Identifiability of the Family of Densities of a Multivariate Probit Model for One or More Respondents with Mean Matrix Parameter}

\begin{theorem}
\label{identifmvp1}
\(\{p(\alpha \mid \omega)\}\) as defined in \eqref{mvpdensity} is strictly identifiable.
\end{theorem}

\begin{proof}
Let \(\omega\) and \(\widetilde{\omega}\) be arbitrary values of \(\Omega\). Assume that for all \(\alpha \in \times_{n=1}^N A_L\), \(p_{\alpha}(\alpha \mid \omega) = p_{\alpha}(\alpha \mid \widetilde{\omega})\). This means we have a set \(S\) of \(L^{NK}\) equations, with one equation for each possible \(\alpha = (\alpha_1, \ldots, \alpha_n)\). More specifically, we have \begin{equation}
  S = \{u_{\alpha}(\omega, \widetilde{\omega}) = 0 \ssep \alpha \in \times_{n=1}^N A_L\},
\end{equation}

\noindent where \(u_{\alpha}(\omega, \widetilde{\omega})\) stands for \(v_{\alpha}(\omega) - v_{\alpha}(\widetilde{\omega})\), and \(v_{\alpha}(\omega)\) stands for \begin{equation}
  \prod_{n=1}^N \int_{\gamma_{K\alpha_{nK}}}^{\gamma_{K, \alpha_{nK} + 1}} \ldots \int_{\gamma_{1\alpha_{n1}}}^{\gamma_{1, \alpha_{n1} + 1}} \phi_K(\alpha_n^{\ast}; M_n, R) d\alpha_n^{\ast}.
\end{equation}

Showing \(\omega = \widetilde{\omega}\) means showing \((M, \gamma, R) = (\widetilde{M}, \widetilde{\gamma}, \widetilde{R})\). We first show that \(M = \widetilde{M}\), by repeating the following procedure for each \(n \in [N]\). Select an arbitrary \(n\). Each possible value for \(\alpha_n\) appears in \(L^{(N-1)K}\) equations. For each possible value of \(\alpha_n\), sum those equations to yield an equation \(u_{\alpha_n}(\omega, \widetilde{\omega}) = 0\), where \(u_{\alpha_n}(\omega, \widetilde{\omega})\) stands for \(v_{\alpha_n}(\omega) - v_{\alpha_n}(\widetilde{\omega})\), and \(v_{\alpha_n}(\omega)\) stands for \begin{equation}
  \int_{\gamma_{K\alpha_{nK}}}^{\gamma_{K, \alpha_{nK} + 1}} \ldots \int_{\gamma_{1\alpha_{n1}}}^{\gamma_{1, \alpha_{n1} + 1}} \phi_K(\alpha_n^{\ast}; M_n, R) d\alpha_n^{\ast}.
\end{equation}

\noindent This holds because for each \(i \in [N] \setminus n\), summing over all values of \(\alpha_i\) produces an integral over the entire support of the multivariate probit specification for that \(\alpha_i\), which evaluates to 1. Since there are \(L^K\) possible values for \(\alpha_n\), this step of the procedure has yielded \(L^K\) equations. Denote the set of these equations by \(S_n\), namely \begin{equation}
  S_n = \{u_{\alpha_n}(\omega, \widetilde{\omega}) = 0 \ssep \alpha_n \in A_L\}.
\end{equation}

Now, write \(M = (M_{n1}, \ldots, M_{nK})\). Select an arbitrary dimension \(k\). Each value of \(\alpha_{nk} \in \{0, 1, \ldots, L - 1\}\) appears in \(L^{K - 1}\) equations of \(S_n\) (since each of \(\alpha_{n1}, \ldots, \alpha_{n,k-1}, \alpha_{n,k+1}, \ldots, \alpha_{nK}\) has \(L\) possible values). For a value \(l\) of \(\alpha_{nk}\) denote \begin{equation}
  S_{nk}^l = \{u_{\alpha_n}(\gamma, M_n, R, \widetilde{\gamma}, \widetilde{M_n}, \widetilde{R}) = 0 \ssep \alpha_{n(k)}\in \{0, 1, \ldots, L - 1\}^{K - 1}, \alpha_{nk} = l\} \subseteq S_n.
\end{equation}

\noindent Summing together the equations in \(S_{nk}^l\) gives one equation, \begin{equation}
  \int_{\gamma_{k\alpha_{nl}}}^{\gamma_{k, \alpha_{nl} + 1}} \phi(\alpha_{nk}^{\ast}; M_{n,l+1}, 1) d\alpha_{nk}^{\ast} - \int_{\widetilde{\gamma}_{k\alpha_{nl}}}^{\widetilde{\gamma}_{k, \alpha_{nl} + 1}} \phi(\alpha_{nk}^{\ast}; \widetilde{M}_{n,l+1}, 1) d\alpha_{nk}^{\ast} = 0
\end{equation}

\noindent since the summation produces integrals whose bounds are over the entire real line in each dimension other than \(k\), corresponding to the support of the density on those dimensions.

Performing the summation for \(S_{nk}^0\) gives \begin{equation}
  \int_{\gamma_{k,0}}^{\gamma_{k,1}} \phi(\alpha_{nk}^{\ast}; M_{n1}, 1) d\alpha_n^{\ast} - \int_{\widetilde{\gamma}_{k,0}}^{\widetilde{\gamma}_{k,1}} \phi(\alpha_{nk}^{\ast}; \widetilde{M}_{n1}, 1) d\alpha_n^{\ast} = 0
\end{equation}

\noindent which is \begin{equation}
  \int_{-\infty}^{0} \phi(\alpha_{nk}^{\ast}; M_{n1}, 1) d\alpha_n^{\ast} - \int_{-\infty}^{0} \phi(\alpha_{nk}^{\ast}; \widetilde{M}_{n1}, 1) d\alpha_n^{\ast} = 0.
\end{equation}

\noindent This is \(\Phi(0 - M_{n1}) - \Phi(0 - \widetilde{M}_{n1}) = 0\), which gives \(M_{n1} = \widetilde{M_{n1}}\). Repeating this process across all \(k \in [K]\) yields \(M_n = \widetilde{M_n}\).

Repeating this process across all \(n \in [N]\) yields \(M = \widetilde{M}\).

We next show that for all \(k \in [K]\), \(\gamma_k = \widetilde{\gamma}_k\). Note that if \(L = 2\), all elements of all \(\gamma_k\) are fixed and there is nothing to be shown. If \(L > 2\) there is at least one element of each \(\gamma_k\) which can vary.

Choose an arbitrary \(n \in [N]\). Choose a \(k \in [K]\). For \(l = 1\), add the equations in \(S_{nk}^l\), which results in the equation \begin{equation}
  \int_{0}^{\gamma_{k,l+1}} \phi(\alpha_{nk}^{\ast}; M_{nk}, 1) d\alpha_n^{\ast} - \int_{0}^{\widetilde{\gamma}_{k,l+1}} \phi(\alpha_{nk}^{\ast}; M_{nk}, 1) d\alpha_n^{\ast} = 0
\end{equation}

\noindent which is equivalent to \(\Phi(\gamma_{k,l+1} - M_{nk}) - \Phi(0 - M_{nk}) = \Phi(\widetilde{\gamma}_{k,l+1} - M_{nk}) - \Phi(0 - M_{nk})\), so \(\gamma_{k,l+1} = \widetilde{\gamma}_{k,l+1}\). Then sequentially, for each \(l > 1\), adding the equations in \(S_{nk}^l\) yields \begin{equation}
  \int_{\gamma_{kl}}^{\gamma_{k,l+1}} \phi(\alpha_{nk}^{\ast}; M_{nk}, 1) d\alpha_n^{\ast} - \int_{\gamma_{kl}}^{\widetilde{\gamma}_{k,l+1}} \phi(\alpha_{nk}^{\ast}; M_{nk}, 1) d\alpha_n^{\ast} = 0
\end{equation}

\noindent which is \(\Phi(\gamma_{k,l+1} - M_{nk}) - \Phi(\gamma_{kl} - M_{nk}) = \Phi(\widetilde{\gamma}_{k,l+1} - M_{nk}) - \Phi(\gamma_{kl} - M_{nk})\), which yields \(\gamma_{k,l+1} = \widetilde{\gamma}_{k,l+1}\). In this manner we obtain \(\gamma_k = \widetilde{\gamma}_k\).

We now show for each \(i, j \in [K] \times [K]\) that \(R_{ij} = \widetilde{R}_{ij}\). Choose an arbitrary \(n \in [N]\). Denote \begin{equation}
  S_{n,i,j}^0 = \{u_{\alpha_n}(\gamma, M_n, R, \widetilde{\gamma}, \widetilde{M_n}, \widetilde{R}) = 0;\, \alpha_{n(i,j)} \in \{0, 1, \ldots, L - 1\}^{K - 2}, \alpha_{ni} = \alpha_{nj} = 0\} \subseteq S_n.
\end{equation}

\noindent Adding the equations in \(S_{n,i,j}^0\) yields \begin{align}
  & \int_{-\infty}^{0} \int_{-\infty}^{0} \phi_2\left((\alpha_{ni}^{\ast}, \alpha_{nj}^{\ast}); (M_{ni}, M_{nj}), \begin{pmatrix}1 & R_{ij} \\ R_{ij} & 1\end{pmatrix}\right) d\alpha_{ni}^{\ast} d\alpha_{nj}^{\ast} \notag\\
  & \quad\quad - \int_{-\infty}^{0} \int_{-\infty}^{0} \phi_2\left((\alpha_{ni}^{\ast}, \alpha_{nj}^{\ast}); (M_{ni}, M_{nj}), \begin{pmatrix}1 & \widetilde{R}_{ij} \\ \widetilde{R}_{ij} & 1\end{pmatrix}\right) d\alpha_{ni}^{\ast} d\alpha_{nj}^{\ast} = 0 \label{identifmvpcorrcoef}.
\end{align}

\noindent This is equivalent to writing \(g(R_{ij}) - g(\widetilde{R}_{ij}) = 0\) for \(g: (-1, 1) \to \mathbb{R}\) where \(g(R_{ij})\) stands for \begin{equation}
  \int_{-\infty}^{0} \int_{-\infty}^{0} \phi_2\left((\alpha_{ni}^{\ast}, \alpha_{nj}^{\ast}); (M_{ni}, M_{nj}), \begin{pmatrix}1 & R_{ij} \\ R_{ij} & 1\end{pmatrix}\right) d\alpha_{ni}^{\ast} d\alpha_{nj}^{\ast}
\end{equation}

\noindent Observe that \begin{align}
  g(R_{ij}) & = \int_{-\infty}^{0} \int_{-\infty}^{0} \phi_2\left((\alpha_{ni}^{\ast} - M_{ni}, \alpha_{nj}^{\ast} - M_{nj}); (0, 0), \begin{pmatrix}1 & R_{ij} \\ R_{ij} & 1\end{pmatrix}\right) d\alpha_{ni}^{\ast} d\alpha_{nj}^{\ast} \notag\\
  & = \int_{-\infty}^{-M_{i2}} \int_{-\infty}^{-M_{j2}} \phi_2\left((Z_1, Z_2); (0, 0), \begin{pmatrix}1 & R_{ij} \\ R_{ij} & 1\end{pmatrix}\right) dZ_1 dZ_2
\end{align}

\noindent where we have let \(Z_1 = \alpha_{ni}^\ast - M_{ni}\), \(Z_2 = \alpha_{nj}^\ast - M_{nj}\). According to a result from \citet[pp. 107]{drezner1990computation}, \begin{equation}
  \frac{\partial}{\partial R_{ij}} g(R_{ij}) = \frac{1}{2\pi \sqrt{1 - R_{ij}^2}} \exp\left(\frac{-(M_{ni}^2 - 2 R_{ij} M_{ni} M_{nj} + M_{nj}^2)}{2(1 - R_{ij}^2)}\right).
\end{equation}

\noindent Therefore \(g\) is strictly increasing in \(R_{ij}\) on \((-1, 1)\), and thus has an inverse function on \(g((-1, 1))\). Observe that from \eqref{identifmvpcorrcoef} we have \(g(R_{ij}) = g(\widetilde{R}_{ij})\), so applying the inverse function on both sides we have \(R_{ij} = \widetilde{R}_{ij}\).

Since the above can be performed for any combination of \((i, j)\), we have that \(R = \widetilde{R}\).

We thus have that \(\{p(\alpha \mid \omega)\}\) is strictly identifiable.
\end{proof}

\subsubsection*{E.2.2 \quad Showing the Strict Identifiability of the Multivariate Probit for the Model's Actual Parameter Space}

We now establish that there is an isomorphism between \(\Theta_s\) and \(\Omega\). This is done by showing that there is an injective and surjective mapping from \(\mathcal{Z}\) onto \(\mathcal{M}\).

\begin{lemma}\label{lemmaaboutg}
If \(\text{rank}(W) = D + H_{\text{otr}}\), then the function \(g: \mathcal{Z} \to \mathcal{M}\) defined by \(g(\zeta) = W\zeta\) is injective and surjective.
\end{lemma}

\begin{proof}
  In the proof of this Lemma, for a matrix \(M\) we write \(M_k\) to denote the \(k\)th column of \(M\). Recall that \(W \in \mathbb{R}^{N \times (D + H_{\text{otr}})}\) and that \(N \geq D + H_{\text{otr}}\). Assume \(\text{rank}(W) = D + H_{\text{otr}}\) (i.e. \(W\) is full rank).

  First we show that that \(g\) is injective. Note that \(g(\zeta) = W\zeta = (W\zeta_1, \ldots, W\zeta_K)\). Define the function \(h: \mathbb{R}^{D + H_{\text{otr}}} \to \text{Im } W\) by \(v \mapsto Wv\), so \(g(\zeta) = (h(\zeta_1), \ldots, h(\zeta_K))\). We observe that \(g\) is injective if and only if for all \(k \in [K]\), \(h(\zeta_k) = h(\widetilde{\zeta}_k) \implies \zeta_k = \widetilde{\zeta}_k\). This would certainly hold if \(h\) itself were injective. Since by the rank-nullity theorem \citep[pp. 61]{lang1987linear} we have \(D + H_{\text{otr}} = \text{dim } \text{Ker } h + \text{dim } \text{Im } h = \text{dim } \text{Ker } h + \text{rank}(W)\), \(h\) is injective if and only if \(\text{rank}(W) = D + H_{\text{otr}}\), which it is by assumption. Therefore \(g\) is injective.

By the definition of \(g\), we have that \(g\) is surjective.
\end{proof}

\begin{lemma}
\label{identifisomorph}
If \(\text{rank}(W) = D + H_{\text{otr}}\), then \(\Theta_s\) and \(\Omega\) are isomorphic
\end{lemma}

\begin{proof}
Let \(f: \Theta_s \to \Omega\) be defined by the mapping \((a, b, c) \mapsto (g(a), b, c)\), where \(g\) is as defined in Lemma \ref{lemmaaboutg}. Clearly \(f\) is an isomorphism from \(\Theta_s\) to \(\Omega\).
\end{proof}

\begin{theorem}
\label{identiflikfinalident}
If \(\text{rank}(W) = D + H_{\text{otr}}\), then \(\{p(\alpha \mid \theta_s)\}\) is strictly identifiable.
\end{theorem}

\begin{proof}
  From Theorem \ref{identifmvp1} and Lemma \ref{identifisomorph}, we conclude that \(\{p(\alpha \mid \theta_s)\}\) is strictly identifiable.
\end{proof}

\section*{Supplementary Material F}

Supplementary Material F describes the model transformation.

\begin{align}
  p(Z) & = \underbrace{p(Y \mid Y^\ast, \kappa) \cdot p(Y^\ast \mid \beta, \alpha) \cdot p(\kappa) \cdot p(\beta \mid \delta) \cdot p(\delta \mid \omega) \cdot p(\omega)}_{(\text{part1})} \cdot \underbrace{p(\alpha \mid \alpha^\ast, \gamma)}_{(\text{part2})} \notag\\
  & \quad\quad \cdot \underbrace{p(\gamma \mid V)}_{(\text{part3})} \cdot \underbrace{p(R, V)}_{(\text{part4})} \cdot \underbrace{p(\alpha^\ast \mid R, \lambda, \xi)}_{(\text{part5})} \cdot \underbrace{p(\lambda, \xi \mid R)}_{(\text{part6})}.
\end{align}

We consider \(g^{-1}\), the inverse of the transformation. The value taken by \(g^{-1}\) at \(\widetilde{z}\) is \(z\). We write \(z = g^{-1}(\widetilde{z}) = (h_1(\widetilde{z}), \ldots, h_8(\widetilde{z}))\), where, writing \(\widetilde{V} = \text{diag}(\sigma_{11}, \ldots, \sigma_{KK})\) and denoting \(z = (z_1, \ldots, z_8)\), \begin{equation}
  \begin{array}{lll}
    z_1 & = & Y \\
    z_2 & = & \alpha \\
    z_3 & = & b_1 = (Y^\ast, \beta, \delta, \omega) \\
    z_4 & = & \alpha^\ast \\
    z_5 & = & \gamma \\
    z_6 & = & \zeta \\
    z_7 & = & (r_{12}, r_{13}, \ldots, r_{1K}, r_{23}, \ldots, r_{2K}, \ldots, r_{K-1,K}) \\
    z_8 & = & (v_1, \ldots, v_K).
  \end{array}
\end{equation}

\noindent Denoting \(\widetilde{z} = (\widetilde{z}_1, \ldots, \widetilde{z}_8)\), with \begin{equation}
  \begin{array}{lll}
    \widetilde{z}_1 & = & Y \\
    \widetilde{z}_2 & = & \alpha \\
    \widetilde{z}_3 & = & b_1 \\
    \widetilde{z}_4 & = & \widetilde{\alpha^\ast} \\
    \widetilde{z}_5 & = & \widetilde{\gamma} \\
    \widetilde{z}_6 & = & \widetilde{\zeta} \\
    \widetilde{z}_7 & = & (\sigma_{12}, \ldots, \sigma_{1K}, \sigma_{23}, \ldots, \sigma_{2K}, \ldots, \sigma_{K-1,K}) \\
    \widetilde{z}_8 & = & (\sigma_{11}, \ldots, \sigma_{KK}). \\  
  \end{array}
\end{equation}

we define \begin{equation}
  \begin{array}{lll}
    h_1(\widetilde{z}) & = & Y \\
    h_2(\widetilde{z}) & = & \alpha \\
    h_3(\widetilde{z}) & = & b_1 \\
    h_4(\widetilde{z}) & = & \widetilde{\alpha^\ast} \widetilde{V}^{-1/2} \\
    h_5(\widetilde{z}) & = & \widetilde{\gamma} \widetilde{V}^{-1/2} \\
    h_6(\widetilde{z}) & = & \widetilde{\zeta} \widetilde{V}^{-1/2} \\
    h_7(\widetilde{z}) & = & \left(\frac{\sigma_{12}}{\sigma_{11}^{1/2} \sigma_{22}^{1/2}}, \ldots, \frac{\sigma_{1K}}{\sigma_{11}^{1/2} \sigma_{KK}^{1/2}}, \frac{\sigma_{23}}{\sigma_{22}^{1/2} \sigma_{33}^{1/2}}, \ldots \frac{\sigma_{2K}}{\sigma_{22}^{1/2} \sigma_{KK}^{1/2}}, \cdots, \frac{\sigma_{K-1,K}}{\sigma_{K-1,K-1}^{1/2} \sigma_{KK}^{1/2}}\right) \\
    h_8(\widetilde{z}) & = & (\sigma_{11}, \ldots, \sigma_{KK}).
  \end{array}
\end{equation}

A derivation in \citet{wayman2025restricted} demonstrates (with \(\widetilde{\zeta}\) playing the role of \(\widetilde{\lambda}\) and with \(\alpha^\ast\) having a different number of rows, namely \(NT\)) that Jacobian determinant in the change of variables formula is \begin{align}
  J_{g^{-1}}(\widetilde{z}) = \underbrace{\left(\prod_{k \in [K]} \sigma_{kk}^{-1/2}\right)^{TN}}_{(\text{D1})} \cdot \underbrace{\left(\prod_{k \in [K]} \sigma_{kk}^{-1/2}\right)^{L - 2}}_{(\text{D2})} \cdot \underbrace{\left(\prod_{k \in [K]} \sigma_{kk}^{-1/2}\right)^{D + H_\text{otr}}}_{(\text{D3})} \cdot \underbrace{\left(\prod_{k \in [K]} \sigma_{kk}^{-1/2}\right)^{K - 1}}_{(\text{D4})}
\end{align}

For \(\widetilde{Z} = (Y, \alpha, b_1, \widetilde{\alpha^\ast}, \widetilde{\gamma}, \widetilde{\zeta}, \Sigma)\), where \(\widetilde{\zeta} = (\widetilde{\lambda}^\prime, \widetilde{\xi}^\prime)^\prime\), writing \(\text{etr}(\cdot)\) to mean \(\exp(\text{tr}(\cdot))\), we have \begin{equation}
  p(\widetilde{Z}) = (\widetilde{\text{part1}}) \cdot (\widetilde{\text{part2}}) \cdot (\widetilde{\text{part3}}) \cdot (\widetilde{\text{part4}}) \cdot (\widetilde{\text{part5}}) \cdot (\widetilde{\text{part6}})
\end{equation}

\noindent where \begin{align}
  (\widetilde{\text{part1}}) & = \prod_{n=1}^N \left[\prod_{t=1}^T \prod_{j=1}^J I(Y_{nj}^{\ast, t} \in (\kappa_{j, Y_{nj}^t - 1}, \kappa_{jY_{nj}^t}]) \cdot \phi(Y_{nj}^{\ast, t}; d_n^t \beta_j, 1)\right] \notag\\
    & \quad \cdot I(-\infty = \kappa_{j0} < 0 = \kappa_{j1} < \cdots < \kappa_{j M_j} = \infty) \notag\\
    & \quad \cdot \prod_{j=1}^J \Bigg[c_j(\delta_j) \cdot I(\beta_j \in \mathcal{R}_j) \notag\\
      & \quad\quad \cdot \left(\prod_{h=1}^H \left[I(\delta_{hj} = 0) \cdot \Delta(\beta_{hj}) + I(\delta_{hj} = 1) \cdot \phi(\beta_{hj}; 0, \sigma_{\beta}^2)\right]\right) \notag\\
      & \quad\quad \cdot \left(\prod_{h=1}^H \omega^{\delta_{hj}} (1 - \omega)^{1 - \delta_{hj}}\right)\Bigg] \notag\\
    & \quad \cdot \frac{1}{B(\omega_0, \omega_1)} \omega^{\omega_0 - 1} (1 - \omega)^{\omega_1 - 1} \label{part1}
\end{align}

\begin{align}
  \widetilde{\text{(part2)}} & = \prod_{n=1}^N \prod_{t=1}^T \prod_{k=1}^K I\left(\widetilde{\alpha_{nk}^{\ast, t}} \sigma_{kk}^{-1/2} \in \Big(\widetilde{\gamma}_{k, \alpha_{nk}^t} \sigma_{kk}^{-1/2}, \widetilde{\gamma}_{k, \alpha_{nk}^t + 1} \sigma_{kk}^{-1/2}\Big]\right) \notag\\
    & = \prod_{n=1}^N \prod_{t=1}^T \prod_{k=1}^K I\left(\widetilde{\alpha_{nk}^{\ast, t}} \in \Big(\widetilde{\gamma}_{k,\alpha_{nk}^t}, \widetilde{\gamma}_{k, \alpha_{nk}^t + 1}\Big]\right) \label{part2}
\end{align}

\begin{equation}
  (\widetilde{\text{part3}}) = \prod_{l=2}^{L - 1} \Bigg[a \exp\left[-a (\widetilde{\gamma}_{kl} - \widetilde{\gamma}_{k,l-1})\right] \cdot I\left(\widetilde{\gamma}_{kl} \in (\widetilde{\gamma}_{k, l-1}, \infty)\right)\Bigg] \label{part3}
\end{equation}

\begin{equation}
(\widetilde{\text{part4}}) = (\det{\Sigma})^{-\frac{1}{2}(v_0 + K + 1)} \exp\left(-\frac{1}{2}\tr(\Sigma^{-1})\right) \label{part4}
\end{equation}

\noindent Letting \((S)_{ij} = \sigma_{ij} / \sigma_{ii}^{1/2} \sigma_{jj}^{1/2}\), \begin{align}
  (\widetilde{\text{part5}}) & = (2\pi)^{-\frac{1}{2}KN} (\det{S})^{-\frac{1}{2}TN} \notag\\
  & \quad\quad \cdot \text{etr}\left\{-\frac{1}{2}\left(\widetilde{\alpha^\ast} \widetilde{V}^{-1/2} - W \widetilde{\zeta} \widetilde{V}^{-1/2}\right) S^{-1} \left(\widetilde{\alpha^\ast} \widetilde{V}^{-1/2} - W \widetilde{\zeta} \widetilde{V}^{-1/2}\right)^\prime\right\} \cdot (\text{D1}) \notag\\
  & = (2\pi)^{-\frac{1}{2}KN} (\det{S})^{-\frac{1}{2}TN} \text{etr}\left\{-\frac{1}{2}\left(\widetilde{\alpha^\ast} - W \widetilde{\zeta}\right) \Sigma^{-1} \left(\widetilde{\alpha^\ast} - W \widetilde{\zeta}\right)^\prime\right\} \cdot (\det{\widetilde{V}})^{-\frac{1}{2}TN} \notag\\
  & = (2\pi)^{-\frac{1}{2}KN} (\det{\Sigma})^{-\frac{1}{2}TN} \text{etr}\left\{-\frac{1}{2}\left(\widetilde{\alpha^\ast} - W \widetilde{\zeta}\right) \Sigma^{-1} \left(\widetilde{\alpha^\ast} - W \widetilde{\zeta}\right)^\prime\right\} \label{part5}
\end{align}

\begin{align}
  (\widetilde{\text{part6}}) & = (2\pi)^{-\frac{1}{2}(D + H_{\text{otr}})K} (\det{S})^{-\frac{1}{2}(D + H_{\text{otr}})} (\det{I_{D + H_{\text{otr}}}})^{-\frac{1}{2}K} \notag\\
  & \quad\quad \cdot \text{etr}\left\{-\frac{1}{2} \left(\widetilde{\zeta} \widetilde{V}^{-1/2}\right) S^{-1} \left(\widetilde{\zeta} \widetilde{V}^{-1/2}\right)^\prime\right\} \cdot (\text{D3}) \notag\\
  & = (2\pi)^{-\frac{1}{2}(D + H_{\text{otr}})K} (\det{S})^{-\frac{1}{2}(D + H_{\text{otr}})} (\det{I_{D + H_{\text{otr}}}})^{-\frac{1}{2}K} \notag\\
  & \quad\quad \cdot \text{etr}\left\{-\frac{1}{2} \widetilde{\zeta} \Sigma^{-1} \widetilde{\zeta}^\prime\right\} \cdot (\det{\widetilde{V}})^{-\frac{1}{2}(D + H_{\text{otr}})} \notag\\
  & = (2\pi)^{-\frac{1}{2}(D + H_{\text{otr}})K} (\det{\Sigma})^{-\frac{1}{2}(D + H_{\text{otr}})} (\det{I_{D + H_{\text{otr}}}})^{-\frac{1}{2}K} \cdot \text{etr}\left\{-\frac{1}{2} \widetilde{\zeta} \Sigma^{-1} \widetilde{\zeta}^\prime\right\} \label{part6}
\end{align}

\noindent where (\(\widetilde{\text{part3}}\)) and (\(\widetilde{\text{part4}}\)) were derived in previous work \citep{wayman2025restricted}.

\section*{Supplementary Material G}

Supplementary Material G describes the sampling steps.

\subsection*{Augmented Data and Measurement Model Thresholds}

For all \(j \in [J]\),  we sample \((\kappa_j, Y_j^\ast)\) using a Metropolis step established for cumulative-link models \citep{cowles1996accelerating} that results in \(\kappa_j\) converging faster than would be the case if Gibbs steps were used for these variables.

\subsection*{Beta and Sparsity Matrix}

For each \(j \in [J]\) and each \(h \in [H]\) we sample \(\delta_{hj}\) using a Gibbs step collapsed on \(\beta_{hj}\), the density for which is a Bernoulli: \begin{align}
& p\left(\delta_{hj} = 1 \mid \alpha, \beta_{(h)j}, Y_j^\ast, \omega\right) = \Bigg[(1 - \omega) + \omega \cdot \left(\Phi\left(\frac{-L_{hj}}{\sigma_{\beta}}\right)\right)^{-1} \cdot \left(\frac{c_2^2}{\sigma_{\beta}^2}\right)^{1/2} \notag\\
    & \quad\quad \cdot \exp\left(\frac{c_1^2}{2c_2^2}\right) \cdot \Phi\left(\frac{-(L_{hj} - c_1)}{c_2}\right)\Bigg]^{-1} \notag\\
& \quad\quad \cdot \omega \cdot \left(\Phi\left(\frac{-L_{hj}}{\sigma_{\beta}}\right)\right)^{-1} \cdot \left(\frac{c_2^2}{\sigma_{\beta}^2}\right)^{1/2} \cdot \exp\left(\frac{c_1^2}{2c_2^2}\right) \cdot \Phi\left(\frac{-(L_{hj} - c_1)}{c_2}\right).\label{deltaprob}
\end{align}

\noindent In \eqref{deltaprob}, \begin{equation*}
  c_2^2 = \left[\left(d^\prime d\right)_{hh} + \frac{1}{\sigma_{\beta}^2}\right]^{-1}
\end{equation*}

\noindent where \(\left(d^\prime d\right)_{hh}\) refers to the entry in row \(h\) and column \(h\) of the \(H \times H\) matrix \(d^\prime d\). Also in \eqref{deltaprob}, \begin{equation*}
  c_1 = c_2^2 \cdot \left(d^\prime Y_j^{\ast} - \left(d^\prime d\right)_{(h)} \beta_{(h)j}\right)_h
\end{equation*}

\noindent is the entry in row \(h\) of the \(H \times 1\) vector resulting from the calculation, where \(\left(d^\prime d\right)_{(h)}\) refers to \(d^\prime d\) with column \(h\) eliminated and where \(\beta_{(h)j}\) refers to the column vector \(\beta_j\) with element \(h\) eliminated.

We then use \(\delta_{hj}\) to sample \(\beta_{hj}\) from its full conditional, \(p(\beta_{hj} \mid \delta_j, Y_j^\ast, \alpha)\), which is a point mass at \(\beta_{hj} = 0\) when \(\delta_{hj} = 0\), and when \(\delta_{hj} = 1\), the density is \begin{equation}
  p(\beta_{hj} \mid \delta_j, \beta_{(h)j}, Y_j^\ast, \alpha) = I\left(\beta_{hj} \in (L_{hj}, \infty)\right) \frac{\phi(\beta_{hj}; c_1, c_2^2)}{[1 - \Phi(L_{hj}; c_1, c_2^2)]}
\end{equation}

\noindent a left-truncated normal whose left-truncation point is \citep{wayman2025restricted} \begin{equation}
  L_{hj} = \max_{u, v: u, v \in A_L \,\land\, u \geq v} -\left(d_{(1,h)u} - d_{(1,h)v}\right) \beta_{(1,h)j}
\end{equation}

\noindent and whose underlying mean and variance are \(c_1\) and \(c_2^2\) respectively (the notation \(\beta_{(1,h)j}\) refers to vector \(\beta_j\) with elements \(1\) and \(h\) removed). When \(h = 0\), \(L_{hj} = -\infty\) and the density is that of a normal distribution.

\subsection*{Latent States and Related Auxiliary Variables}

For each \(t \in [T]\), \(n \in [N]\), and \(k \in [K]\), we sample \((\widetilde{\alpha_{nk}^{\ast, t}}, \alpha_{nk}^t)\) by first sampling \(\alpha_{nk}^t\) using a Gibbs step collapsed on \(\widetilde{\alpha_{nk}^{\ast, t}}\), and then using \(\alpha_{nk}^t\) to sample \(\widetilde{\alpha_{nk}^{\ast, t}}\) from its full conditional.

We find the sampling density for this first step by finding the full conditional of \((\widetilde{\alpha_{nk}^{\ast, t}}, \alpha_{nk}^t)\) and integrating with respect to \(\widetilde{\alpha_{nk}^{\ast, t}}\). This depends on the particular value of \(t\). For \(t \in \{1, 2, \ldots, T-1\}\), the full conditional of \((\widetilde{\alpha_{nk}^{\ast, t}}, \alpha_{nk}^t)\) is \begin{equation}\label{bg3pt1}
  p\left(\widetilde{\alpha_{nk}^{\ast, t}}, \alpha_{nk}^t \mid Y_n^{\ast, t}, \alpha_{n(k)}^t, \beta, \gamma_k, \alpha_{nk}^{t-1}, \widetilde{\alpha_{n(k)}^{\ast, t}}, \widetilde{\alpha_n^{\ast, t + 1}}, \widetilde{\zeta}, \Sigma\right)
\end{equation}

For \(t = T\), the full conditional of \((\widetilde{\alpha_{nk}^{\ast, t}}, \alpha_{nk}^t)\) is \begin{equation}\label{bg3pt1pt5}
  p\left(\widetilde{\alpha_{nk}^{\ast, t}}, \alpha_{nk}^t \mid Y_n^{\ast, t}, \alpha_{n(k)}^t, \beta, \gamma_k, \alpha_{nk}^{t-1}, \widetilde{\alpha_{n(k)}^{\ast, t}}, \widetilde{\zeta}, \Sigma\right)
\end{equation}

\noindent Both of these are proportional to quantities appearing in (\(\widetilde{\text{part1}}\)), (\(\widetilde{\text{part2}}\)), and (\(\widetilde{\text{part5}}\)). Note that \eqref{part5} is the density of a matrix variate normal with variable \(\widetilde{\alpha}\), mean \(W \widetilde{\zeta}\), and covariance \(I_{TN} \otimes \Sigma\). Observe that the density of this matrix variate normal can be factored as follows: \begin{align}
  (\text{part5}) & = c_1 \cdot \prod_{n=1}^N \Bigg(\Bigg\{\prod_{t=2}^T (\det{\Sigma})^{-\frac{1}{2}} \text{etr}\left[-\frac{1}{2} \left(\widetilde{\alpha_n^{\ast, t}} - W_n^t \widetilde{\zeta}\right) \Sigma^{-1} \left(\widetilde{\alpha_n^{\ast, t}} - W_n^t \widetilde{\zeta}\right)^\prime\right]\Bigg\} \notag\\
  & \qquad\qquad\qquad \cdot (\det{\Sigma})^{-\frac{1}{2}} \text{etr}\left[-\frac{1}{2} \left(\widetilde{\alpha_n^{\ast, 1}} - X_n^1 \widetilde{\lambda}\right) \Sigma^{-1} \left(\widetilde{\alpha_n^{\ast, 1}} - X_n^1 \widetilde{\lambda}\right)^\prime\right]\Bigg)\label{bg3bpt2}
\end{align}

\noindent Making use of \eqref{bg3bpt2} and using proportionality, for \(t \in \{1, 2, \ldots, T-1\}\), \begin{align}
  & p\left(\widetilde{\alpha_{nk}^{\ast, t}}, \alpha_{nk}^t \mid Y_n^{\ast, t}, \alpha_{n(k)}^t, \beta, \gamma_k, \alpha_{nk}^{t-1}, \widetilde{\alpha_{n(k)}^{\ast, t}}, \widetilde{\alpha_n^{\ast, t + 1}}, \widetilde{\zeta}, \Sigma\right) \notag\\
  & = c_1 \cdot \left[\prod_{j=1}^J \phi(Y_{nj}^{\ast, t}; d_n^t \beta_j, 1)\right] \cdot I\left(\widetilde{\alpha_{nk}^{\ast, t}} \in \Big(\widetilde{\gamma}_{k,\alpha_{nk}^t}, \widetilde{\gamma}_{k, \alpha_{nk}^t + 1}\Big]\right) \notag\\
    & \quad\quad \cdot \underbrace{(\det{\Sigma})^{-\frac{1}{2}} \text{etr}\left[-\frac{1}{2} \left(\widetilde{\alpha_n^{\ast, t}} - W_n^t \widetilde{\zeta}\right) \Sigma^{-1} \left(\widetilde{\alpha_n^{\ast, t}} - W_n^t \widetilde{\zeta}\right)^\prime\right]}_{(1)} \notag\\
    & \quad\quad \cdot \underbrace{(\det{\Sigma})^{-\frac{1}{2}} \text{etr}\left[-\frac{1}{2} \left(\widetilde{\alpha_n^{\ast, t + 1}} - W_n^{t + 1} \widetilde{\zeta}\right) \Sigma^{-1} \left(\widetilde{\alpha_n^{\ast, t + 1}} - W_n^{t + 1} \widetilde{\zeta}\right)^\prime\right]}_{(2)} \notag\\
    & = c_2 \cdot \left[\prod_{j=1}^J \phi(Y_{nj}^{\ast, t}; d_n^t \beta_j, 1)\right] \cdot I\left(\widetilde{\alpha_{nk}^{\ast, t}} \in \Big(\widetilde{\gamma}_{k,\alpha_{nk}^t}, \widetilde{\gamma}_{k, \alpha_{nk}^t + 1}\Big]\right) \notag\\
    & \quad\quad \cdot \phi\left(\widetilde{\alpha_{nk}^{\ast, t}}; \mu_{nk}^t, \sigma_k^2\right) \cdot \phi_K\left(\widetilde{\alpha_n^{\ast, t + 1}}; W_n^{t + 1} \widetilde{\zeta}, \Sigma\right) \label{bg3bpt3}
\end{align}

\noindent where we have used the fact that since (1) is the density of a multivariate normal with variable \(\widetilde{\alpha_n^{\ast, t}}\), this density can be written \citep{marden2015multivariate} as the product of two densities, one which only involves \(\widetilde{\alpha_{n(k)}^{\ast, t}}\), and one of which is the density of a multivariate normal with variable \(\widetilde{\alpha_{nk}^{\ast, t}}\), mean \(\mu_k^t\), and covariance \(\sigma_k^2\), where \(\mu_{nk}^t = W_n^t \widetilde{\zeta}_k + (\widetilde{\alpha_{n(k)}^{\ast, t}} - W_n^t \widetilde{\zeta}_{(k)})\Sigma_{(k)(k)}^{-1} \Sigma_{(k)k}\) and \(\sigma_k^2 = \Sigma_{kk} - \Sigma_{k(k)} \Sigma_{(k)(k)}^{-1} \Sigma_{(k)k}\). We note that (2) is the density of a multivariate normal with variable, \(\widetilde{\alpha_n^{\ast, t + 1}}\), mean \(W_n^{t + 1} \widetilde{\zeta}\), and covariance \(\Sigma\).

Similarly, for \(t = T\), \begin{align}
  & p\left(\widetilde{\alpha_{nk}^{\ast, t}}, \alpha_{nk}^t \mid Y_n^{\ast, t}, \alpha_{n(k)}^t, \beta, \gamma_k, \alpha_{nk}^{t-1}, \widetilde{\alpha_{n(k)}^{\ast, t}}, \widetilde{\zeta}, \Sigma\right) \notag\\
  & = c_3 \cdot \left[\prod_{j=1}^J \phi(Y_{nj}^{\ast, t}; d_n^t \beta_j, 1)\right] \cdot I\left(\widetilde{\alpha_{nk}^{\ast, t}} \in \Big(\widetilde{\gamma}_{k,\alpha_{nk}^t}, \widetilde{\gamma}_{k, \alpha_{nk}^t + 1}\Big]\right) \notag\\
  & \quad\quad \cdot \underbrace{(\det{\Sigma})^{-\frac{1}{2}} \text{etr}\left[-\frac{1}{2} \left(\widetilde{\alpha_n^{\ast, t}} - W_n^t \widetilde{\zeta}\right) \Sigma^{-1} \left(\widetilde{\alpha_n^{\ast, t}} - W_n^t \widetilde{\zeta}\right)^\prime\right]}_{(1)} \notag\\
  & = c_4 \cdot \left[\prod_{j=1}^J \phi(Y_{nj}^{\ast, t}; d_n^t \beta_j, 1)\right] \cdot I\left(\widetilde{\alpha_{nk}^{\ast, t}} \in \Big(\widetilde{\gamma}_{k,\alpha_{nk}^t}, \widetilde{\gamma}_{k, \alpha_{nk}^t + 1}\Big]\right) \cdot \phi\left(\widetilde{\alpha_{nk}^{\ast, t}}; \mu_{nk}^t, \sigma_k^2\right).
\end{align}

Taking the integral, for \(t \in \{1, 2, \ldots, T - 1\}\) we have \begin{align}
  & p\left(\alpha_{nk}^t \mid Y_n^{\ast, t}, \alpha_{n(k)}^t, \beta, \gamma_k, \alpha_{nk}^{t-1}, \widetilde{\alpha_{n(k)}^{\ast, t}}, \widetilde{\alpha_n^{\ast, t + 1}}, \widetilde{\zeta}, \Sigma\right) \notag\\
  & = \int p(\widetilde{\alpha_{nk}^{\ast, t}}, \alpha_{nk}^t \mid Y_n^{\ast, t}, \alpha_{n(k)}^t, \beta, \gamma_k, \alpha_{nk}^{t-1}, \widetilde{\alpha_{n(k)}^{\ast, t}}, \widetilde{\alpha_n^{\ast, t + 1}}, \widetilde{\zeta}, \Sigma) d\widetilde{\alpha_{nk}^{\ast, t}} \notag\\
  & = c_2 \cdot \left[\prod_{j=1}^J \phi(Y_{nj}^{\ast, t}; d_n^t \beta_j, 1)\right] \cdot \phi_K\left(\widetilde{\alpha_n^{\ast, t + 1}}; W_n^{t + 1} \widetilde{\zeta}, \Sigma\right) \notag\\
  & \quad\quad \cdot \int I\left(\widetilde{\alpha_{nk}^{\ast, t}} \in \Big(\widetilde{\gamma}_{k,\alpha_{nk}^t}, \widetilde{\gamma}_{k, \alpha_{nk}^t + 1}\Big]\right) \cdot \phi\left(\widetilde{\alpha_n^{\ast, t}}; \mu_{nk}^t, \sigma_k^2\right) d\widetilde{\alpha_{nk}^{\ast, t}} \notag\\
  & = c_2 \cdot \left[\prod_{j=1}^J \phi(Y_{nj}^{\ast, t}; d_n^t \beta_j, 1)\right] \cdot \phi_K\left(\widetilde{\alpha_n^{\ast, t + 1}}; W_n^{t + 1} \widetilde{\zeta}, \Sigma\right) \notag\\
  & \quad\quad \cdot \int_{\widetilde{\gamma}_{k,\alpha_{nk}^t}}^{\widetilde{\gamma}_{k, \alpha_{nk}^t + 1}} \phi\left(\widetilde{\alpha_n^{\ast, t}}; \mu_{nk}^t, \sigma_k^2\right) d\widetilde{\alpha_{nk}^{\ast, t}} \notag\\
    & = c_2 \cdot \left[\prod_{j=1}^J \phi(Y_{nj}^{\ast, t}; d_n^t \beta_j, 1)\right] \cdot \phi_K\left(\widetilde{\alpha_n^{\ast, t + 1}}; W_n^{t + 1} \widetilde{\zeta}, \Sigma\right) \notag\\
  & \quad\quad \cdot \left[\Phi\left(\frac{\widetilde{\gamma}_{k, \alpha_{nk}^t + 1} - \mu_{nk}^t}{\sigma_k}\right) - \Phi\left(\frac{\widetilde{\gamma}_{k, \alpha_{nk}^t} - \mu_{nk}^t}{\sigma_k}\right)\right] \label{bg3bpt5}
\end{align}

and similarly for \(t = T\), \begin{align}
  & p\left(\alpha_{nk}^t \mid Y_n^{\ast, t}, \alpha_{n(k)}^t, \beta, \gamma_k, \alpha_{nk}^{t-1}, \widetilde{\alpha_{n(k)}^{\ast, t}}, \widetilde{\zeta}, \Sigma\right) \notag\\
  & = c_4 \cdot \left[\prod_{j=1}^J \phi(Y_{nj}^{\ast, t}; d_n^t \beta_j, 1)\right] \cdot \left[\Phi\left(\frac{\widetilde{\gamma}_{k, \alpha_{nk}^t + 1} - \mu_{nk}^t}{\sigma_k}\right) - \Phi\left(\frac{\widetilde{\gamma}_{k, \alpha_{nk}^t} - \mu_{nk}^t}{\sigma_k}\right)\right].
\end{align}

To calculate the probability for each of \(\alpha_{nk}^t = l \in \{0, 1, \ldots, L - 1\}\), for \(t \in \{1, 2, \ldots, T - 1\}\) we plug \(l\) into \begin{align}
  p_l & := \left[\prod_{j=1}^J \phi(Y_{nj}^{\ast, t}; d_n^t \beta_j, 1)\right] \cdot \phi_K\left(\widetilde{\alpha_n^{\ast, t + 1}}; W_n^{t + 1} \widetilde{\zeta}, \Sigma\right) \notag\\
  & \quad\quad \cdot \left[\Phi\left(\frac{\widetilde{\gamma}_{k, \alpha_{nk}^t + 1} - \mu_{nk}^t}{\sigma_k}\right) - \Phi\left(\frac{\widetilde{\gamma}_{k, \alpha_{nk}^t} - \mu_{nk}^t}{\sigma_k}\right)\right]
\end{align}

\noindent and then calculate \(c = (\sum_{l=0}^{L - 1} p_l)^{-1}\), so then \begin{equation}\label{probalphatcateg1}
  p\left(\alpha_{nk}^t \mid Y_n^{\ast, t}, \alpha_{n(k)}^t, \beta, \gamma_k, \alpha_{nk}^{t-1}, \widetilde{\alpha_{n(k)}^{\ast, t}}, \widetilde{\alpha_n^{\ast, t + 1}}, \widetilde{\zeta}, \Sigma\right) = c \cdot p_{\alpha_{nk}^t}.
\end{equation}

\noindent For \(t = T\), we plug \(l\) into \begin{equation}
  p_l := \left[\prod_{j=1}^J \phi(Y_{nj}^{\ast, t}; d_n^t \beta_j, 1)\right] \cdot \left[\Phi\left(\frac{\widetilde{\gamma}_{k, \alpha_{nk}^t + 1} - \mu_{nk}^t}{\sigma_k}\right) - \Phi\left(\frac{\widetilde{\gamma}_{k, \alpha_{nk}^t} - \mu_{nk}^t}{\sigma_k}\right)\right]
\end{equation}

\noindent and calculate \(c = (\sum_{l=0}^{L - 1} p_l)^{-1}\), so that \begin{equation}\label{probalphatcateg2}
  p\left(\alpha_{nk}^t \mid Y_n^{\ast, t}, \alpha_{n(k)}^t, \beta, \gamma_k, \alpha_{nk}^{t-1}, \widetilde{\alpha_{n(k)}^{\ast, t}}, \widetilde{\zeta}, \Sigma\right) = c \cdot p_{\alpha_{nk}^t}.
\end{equation}

The full conditional of \(\widetilde{\alpha_{nk}^{\ast, t}}\) is, for all \(t \in \{1, 2, \ldots, T\}\), \begin{equation}
  p\left(\widetilde{\alpha_{nk}^{\ast, t}} \mid \alpha_{nk}^t, \widetilde{\gamma}_k, \alpha_n^{t-1}, \widetilde{\alpha_{n(k)}^{\ast, t}}, \widetilde{\zeta}, \Sigma\right) = c_5 \cdot I\left(\widetilde{\alpha_{nk}^{\ast, t}} \in \Big(\widetilde{\gamma}_{k,\alpha_{nk}^t}, \widetilde{\gamma}_{k, \alpha_{nk}^t + 1}\Big]\right) \cdot \phi\left(\widetilde{\alpha_{nk}^{\ast, t}}; \mu_{nk}^t, \sigma_k^2\right)
\end{equation}

\noindent so we conclude that \begin{align}\label{probalphat}
  & p\left(\widetilde{\alpha_{nk}^{\ast, t}} \mid \alpha_{nk}^t, \widetilde{\gamma}_k, \alpha_n^{t-1}, \widetilde{\alpha_{n(k)}^{\ast, t}}, \widetilde{\zeta}, \Sigma\right) \notag\\
  & = I\left(\widetilde{\alpha_{nk}^{\ast, t}} \in (\widetilde{\gamma}_{k,\alpha_{nk}^t}, \widetilde{\gamma}_{k, \alpha_{nk}^t + 1})\right) \frac{\phi(\widetilde{\alpha_{nk}^{\ast, t}}; \mu_{nk}^t, \sigma_k^2)}{\Phi(\widetilde{\gamma}_{k,\alpha_{nk}^t + 1}; \mu_{nk}^t, \sigma_k^2) - \Phi(\widetilde{\gamma}_{k\alpha_{nk}^t}; \mu_{nk}^t, \sigma_k^2)}
\end{align}

\noindent a truncated normal with left and right truncation points \(\widetilde{\gamma}_{k,\alpha_{nk}^t}\) and \(\widetilde{\gamma}_{k, \alpha_{nk}^t + 1}\) respectively, and where the mean and variance of the underlying normal distribution are \(\mu_{nk}^t\) and \(\sigma_k^2\) respectively.

\subsection*{Thresholds for Latent State Levels}

For each \(k \in [K]\), each threshold \(\widetilde{\gamma}_{kl}\) where \(l \in \{2, 3, \ldots, L - 1\}\) is sampled from its full conditional; these densities are derived in the paper that introduced the cross-sectional model we are extending \cite{wayman2025restricted}. For \(l \in \{2, 3, \ldots, L - 2\}\), the full conditional \(p(\widetilde{\gamma}_{kl} \mid \widetilde{\gamma}_{k,l-1}, \widetilde{\gamma}_{k,l+1}, \widetilde{\alpha^\ast})\) is a continuous uniform distribution on the range \begin{equation}
  \left(\max\left(\max_{n \in [N]:\, \alpha_{nk} = l - 1} \left(\widetilde{\alpha^{\ast}}_{nk}\right), \widetilde{\gamma}_{k,l-1}\right), \min\left(\min_{n \in [N]:\, \alpha_{nk} = l} \left(\widetilde{\alpha^{\ast}}_{nk}\right), \widetilde{\gamma}_{k,l+1}\right)\right).
\end{equation}

\noindent For  \(l = L - 1\), the full conditional for \(\widetilde{\gamma}_{kl}\) is a left-truncated exponential: \begin{align}
  p(\widetilde{\gamma}_{kl} \mid \widetilde{\gamma}_{k,l-1}, \widetilde{\alpha^\ast}) & = c \cdot I\left(\widetilde{\gamma}_{kl} \geq \max\left(\max_{n \in [N]:\, \alpha_{nk} = l - 1} \left(\widetilde{\alpha^{\ast}}_{nk}\right), \widetilde{\gamma}_{k,l-1}\right)\right) \notag\\
  & \quad\quad \cdot I\left(\widetilde{\gamma}_{kl} < \min\left(\min_{n \in [N]:\, \alpha_{nk} = l} \left(\widetilde{\alpha^{\ast}}_{nk}\right), \infty\right)\right) \notag\\
  & \quad\quad \cdot \exp\left(-a \widetilde{\gamma}_{kl}\right).
\end{align}

\subsection*{Covariance Matrix and Slope Parameter for Covariates}

We sample \((\widetilde{\zeta}, \Sigma)\) by first sampling \(\widetilde{\Sigma}\) using a Gibbs step collapsed on \(\widetilde{\zeta}\) and then using \(\widetilde{\Sigma}\) to sample \(\widetilde{\zeta}\) from its full conditional. To find that first sampling density, we find the full conditional of \((\widetilde{\zeta}, \Sigma)\) and integrate with respect to \(\widetilde{\zeta}\).

We observe that the full conditional of \((\widetilde{\zeta}, \Sigma)\) is \begin{align}
  p(\widetilde{\zeta}, \Sigma \mid \alpha^\ast) & = c_1 \cdot (\text{part3}) \cdot (\text{part5}) \cdot (\text{part6}) \notag\\
& = c_2 \cdot (\det{\Sigma})^{-\frac{1}{2}NT} \text{etr}\left[-\frac{1}{2}\left(\widetilde{\alpha^\ast} - W \widetilde{\zeta}\right) \Sigma^{-1} \left(\widetilde{\alpha^\ast} - W \widetilde{\zeta}\right)^\prime\right] \cdot (\det{\Sigma})^{-\frac{1}{2}(D + H_{\text{otr}})} \notag\\
& \quad\quad \cdot \text{etr}\left[-\frac{1}{2} \widetilde{\zeta} \Sigma^{-1} \widetilde{\zeta}^\prime\right] \cdot (\det{\Sigma})^{-(v_0 + K + 1)/2} \cdot \text{etr}\left[-\frac{1}{2}\Sigma^{-1}\right]
\end{align}

Let \(\Xi = (\widetilde{\alpha^\ast}^\prime, O^\prime)^\prime\) and \(\Omega = (W^\prime, I_{D + H_{\text{otr}}}^\prime)^\prime\). Note that \(\Omega^\prime \Omega = W^\prime W + I_{D + H_{\text{otr}}}\) and \(\Omega^\prime \Xi = W^\prime \widetilde{\alpha^\ast}\). Define \(\widehat{L_2} = (W^\prime W + I_{D + H_{\text{otr}}})^{-1} W^\prime \widetilde{\alpha^\ast}\) and \(S = (\widetilde{\alpha^\ast} - W \widehat{L_2})^\prime (\widetilde{\alpha^\ast} - W \widehat{L_2}) + \widehat{L_2}^\prime I_{D + H_{\text{otr}}} \widehat{L_2}\). From a derivation of Bayesian multiple linear regression \citep{rossi2005bayesian,wayman2025restricted}, we have that \begin{equation}
(\widetilde{\alpha^{\ast}} - W \widetilde{\zeta})^\prime (\widetilde{\alpha^{\ast}} - W \widetilde{\zeta}) = S + (\widetilde{\zeta} - \widehat{L_2})^\prime \Omega^\prime \Omega (\widetilde{\zeta} - \widehat{L_2}) - \widetilde{\zeta}^\prime I_{D + H_{\text{otr}}} \widetilde{\zeta}
\end{equation}

\noindent Therefore we can write \begin{align}
  & p(\widetilde{\zeta}, \Sigma \mid \widetilde{\alpha^\ast}) \notag\\
  & = c_3 \cdot (\det{\Sigma})^{-\frac{1}{2}(NT + D + H_{\text{otr}} + v_0 + K + 1)} \notag\\
  & \quad\quad \cdot \text{etr}\left[-\frac{1}{2} \Sigma^{-1} \left\{S + (\widetilde{\zeta} - \widehat{L_2})^\prime \Omega^\prime \Omega (\widetilde{\zeta} - \widehat{L_2}) - \widetilde{\zeta}^\prime I_{D + H_{\text{otr}}} \widetilde{\zeta}\right\} \right] \notag\\
& \quad\quad \cdot \text{etr}\left[-\frac{1}{2}\Sigma^{-1} \widetilde{\zeta}^\prime I_{D + H_{\text{otr}}} \widetilde{\zeta}\right] \cdot \text{etr}\left[-\frac{1}{2} \Sigma^{-1}\right] \notag\\
& = c_3 \cdot (\det{\Sigma})^{-\frac{1}{2}(NT + D + H_{\text{otr}} + v_0 + K + 1)} \cdot \text{etr}\left[-\frac{1}{2} \Sigma^{-1} \left\{S + (\widetilde{\zeta} - \widehat{L_2})^\prime \Omega^\prime \Omega (\widetilde{\zeta} - \widehat{L_2})\right\}\right] \notag\\
  & \quad\quad \cdot \text{etr}\left[-\frac{1}{2} \Sigma^{-1}\right] \notag\\
  & = c_3 \cdot (\det{\Sigma})^{-\frac{1}{2}(NT + D + H_{\text{otr}} + v_0 + K + 1)} \cdot \text{etr}\left[-\frac{1}{2} \Sigma^{-1} \left\{(\widetilde{\zeta} - \widehat{L_2})^\prime \Omega^\prime \Omega (\widetilde{\zeta} - \widehat{L_2})\right\}\right] \notag\\
  & \quad\quad \cdot \text{etr}\left[-\frac{1}{2} \Sigma^{-1} (I_K + S)\right] \notag\\
& = c_3 \cdot (\det{\Sigma})^{-\frac{1}{2}(NT + v_0 + K + 1)} \cdot \text{etr}\left[-\frac{1}{2} \Sigma^{-1} (I_K + S)\right] \cdot (\det{\Sigma})^{-\frac{1}{2}(D + H_{\text{otr}})} \notag\\
  & \quad\quad \cdot \text{etr}\left[-\frac{1}{2} \Sigma^{-1} \left\{(\widetilde{\zeta} - \widehat{L_2})^\prime \Omega^\prime \Omega (\widetilde{\zeta} - \widehat{L_2})\right\}\right]
\end{align}

\noindent Thus our first step samples \(\Sigma\) from \begin{align}
& p(\Sigma \mid \widetilde{\alpha^\ast}) \notag\\
& = c_3 \cdot (\det{\Sigma})^{-\frac{1}{2}(NT + v_0 + K + 1)} \cdot \text{etr}\left[-\frac{1}{2} \Sigma^{-1} (I_K + S)\right] \notag\\
  & \quad\quad \cdot \int (\det{\Sigma})^{-\frac{1}{2}(D + H_{\text{otr}})} \cdot \text{etr}\left[-\frac{1}{2} \Sigma^{-1} \left\{(\widetilde{\zeta} - \widehat{L_2})^\prime \Omega^\prime \Omega (\widetilde{\zeta} - \widehat{L_2})\right\}\right] d\widetilde{\zeta} \notag\\
& = c_4 \cdot (\det{\Sigma})^{-\frac{1}{2}(NT + v_0 + K + 1)} \cdot \text{etr}\left[-\frac{1}{2} \Sigma^{-1} (I_K + S)\right] \label{samplsigma}
\end{align}

\noindent the density of an inverse Wishart distribution with matrix parameter \(I_K + S\) and scalar parameter \(NT + v_0\).

Utilizing the above algebraic manipulations and proportionality, we find that the full conditional of \(\widetilde{\zeta}\) is \begin{equation}
  \widetilde{\zeta} \mid \Sigma, \widetilde{\alpha^\ast} \sim N_{D + H_{\text{otr}}, K}\left((W^\prime W + I_{D + H_{\text{otr}}})^{-1} W^\prime \widetilde{\alpha^{\ast}}; (W^\prime W + I_{D + H_{\text{otr}}})^{-1} \otimes \Sigma\right). \label{samplzetatilde}
\end{equation}

\subsection*{Sparsity Matrix Related Parameter}

We sample \(\omega\) from its full conditional, which is \begin{equation}
  \omega \mid \delta \sim \text{Beta}\left(\sum_{j \in [J], h \in [H]} \delta_{hj} + \omega_0,\, HJ - \sum_{j \in [J], h \in [H]} \delta_{hj} + \omega_1\right).
\end{equation}

\section*{Supplementary Material H}

Supplementary Material H displays the simulation results.

\begin{table}
\small
\begin{center}
\caption{\label{tab:sim01part01} Parameter recovery, simulation study one}
\vspace{0.5\baselineskip}
\begin{threeparttable}[t]
\begin{tabular}{rrrlrlrrrr}
\toprule
\(N\) & \(J\) & \(K\) & \(L\) & \(\rho\) & \(\gamma\) & \(\eta\) & \(R\) & \(\lambda\) & \(\xi\) \\
\midrule
250 & 15 & 2 & 2 & 0.000 &  & 0.012 & 0.026 & 0.070 & 0.140 \\
250 & 15 & 2 & 3 & 0.000 & 0.053 & 0.039 & 0.104 & 0.173 & 0.133 \\
250 & 25 & 3 & 2 & 0.000 &  & 0.011 & 0.044 & 0.116 & 0.065 \\
250 & 25 & 3 & 3 & 0.000 & 0.055 & 0.013 & 0.032 & 0.063 & 0.104 \\
250 & 45 & 4 & 2 & 0.000 &  & 0.011 & 0.032 & 0.071 & 0.101 \\
250 & 15 & 2 & 2 & 0.250 &  & 0.012 & 0.047 & 0.106 & 0.109 \\
250 & 15 & 2 & 3 & 0.250 & 0.039 & 0.013 & 0.022 & 0.079 & 0.113 \\
250 & 25 & 3 & 2 & 0.250 &  & 0.012 & 0.044 & 0.081 & 0.108 \\
250 & 25 & 3 & 3 & 0.250 & 0.053 & 0.014 & 0.034 & 0.075 & 0.137 \\
250 & 45 & 4 & 2 & 0.250 &  & 0.011 & 0.031 & 0.064 & 0.103 \\
250 & 15 & 2 & 2 & 0.500 &  & 0.012 & 0.045 & 0.091 & 0.178 \\
250 & 15 & 2 & 3 & 0.500 & 0.035 & 0.015 & 0.018 & 0.070 & 0.089 \\
250 & 25 & 3 & 2 & 0.500 &  & 0.050 & 0.245 & 0.226 & 0.108 \\
250 & 25 & 3 & 3 & 0.500 & 0.342 & 0.081 & 0.280 & 0.169 & 0.113 \\
250 & 45 & 4 & 2 & 0.500 &  & 0.012 & 0.025 & 0.127 & 0.117 \\
500 & 15 & 2 & 2 & 0.000 &  & 0.008 & 0.017 & 0.064 & 0.076 \\
500 & 15 & 2 & 3 & 0.000 & 0.053 & 0.025 & 0.057 & 0.121 & 0.079 \\
500 & 25 & 3 & 2 & 0.000 &  & 0.008 & 0.027 & 0.059 & 0.078 \\
500 & 25 & 3 & 3 & 0.000 & 0.034 & 0.012 & 0.024 & 0.066 & 0.075 \\
500 & 45 & 4 & 2 & 0.000 &  & 0.009 & 0.029 & 0.065 & 0.076 \\
500 & 15 & 2 & 2 & 0.250 &  & 0.008 & 0.016 & 0.059 & 0.073 \\
500 & 15 & 2 & 3 & 0.250 & 0.037 & 0.009 & 0.013 & 0.060 & 0.077 \\
500 & 25 & 3 & 2 & 0.250 &  & 0.008 & 0.022 & 0.061 & 0.072 \\
500 & 25 & 3 & 3 & 0.250 & 0.076 & 0.030 & 0.063 & 0.085 & 0.087 \\
500 & 45 & 4 & 2 & 0.250 &  & 0.009 & 0.027 & 0.063 & 0.075 \\
500 & 15 & 2 & 2 & 0.500 &  & 0.008 & 0.013 & 0.061 & 0.072 \\
500 & 15 & 2 & 3 & 0.500 & 0.033 & 0.010 & 0.012 & 0.058 & 0.077 \\
500 & 25 & 3 & 2 & 0.500 &  & 0.037 & 0.170 & 0.158 & 0.088 \\
500 & 25 & 3 & 3 & 0.500 & 0.471 & 0.076 & 0.298 & 0.228 & 0.110 \\
500 & 45 & 4 & 2 & 0.500 &  & 0.023 & 0.088 & 0.117 & 0.089 \\
\bottomrule
\end{tabular}
  \begin{tablenotes}
    \footnotesize
    \item Values displayed for all columns are the average, taken over all elements of the parameter, of the mean absolute error of estimation of that element over all replications.
  \end{tablenotes}
  \end{threeparttable}
\end{center}
\end{table}

\begin{table}
\small
\begin{center}
\caption{\label{tab:sim01part02} Parameter recovery, simulation study one (contd.)}
\vspace{0.5\baselineskip}
\begin{threeparttable}[t]
\begin{tabular}{rrrlrlrrrr}
\toprule
\(N\) & \(J\) & \(K\) & \(L\) & \(\rho\) & \(\gamma\) & \(\eta\) & \(R\) & \(\lambda\) & \(\xi\) \\
\midrule
1500 & 15 & 2 & 2 & 0.000 &  & 0.005 & 0.012 & 0.043 & 0.045 \\
1500 & 15 & 2 & 3 & 0.000 & 0.041 & 0.021 & 0.053 & 0.107 & 0.044 \\
1500 & 25 & 3 & 2 & 0.000 &  & 0.005 & 0.017 & 0.046 & 0.046 \\
1500 & 25 & 3 & 3 & 0.000 & 0.017 & 0.005 & 0.011 & 0.042 & 0.042 \\
1500 & 45 & 4 & 2 & 0.000 &  & 0.006 & 0.020 & 0.049 & 0.045 \\
1500 & 15 & 2 & 2 & 0.250 &  & 0.005 & 0.009 & 0.044 & 0.043 \\
1500 & 15 & 2 & 3 & 0.250 & 0.019 & 0.005 & 0.007 & 0.041 & 0.044 \\
1500 & 25 & 3 & 2 & 0.250 &  & 0.005 & 0.013 & 0.044 & 0.040 \\
1500 & 25 & 3 & 3 & 0.250 & 0.084 & 0.031 & 0.069 & 0.080 & 0.049 \\
1500 & 45 & 4 & 2 & 0.250 &  & 0.005 & 0.016 & 0.045 & 0.043 \\
1500 & 15 & 2 & 2 & 0.500 &  & 0.005 & 0.008 & 0.042 & 0.044 \\
1500 & 15 & 2 & 3 & 0.500 & 0.018 & 0.006 & 0.007 & 0.043 & 0.044 \\
1500 & 25 & 3 & 2 & 0.500 &  & 0.030 & 0.149 & 0.137 & 0.048 \\
1500 & 25 & 3 & 3 & 0.500 & 0.482 & 0.079 & 0.328 & 0.251 & 0.068 \\
1500 & 45 & 4 & 2 & 0.500 &  & 0.018 & 0.070 & 0.095 & 0.054 \\
3000 & 15 & 2 & 2 & 0.000 &  & 0.003 & 0.008 & 0.034 & 0.031 \\
3000 & 15 & 2 & 3 & 0.000 & 0.044 & 0.024 & 0.063 & 0.116 & 0.033 \\
3000 & 25 & 3 & 2 & 0.000 &  & 0.004 & 0.012 & 0.038 & 0.032 \\
3000 & 25 & 3 & 3 & 0.000 & 0.016 & 0.005 & 0.010 & 0.037 & 0.031 \\
3000 & 45 & 4 & 2 & 0.000 &  & 0.005 & 0.014 & 0.043 & 0.032 \\
3000 & 15 & 2 & 2 & 0.250 &  & 0.003 & 0.007 & 0.035 & 0.029 \\
3000 & 15 & 2 & 3 & 0.250 & 0.016 & 0.004 & 0.006 & 0.034 & 0.032 \\
3000 & 25 & 3 & 2 & 0.250 &  & 0.004 & 0.012 & 0.039 & 0.029 \\
3000 & 25 & 3 & 3 & 0.250 & 0.099 & 0.032 & 0.075 & 0.090 & 0.038 \\
3000 & 45 & 4 & 2 & 0.250 &  & 0.004 & 0.013 & 0.039 & 0.030 \\
3000 & 15 & 2 & 2 & 0.500 &  & 0.003 & 0.006 & 0.034 & 0.031 \\
3000 & 15 & 2 & 3 & 0.500 & 0.014 & 0.004 & 0.005 & 0.034 & 0.030 \\
3000 & 25 & 3 & 2 & 0.500 &  & 0.032 & 0.158 & 0.144 & 0.034 \\
3000 & 25 & 3 & 3 & 0.500 & 0.449 & 0.071 & 0.324 & 0.212 & 0.047 \\
3000 & 45 & 4 & 2 & 0.500 &  & 0.014 & 0.054 & 0.080 & 0.039 \\
\bottomrule
\end{tabular}
  \begin{tablenotes}
    \footnotesize
    \item Values displayed for all columns are the average, taken over all elements of the parameter, of the mean absolute error of estimation of that element over all replications.
  \end{tablenotes}
  \end{threeparttable}
\end{center}
\end{table}

\begin{table}
\small
\begin{center}
\caption{\label{tab:sim01part03} Parameter recovery, simulation study one (contd.)}
\vspace{0.5\baselineskip}
\begin{threeparttable}[t]
\begin{tabular}{rrrlrrrrrrr}
\toprule
\(N\) & \(J\) & \(K\) & \(L\) & \(\rho\) & \(\beta\) & \(\delta\) & \(\delta^0\) & \(\delta^1\) & \(\beta^0\) & \(\beta^1\) \\
\midrule
250 & 15 & 2 & 2 & 0.000 & 0.090 & 1.000 & 1.000 & 1.000 & 0.020 & 0.125 \\
250 & 15 & 2 & 3 & 0.000 & 0.245 & 0.930 & 0.875 & 0.973 & 0.120 & 0.345 \\
250 & 25 & 3 & 2 & 0.000 & 0.052 & 0.991 & 0.986 & 1.000 & 0.017 & 0.106 \\
250 & 25 & 3 & 3 & 0.000 & 0.079 & 0.979 & 0.997 & 0.933 & 0.014 & 0.243 \\
250 & 45 & 4 & 2 & 0.000 & 0.035 & 0.997 & 0.996 & 1.000 & 0.010 & 0.100 \\
250 & 15 & 2 & 2 & 0.250 & 0.083 & 0.975 & 0.925 & 1.000 & 0.033 & 0.107 \\
250 & 15 & 2 & 3 & 0.250 & 0.150 & 0.985 & 0.983 & 0.987 & 0.015 & 0.259 \\
250 & 25 & 3 & 2 & 0.250 & 0.050 & 0.989 & 0.981 & 1.000 & 0.014 & 0.105 \\
250 & 25 & 3 & 3 & 0.250 & 0.094 & 0.967 & 0.991 & 0.907 & 0.014 & 0.295 \\
250 & 45 & 4 & 2 & 0.250 & 0.036 & 0.994 & 0.992 & 1.000 & 0.010 & 0.101 \\
250 & 15 & 2 & 2 & 0.500 & 0.071 & 0.992 & 0.975 & 1.000 & 0.024 & 0.095 \\
250 & 15 & 2 & 3 & 0.500 & 0.127 & 0.974 & 0.975 & 0.973 & 0.028 & 0.207 \\
250 & 25 & 3 & 2 & 0.500 & 0.227 & 0.880 & 0.814 & 0.979 & 0.201 & 0.265 \\
250 & 25 & 3 & 3 & 0.500 & 0.305 & 0.873 & 0.868 & 0.885 & 0.190 & 0.594 \\
250 & 45 & 4 & 2 & 0.500 & 0.046 & 0.992 & 0.989 & 1.000 & 0.015 & 0.125 \\
500 & 15 & 2 & 2 & 0.000 & 0.051 & 0.990 & 0.970 & 1.000 & 0.014 & 0.069 \\
500 & 15 & 2 & 3 & 0.000 & 0.159 & 0.971 & 0.957 & 0.982 & 0.057 & 0.239 \\
500 & 25 & 3 & 2 & 0.000 & 0.037 & 0.993 & 0.989 & 1.000 & 0.011 & 0.077 \\
500 & 25 & 3 & 3 & 0.000 & 0.068 & 0.980 & 0.990 & 0.956 & 0.015 & 0.200 \\
500 & 45 & 4 & 2 & 0.000 & 0.031 & 0.994 & 0.992 & 0.999 & 0.009 & 0.087 \\
500 & 15 & 2 & 2 & 0.250 & 0.053 & 0.989 & 0.969 & 1.000 & 0.016 & 0.071 \\
500 & 15 & 2 & 3 & 0.250 & 0.116 & 0.990 & 0.990 & 0.990 & 0.010 & 0.201 \\
500 & 25 & 3 & 2 & 0.250 & 0.036 & 0.994 & 0.990 & 1.000 & 0.010 & 0.074 \\
500 & 25 & 3 & 3 & 0.250 & 0.115 & 0.959 & 0.969 & 0.936 & 0.044 & 0.295 \\
500 & 45 & 4 & 2 & 0.250 & 0.029 & 0.994 & 0.992 & 0.999 & 0.009 & 0.078 \\
500 & 15 & 2 & 2 & 0.500 & 0.053 & 0.988 & 0.965 & 1.000 & 0.017 & 0.071 \\
500 & 15 & 2 & 3 & 0.500 & 0.100 & 0.989 & 0.986 & 0.992 & 0.015 & 0.169 \\
500 & 25 & 3 & 2 & 0.500 & 0.160 & 0.921 & 0.883 & 0.976 & 0.137 & 0.195 \\
500 & 25 & 3 & 3 & 0.500 & 0.246 & 0.877 & 0.859 & 0.923 & 0.178 & 0.418 \\
500 & 45 & 4 & 2 & 0.500 & 0.083 & 0.959 & 0.950 & 0.982 & 0.053 & 0.161 \\
\bottomrule
\end{tabular}
  \begin{tablenotes}
    \footnotesize
    \item Values displayed for \(\beta\) parameters are the average, taken over all elements of the parameter, of the mean absolute error of estimation of each element over all replications. Values displayed for \(\delta\) parameters are the average, taken over all elements of the parameter, of the recovery accuracy of each element.
  \end{tablenotes}
  \end{threeparttable}
\end{center}
\end{table}

\begin{table}
\small
\begin{center}
\caption{\label{tab:sim01part04} Parameter recovery, simulation study one (contd.)}
\vspace{0.5\baselineskip}
\begin{threeparttable}[t]
\begin{tabular}{rrrlrrrrrrr}
\toprule
\(N\) & \(J\) & \(K\) & \(L\) & \(\rho\) & \(\beta\) & \(\delta\) & \(\delta^0\) & \(\delta^1\) & \(\beta^0\) & \(\beta^1\) \\
\midrule
1500 & 15 & 2 & 2 & 0.000 & 0.029 & 0.993 & 0.980 & 1.000 & 0.006 & 0.040 \\
1500 & 15 & 2 & 3 & 0.000 & 0.139 & 0.975 & 0.970 & 0.980 & 0.046 & 0.213 \\
1500 & 25 & 3 & 2 & 0.000 & 0.021 & 0.996 & 0.993 & 1.000 & 0.006 & 0.044 \\
1500 & 25 & 3 & 3 & 0.000 & 0.041 & 0.991 & 0.998 & 0.973 & 0.003 & 0.136 \\
1500 & 45 & 4 & 2 & 0.000 & 0.021 & 0.995 & 0.993 & 0.998 & 0.007 & 0.055 \\
1500 & 15 & 2 & 2 & 0.250 & 0.029 & 0.993 & 0.980 & 1.000 & 0.007 & 0.041 \\
1500 & 15 & 2 & 3 & 0.250 & 0.091 & 0.991 & 0.995 & 0.989 & 0.004 & 0.161 \\
1500 & 25 & 3 & 2 & 0.250 & 0.018 & 0.997 & 0.996 & 1.000 & 0.003 & 0.041 \\
1500 & 25 & 3 & 3 & 0.250 & 0.105 & 0.956 & 0.961 & 0.945 & 0.045 & 0.255 \\
1500 & 45 & 4 & 2 & 0.250 & 0.015 & 0.997 & 0.996 & 1.000 & 0.003 & 0.044 \\
1500 & 15 & 2 & 2 & 0.500 & 0.030 & 0.996 & 0.988 & 1.000 & 0.006 & 0.042 \\
1500 & 15 & 2 & 3 & 0.500 & 0.084 & 0.991 & 0.993 & 0.989 & 0.006 & 0.146 \\
1500 & 25 & 3 & 2 & 0.500 & 0.135 & 0.924 & 0.888 & 0.979 & 0.124 & 0.151 \\
1500 & 25 & 3 & 3 & 0.500 & 0.246 & 0.858 & 0.830 & 0.927 & 0.193 & 0.381 \\
1500 & 45 & 4 & 2 & 0.500 & 0.061 & 0.965 & 0.955 & 0.989 & 0.041 & 0.112 \\
3000 & 15 & 2 & 2 & 0.000 & 0.020 & 0.997 & 0.992 & 1.000 & 0.003 & 0.028 \\
3000 & 15 & 2 & 3 & 0.000 & 0.136 & 0.975 & 0.966 & 0.982 & 0.053 & 0.202 \\
3000 & 25 & 3 & 2 & 0.000 & 0.015 & 0.997 & 0.995 & 1.000 & 0.004 & 0.031 \\
3000 & 25 & 3 & 3 & 0.000 & 0.035 & 0.991 & 0.994 & 0.984 & 0.007 & 0.106 \\
3000 & 45 & 4 & 2 & 0.000 & 0.017 & 0.993 & 0.991 & 0.999 & 0.007 & 0.043 \\
3000 & 15 & 2 & 2 & 0.250 & 0.020 & 0.997 & 0.991 & 1.000 & 0.003 & 0.028 \\
3000 & 15 & 2 & 3 & 0.250 & 0.068 & 0.994 & 0.997 & 0.992 & 0.002 & 0.121 \\
3000 & 25 & 3 & 2 & 0.250 & 0.016 & 0.995 & 0.992 & 1.000 & 0.006 & 0.032 \\
3000 & 25 & 3 & 3 & 0.250 & 0.111 & 0.949 & 0.950 & 0.948 & 0.057 & 0.249 \\
3000 & 45 & 4 & 2 & 0.250 & 0.014 & 0.995 & 0.994 & 0.998 & 0.004 & 0.037 \\
3000 & 15 & 2 & 2 & 0.500 & 0.021 & 0.996 & 0.987 & 1.000 & 0.004 & 0.029 \\
3000 & 15 & 2 & 3 & 0.500 & 0.064 & 0.993 & 0.995 & 0.992 & 0.003 & 0.113 \\
3000 & 25 & 3 & 2 & 0.500 & 0.138 & 0.915 & 0.877 & 0.973 & 0.128 & 0.153 \\
3000 & 25 & 3 & 3 & 0.500 & 0.218 & 0.864 & 0.833 & 0.942 & 0.175 & 0.326 \\
3000 & 45 & 4 & 2 & 0.500 & 0.049 & 0.972 & 0.966 & 0.987 & 0.030 & 0.096 \\
\bottomrule
\end{tabular}
  \begin{tablenotes}
    \footnotesize
    \item Values displayed for \(\beta\) parameters are the average, taken over all elements of the parameter, of the mean absolute error of estimation of each element over all replications. Values displayed for \(\delta\) parameters are the average, taken over all elements of the parameter, of the recovery accuracy of each element.
  \end{tablenotes}
  \end{threeparttable}
\end{center}
\end{table}

\begin{table}
\small
\begin{center}
\caption{\label{3-tab:simresults1} Parameter recovery, simulation study two, no missing data}
\vspace{0.5\baselineskip}
\begin{threeparttable}[t]
\begin{tabular}{r|rllr|rrrrrr}
\toprule
\(N\) & \(J\) & \(K\) & \(L\) & \(\rho\) & \(\gamma\) & \(\eta\) & \(R\) & \(\lambda\) & \(\xi\)\\
\midrule
125 & 15 & 2 & 2 & 0.000 &  & 0.005 & 0.013 & 0.068 & 0.063\\
125 & 15 & 2 & 3 & 0.000 & 0.044 & 0.024 & 0.059 & 0.132 & 0.054\\
125 & 25 & 3 & 2 & 0.000 &  & 0.006 & 0.019 & 0.070 & 0.056\\
125 & 25 & 3 & 3 & 0.000 & 0.019 & 0.006 & 0.011 & 0.064 & 0.047\\
125 & 45 & 4 & 2 & 0.000 &  & 0.009 & 0.025 & 0.080 & 0.052\\
125 & 15 & 2 & 2 & 0.250 &  & 0.005 & 0.012 & 0.066 & 0.058\\
125 & 15 & 2 & 3 & 0.250 & 0.023 & 0.006 & 0.009 & 0.061 & 0.052\\
125 & 25 & 3 & 2 & 0.250 &  & 0.005 & 0.014 & 0.069 & 0.052\\
125 & 25 & 3 & 3 & 0.250 & 0.097 & 0.035 & 0.077 & 0.104 & 0.051\\
125 & 45 & 4 & 2 & 0.250 &  & 0.006 & 0.017 & 0.068 & 0.051\\
125 & 15 & 2 & 2 & 0.500 &  & 0.005 & 0.009 & 0.064 & 0.062\\
125 & 15 & 2 & 3 & 0.500 & 0.022 & 0.006 & 0.007 & 0.065 & 0.052\\
125 & 25 & 3 & 2 & 0.500 &  & 0.036 & 0.169 & 0.172 & 0.061\\
125 & 25 & 3 & 3 & 0.500 & 0.471 & 0.070 & 0.305 & 0.221 & 0.067\\
125 & 45 & 4 & 2 & 0.500 &  & 0.017 & 0.063 & 0.117 & 0.061\\
250 & 15 & 2 & 2 & 0.000 &  & 0.004 & 0.009 & 0.055 & 0.047\\
250 & 15 & 2 & 3 & 0.000 & 0.040 & 0.019 & 0.047 & 0.111 & 0.037\\
250 & 25 & 3 & 2 & 0.000 &  & 0.004 & 0.015 & 0.056 & 0.039\\
250 & 25 & 3 & 3 & 0.000 & 0.013 & 0.004 & 0.008 & 0.051 & 0.033\\
250 & 45 & 4 & 2 & 0.000 &  & 0.006 & 0.018 & 0.062 & 0.038\\
250 & 15 & 2 & 2 & 0.250 &  & 0.004 & 0.009 & 0.055 & 0.045\\
250 & 15 & 2 & 3 & 0.250 & 0.017 & 0.004 & 0.006 & 0.050 & 0.035\\
250 & 25 & 3 & 2 & 0.250 &  & 0.005 & 0.014 & 0.059 & 0.038\\
250 & 25 & 3 & 3 & 0.250 & 0.081 & 0.030 & 0.066 & 0.086 & 0.038\\
250 & 45 & 4 & 2 & 0.250 &  & 0.005 & 0.019 & 0.059 & 0.038\\
250 & 15 & 2 & 2 & 0.500 &  & 0.004 & 0.007 & 0.050 & 0.045\\
250 & 15 & 2 & 3 & 0.500 & 0.014 & 0.004 & 0.004 & 0.052 & 0.038\\
250 & 25 & 3 & 2 & 0.500 &  & 0.032 & 0.152 & 0.164 & 0.045\\
250 & 25 & 3 & 3 & 0.500 & 0.484 & 0.078 & 0.343 & 0.258 & 0.048\\
250 & 45 & 4 & 2 & 0.500 &  & 0.017 & 0.071 & 0.105 & 0.043\\
500 & 15 & 2 & 2 & 0.000 &  & 0.003 & 0.007 & 0.041 & 0.030\\
500 & 15 & 2 & 3 & 0.000 & 0.042 & 0.021 & 0.054 & 0.114 & 0.026\\
500 & 25 & 3 & 2 & 0.000 &  & 0.003 & 0.008 & 0.044 & 0.028\\
500 & 25 & 3 & 3 & 0.000 & 0.010 & 0.003 & 0.006 & 0.041 & 0.023\\
500 & 45 & 4 & 2 & 0.000 &  & 0.010 & 0.024 & 0.064 & 0.026\\
500 & 15 & 2 & 2 & 0.250 &  & 0.003 & 0.006 & 0.043 & 0.031\\
500 & 15 & 2 & 3 & 0.250 & 0.012 & 0.003 & 0.004 & 0.039 & 0.024\\
500 & 25 & 3 & 2 & 0.250 &  & 0.002 & 0.007 & 0.045 & 0.026\\
500 & 25 & 3 & 3 & 0.250 & 0.099 & 0.033 & 0.071 & 0.094 & 0.029\\
500 & 45 & 4 & 2 & 0.250 &  & 0.003 & 0.011 & 0.047 & 0.026\\
500 & 15 & 2 & 2 & 0.500 &  & 0.003 & 0.005 & 0.042 & 0.031\\
500 & 15 & 2 & 3 & 0.500 & 0.011 & 0.003 & 0.004 & 0.042 & 0.027\\
500 & 25 & 3 & 2 & 0.500 &  & 0.035 & 0.175 & 0.161 & 0.030\\
500 & 25 & 3 & 3 & 0.500 & 0.477 & 0.075 & 0.345 & 0.233 & 0.037\\
500 & 45 & 4 & 2 & 0.500 &  & 0.014 & 0.058 & 0.090 & 0.032\\
\bottomrule
\end{tabular}
  \begin{tablenotes}
    \footnotesize
    \item Values displayed for all columns are the average, taken over all elements of the parameter, of the mean absolute error of estimation of that element over all replications.
  \end{tablenotes}
  \end{threeparttable}
\end{center}
\end{table}

\begin{table}
\small
\begin{center}
\caption{\label{3-tab:simresults1contd} Parameter recovery, simulation study two, no missing data (contd.)}
\vspace{0.5\baselineskip}
\begin{threeparttable}[t]
\begin{tabular}{r|rllr|rrrrrr}
\toprule
\(N\) & \(J\) & \(K\) & \(L\) & \(\rho\) & \(\beta\) & \(\delta\) & \(\delta^0\) & \(\delta^1\) & \(\beta^0\) & \(\beta^1\) \\
\midrule
125 & 15 & 2 & 2 & 0.000 & 0.032 & 0.995 & 0.984 & 1.000 & 0.007 & 0.044 \\
125 & 15 & 2 & 3 & 0.000 & 0.147 & 0.974 & 0.966 & 0.980 & 0.051 & 0.223 \\
125 & 25 & 3 & 2 & 0.000 & 0.025 & 0.994 & 0.990 & 1.000 & 0.009 & 0.048 \\
125 & 25 & 3 & 3 & 0.000 & 0.043 & 0.989 & 0.997 & 0.970 & 0.004 & 0.141 \\
125 & 45 & 4 & 2 & 0.000 & 0.030 & 0.990 & 0.987 & 0.997 & 0.013 & 0.072 \\
125 & 15 & 2 & 2 & 0.250 & 0.032 & 0.994 & 0.983 & 1.000 & 0.007 & 0.045 \\
125 & 15 & 2 & 3 & 0.250 & 0.095 & 0.992 & 0.996 & 0.989 & 0.004 & 0.167 \\
125 & 25 & 3 & 2 & 0.250 & 0.020 & 0.997 & 0.995 & 1.000 & 0.003 & 0.044 \\
125 & 25 & 3 & 3 & 0.250 & 0.117 & 0.951 & 0.955 & 0.942 & 0.053 & 0.277 \\
125 & 45 & 4 & 2 & 0.250 & 0.017 & 0.996 & 0.994 & 0.999 & 0.005 & 0.049 \\
125 & 15 & 2 & 2 & 0.500 & 0.032 & 0.992 & 0.978 & 1.000 & 0.009 & 0.044 \\
125 & 15 & 2 & 3 & 0.500 & 0.080 & 0.990 & 0.990 & 0.990 & 0.007 & 0.139 \\
125 & 25 & 3 & 2 & 0.500 & 0.151 & 0.915 & 0.878 & 0.972 & 0.136 & 0.175 \\
125 & 25 & 3 & 3 & 0.500 & 0.213 & 0.881 & 0.858 & 0.936 & 0.163 & 0.341 \\
125 & 45 & 4 & 2 & 0.500 & 0.059 & 0.970 & 0.963 & 0.990 & 0.036 & 0.117 \\
250 & 15 & 2 & 2 & 0.000 & 0.022 & 0.996 & 0.987 & 1.000 & 0.004 & 0.030 \\
250 & 15 & 2 & 3 & 0.000 & 0.121 & 0.981 & 0.974 & 0.986 & 0.040 & 0.186 \\
250 & 25 & 3 & 2 & 0.000 & 0.019 & 0.994 & 0.990 & 1.000 & 0.008 & 0.034 \\
250 & 25 & 3 & 3 & 0.000 & 0.031 & 0.995 & 0.998 & 0.985 & 0.002 & 0.104 \\
250 & 45 & 4 & 2 & 0.000 & 0.020 & 0.993 & 0.991 & 0.998 & 0.009 & 0.048 \\
250 & 15 & 2 & 2 & 0.250 & 0.022 & 0.997 & 0.992 & 1.000 & 0.004 & 0.031 \\
250 & 15 & 2 & 3 & 0.250 & 0.071 & 0.993 & 0.996 & 0.991 & 0.003 & 0.126 \\
250 & 25 & 3 & 2 & 0.250 & 0.019 & 0.995 & 0.991 & 1.000 & 0.008 & 0.035 \\
250 & 25 & 3 & 3 & 0.250 & 0.098 & 0.957 & 0.960 & 0.952 & 0.044 & 0.234 \\
250 & 45 & 4 & 2 & 0.250 & 0.018 & 0.991 & 0.988 & 0.998 & 0.009 & 0.042 \\
250 & 15 & 2 & 2 & 0.500 & 0.023 & 0.997 & 0.991 & 1.000 & 0.004 & 0.032 \\
250 & 15 & 2 & 3 & 0.500 & 0.070 & 0.992 & 0.995 & 0.989 & 0.004 & 0.123 \\
250 & 25 & 3 & 2 & 0.500 & 0.140 & 0.916 & 0.878 & 0.974 & 0.136 & 0.145 \\
250 & 25 & 3 & 3 & 0.500 & 0.264 & 0.850 & 0.816 & 0.934 & 0.214 & 0.390 \\
250 & 45 & 4 & 2 & 0.500 & 0.062 & 0.961 & 0.950 & 0.988 & 0.043 & 0.111 \\
500 & 15 & 2 & 2 & 0.000 & 0.015 & 0.998 & 0.994 & 1.000 & 0.002 & 0.022 \\
500 & 15 & 2 & 3 & 0.000 & 0.117 & 0.983 & 0.974 & 0.990 & 0.046 & 0.174 \\
500 & 25 & 3 & 2 & 0.000 & 0.010 & 0.999 & 0.998 & 1.000 & 0.001 & 0.023 \\
500 & 25 & 3 & 3 & 0.000 & 0.023 & 0.996 & 0.998 & 0.992 & 0.001 & 0.077 \\
500 & 45 & 4 & 2 & 0.000 & 0.036 & 0.980 & 0.975 & 0.992 & 0.024 & 0.066 \\
500 & 15 & 2 & 2 & 0.250 & 0.016 & 0.997 & 0.992 & 1.000 & 0.002 & 0.022 \\
500 & 15 & 2 & 3 & 0.250 & 0.054 & 0.995 & 0.997 & 0.993 & 0.002 & 0.096 \\
500 & 25 & 3 & 2 & 0.250 & 0.010 & 0.999 & 0.998 & 1.000 & 0.001 & 0.023 \\
500 & 25 & 3 & 3 & 0.250 & 0.143 & 0.937 & 0.930 & 0.953 & 0.091 & 0.276 \\
500 & 45 & 4 & 2 & 0.250 & 0.011 & 0.995 & 0.993 & 1.000 & 0.005 & 0.028 \\
500 & 15 & 2 & 2 & 0.500 & 0.016 & 0.997 & 0.991 & 1.000 & 0.003 & 0.022 \\
500 & 15 & 2 & 3 & 0.500 & 0.056 & 0.993 & 0.997 & 0.990 & 0.002 & 0.100 \\
500 & 25 & 3 & 2 & 0.500 & 0.150 & 0.897 & 0.848 & 0.970 & 0.149 & 0.152 \\
500 & 25 & 3 & 3 & 0.500 & 0.234 & 0.848 & 0.811 & 0.942 & 0.195 & 0.333 \\
500 & 45 & 4 & 2 & 0.500 & 0.048 & 0.969 & 0.962 & 0.988 & 0.032 & 0.087 \\
\bottomrule
\end{tabular}
  \begin{tablenotes}
    \footnotesize
    \item Values displayed for \(\beta\) parameters are the average, taken over all elements of the parameter, of the mean absolute error of estimation of each element over all replications. Values displayed for \(\delta\) parameters are the average, taken over all elements of the parameter, of the recovery accuracy of each element over all replications.
  \end{tablenotes}
  \end{threeparttable}
\end{center}
\end{table}

\begin{table}
\small
\begin{center}
\caption{\label{3-tab:simresults2} Parameter recovery, simulation study two, 10\% missing data}
\vspace{0.5\baselineskip}
\begin{threeparttable}[t]
\begin{tabular}{r|rllr|rrrrrr}
\toprule
\(N\) & \(J\) & \(K\) & \(L\) & \(\rho\) & \(\gamma\) & \(\eta\) & \(R\) & \(\lambda\) & \(\xi\) \\
\midrule
125 & 15 & 2 & 2 & 0.000 &  & 0.005 & 0.014 & 0.072 & 0.070\\
125 & 15 & 2 & 3 & 0.000 & 0.042 & 0.022 & 0.056 & 0.125 & 0.058\\
125 & 25 & 3 & 2 & 0.000 &  & 0.005 & 0.017 & 0.076 & 0.070\\
125 & 25 & 3 & 3 & 0.000 & 0.021 & 0.006 & 0.012 & 0.068 & 0.050\\
125 & 45 & 4 & 2 & 0.000 &  & 0.009 & 0.028 & 0.084 & 0.083\\
125 & 15 & 2 & 2 & 0.250 &  & 0.005 & 0.013 & 0.070 & 0.064\\
125 & 15 & 2 & 3 & 0.250 & 0.025 & 0.006 & 0.010 & 0.064 & 0.055\\
125 & 25 & 3 & 2 & 0.250 &  & 0.005 & 0.016 & 0.074 & 0.066\\
125 & 25 & 3 & 3 & 0.250 & 0.113 & 0.036 & 0.075 & 0.109 & 0.063\\
125 & 45 & 4 & 2 & 0.250 &  & 0.006 & 0.023 & 0.082 & 0.078\\
125 & 15 & 2 & 2 & 0.500 &  & 0.006 & 0.009 & 0.072 & 0.069\\
125 & 15 & 2 & 3 & 0.500 & 0.024 & 0.007 & 0.008 & 0.068 & 0.054\\
125 & 25 & 3 & 2 & 0.500 &  & 0.034 & 0.161 & 0.172 & 0.073\\
125 & 25 & 3 & 3 & 0.500 & 0.531 & 0.074 & 0.293 & 0.226 & 0.077\\
125 & 45 & 4 & 2 & 0.500 &  & 0.017 & 0.071 & 0.122 & 0.090\\
250 & 15 & 2 & 2 & 0.000 &  & 0.004 & 0.009 & 0.058 & 0.052\\
250 & 15 & 2 & 3 & 0.000 & 0.028 & 0.013 & 0.029 & 0.085 & 0.043\\
250 & 25 & 3 & 2 & 0.000 &  & 0.004 & 0.012 & 0.061 & 0.058\\
250 & 25 & 3 & 3 & 0.000 & 0.016 & 0.004 & 0.009 & 0.055 & 0.035\\
250 & 45 & 4 & 2 & 0.000 &  & 0.007 & 0.022 & 0.071 & 0.073\\
250 & 15 & 2 & 2 & 0.250 &  & 0.004 & 0.010 & 0.057 & 0.050\\
250 & 15 & 2 & 3 & 0.250 & 0.019 & 0.004 & 0.007 & 0.052 & 0.039\\
250 & 25 & 3 & 2 & 0.250 &  & 0.005 & 0.014 & 0.064 & 0.057\\
250 & 25 & 3 & 3 & 0.250 & 0.086 & 0.033 & 0.070 & 0.093 & 0.048\\
250 & 45 & 4 & 2 & 0.250 &  & 0.005 & 0.022 & 0.074 & 0.068\\
250 & 15 & 2 & 2 & 0.500 &  & 0.004 & 0.007 & 0.054 & 0.050\\
250 & 15 & 2 & 3 & 0.500 & 0.016 & 0.005 & 0.005 & 0.052 & 0.039\\
250 & 25 & 3 & 2 & 0.500 &  & 0.030 & 0.142 & 0.161 & 0.066\\
250 & 25 & 3 & 3 & 0.500 & 0.501 & 0.081 & 0.329 & 0.262 & 0.059\\
250 & 45 & 4 & 2 & 0.500 &  & 0.017 & 0.084 & 0.113 & 0.077\\
500 & 15 & 2 & 2 & 0.000 &  & 0.003 & 0.007 & 0.046 & 0.038\\
500 & 15 & 2 & 3 & 0.000 & 0.036 & 0.018 & 0.048 & 0.104 & 0.031\\
500 & 25 & 3 & 2 & 0.000 &  & 0.003 & 0.009 & 0.050 & 0.046\\
500 & 25 & 3 & 3 & 0.000 & 0.013 & 0.003 & 0.007 & 0.044 & 0.025\\
500 & 45 & 4 & 2 & 0.000 &  & 0.010 & 0.029 & 0.073 & 0.067\\
500 & 15 & 2 & 2 & 0.250 &  & 0.003 & 0.006 & 0.046 & 0.037\\
500 & 15 & 2 & 3 & 0.250 & 0.014 & 0.003 & 0.005 & 0.041 & 0.028\\
500 & 25 & 3 & 2 & 0.250 &  & 0.003 & 0.008 & 0.049 & 0.046\\
500 & 25 & 3 & 3 & 0.250 & 0.097 & 0.031 & 0.070 & 0.093 & 0.042\\
500 & 45 & 4 & 2 & 0.250 &  & 0.003 & 0.018 & 0.065 & 0.059\\
500 & 15 & 2 & 2 & 0.500 &  & 0.003 & 0.005 & 0.047 & 0.041\\
500 & 15 & 2 & 3 & 0.500 & 0.012 & 0.003 & 0.004 & 0.044 & 0.029\\
500 & 25 & 3 & 2 & 0.500 &  & 0.034 & 0.170 & 0.164 & 0.058\\
500 & 25 & 3 & 3 & 0.500 & 0.479 & 0.079 & 0.336 & 0.235 & 0.048\\
500 & 45 & 4 & 2 & 0.500 &  & 0.013 & 0.064 & 0.100 & 0.072\\
\bottomrule
\end{tabular}
  \begin{tablenotes}
    \footnotesize
    \item Values displayed for all columns are the average, taken over all elements of the parameter, of the mean absolute error of estimation of that element over all replications.
  \end{tablenotes}
  \end{threeparttable}
\end{center}
\end{table}

\begin{table}
\small
\begin{center}
\caption{\label{3-tab:simresults2contd} Parameter recovery, simulation study two, 10\% missing data (contd.)}
\vspace{0.5\baselineskip}
\begin{threeparttable}[t]
\begin{tabular}{r|rllr|rrrrrr}
\toprule
\(N\) & \(J\) & \(K\) & \(L\) & \(\rho\) & \(\beta\) & \(\delta\) & \(\delta^0\) & \(\delta^1\) & \(\beta^0\) & \(\beta^1\) \\
\midrule
125 & 15 & 2 & 2 & 0.000 & 0.034 & 0.994 & 0.983 & 1.000 & 0.008 & 0.047 \\
125 & 15 & 2 & 3 & 0.000 & 0.145 & 0.973 & 0.968 & 0.977 & 0.046 & 0.223 \\
125 & 25 & 3 & 2 & 0.000 & 0.022 & 0.997 & 0.995 & 1.000 & 0.004 & 0.050 \\
125 & 25 & 3 & 3 & 0.000 & 0.044 & 0.990 & 0.997 & 0.971 & 0.004 & 0.143 \\
125 & 45 & 4 & 2 & 0.000 & 0.031 & 0.989 & 0.987 & 0.997 & 0.014 & 0.073 \\
125 & 15 & 2 & 2 & 0.250 & 0.034 & 0.994 & 0.982 & 1.000 & 0.008 & 0.047 \\
125 & 15 & 2 & 3 & 0.250 & 0.097 & 0.992 & 0.996 & 0.989 & 0.005 & 0.171 \\
125 & 25 & 3 & 2 & 0.250 & 0.021 & 0.997 & 0.996 & 1.000 & 0.004 & 0.047 \\
125 & 25 & 3 & 3 & 0.250 & 0.119 & 0.952 & 0.958 & 0.935 & 0.049 & 0.297 \\
125 & 45 & 4 & 2 & 0.250 & 0.020 & 0.994 & 0.992 & 0.998 & 0.006 & 0.055 \\
125 & 15 & 2 & 2 & 0.500 & 0.034 & 0.991 & 0.972 & 1.000 & 0.010 & 0.047 \\
125 & 15 & 2 & 3 & 0.500 & 0.084 & 0.990 & 0.990 & 0.990 & 0.008 & 0.145 \\
125 & 25 & 3 & 2 & 0.500 & 0.144 & 0.924 & 0.892 & 0.973 & 0.125 & 0.173 \\
125 & 25 & 3 & 3 & 0.500 & 0.217 & 0.886 & 0.870 & 0.925 & 0.157 & 0.370 \\
125 & 45 & 4 & 2 & 0.500 & 0.055 & 0.976 & 0.970 & 0.991 & 0.030 & 0.119 \\
250 & 15 & 2 & 2 & 0.000 & 0.023 & 0.995 & 0.985 & 1.000 & 0.004 & 0.032 \\
250 & 15 & 2 & 3 & 0.000 & 0.097 & 0.985 & 0.983 & 0.988 & 0.023 & 0.157 \\
250 & 25 & 3 & 2 & 0.000 & 0.015 & 0.998 & 0.996 & 1.000 & 0.002 & 0.035 \\
250 & 25 & 3 & 3 & 0.000 & 0.032 & 0.994 & 0.998 & 0.983 & 0.002 & 0.108 \\
250 & 45 & 4 & 2 & 0.000 & 0.023 & 0.992 & 0.990 & 0.996 & 0.010 & 0.057 \\
250 & 15 & 2 & 2 & 0.250 & 0.023 & 0.996 & 0.988 & 1.000 & 0.004 & 0.033 \\
250 & 15 & 2 & 3 & 0.250 & 0.074 & 0.993 & 0.996 & 0.991 & 0.003 & 0.131 \\
250 & 25 & 3 & 2 & 0.250 & 0.019 & 0.995 & 0.992 & 1.000 & 0.007 & 0.038 \\
250 & 25 & 3 & 3 & 0.250 & 0.104 & 0.956 & 0.959 & 0.947 & 0.044 & 0.255 \\
250 & 45 & 4 & 2 & 0.250 & 0.015 & 0.996 & 0.995 & 0.999 & 0.004 & 0.040 \\
250 & 15 & 2 & 2 & 0.500 & 0.024 & 0.997 & 0.991 & 1.000 & 0.004 & 0.034 \\
250 & 15 & 2 & 3 & 0.500 & 0.070 & 0.992 & 0.994 & 0.989 & 0.004 & 0.123 \\
250 & 25 & 3 & 2 & 0.500 & 0.131 & 0.924 & 0.889 & 0.975 & 0.122 & 0.145 \\
250 & 25 & 3 & 3 & 0.500 & 0.251 & 0.848 & 0.815 & 0.932 & 0.200 & 0.379 \\
250 & 45 & 4 & 2 & 0.500 & 0.061 & 0.961 & 0.951 & 0.987 & 0.042 & 0.111 \\
500 & 15 & 2 & 2 & 0.000 & 0.016 & 0.998 & 0.994 & 1.000 & 0.002 & 0.023 \\
500 & 15 & 2 & 3 & 0.000 & 0.110 & 0.984 & 0.977 & 0.989 & 0.040 & 0.167 \\
500 & 25 & 3 & 2 & 0.000 & 0.010 & 0.999 & 0.998 & 1.000 & 0.001 & 0.024 \\
500 & 25 & 3 & 3 & 0.000 & 0.024 & 0.997 & 0.999 & 0.991 & 0.001 & 0.080 \\
500 & 45 & 4 & 2 & 0.000 & 0.037 & 0.978 & 0.973 & 0.993 & 0.025 & 0.065 \\
500 & 15 & 2 & 2 & 0.250 & 0.016 & 0.997 & 0.992 & 1.000 & 0.003 & 0.023 \\
500 & 15 & 2 & 3 & 0.250 & 0.058 & 0.995 & 0.996 & 0.993 & 0.002 & 0.103 \\
500 & 25 & 3 & 2 & 0.250 & 0.010 & 0.999 & 0.998 & 1.000 & 0.001 & 0.024 \\
500 & 25 & 3 & 3 & 0.250 & 0.117 & 0.949 & 0.947 & 0.954 & 0.064 & 0.249 \\
500 & 45 & 4 & 2 & 0.250 & 0.010 & 0.997 & 0.995 & 1.000 & 0.004 & 0.028 \\
500 & 15 & 2 & 2 & 0.500 & 0.017 & 0.996 & 0.987 & 1.000 & 0.003 & 0.023 \\
500 & 15 & 2 & 3 & 0.500 & 0.055 & 0.994 & 0.995 & 0.993 & 0.002 & 0.097 \\
500 & 25 & 3 & 2 & 0.500 & 0.143 & 0.907 & 0.865 & 0.971 & 0.137 & 0.153 \\
500 & 25 & 3 & 3 & 0.500 & 0.233 & 0.848 & 0.815 & 0.933 & 0.184 & 0.357 \\
500 & 45 & 4 & 2 & 0.500 & 0.043 & 0.973 & 0.966 & 0.991 & 0.027 & 0.082 \\
\bottomrule
\end{tabular}
  \begin{tablenotes}
    \footnotesize
    \item Values displayed for \(\beta\) parameters are the average, taken over all elements of the parameter, of the mean absolute error of estimation of each element over all replications. Values displayed for \(\delta\) parameters are the average, taken over all elements of the parameter, of the recovery accuracy of each element over all replications.
  \end{tablenotes}
  \end{threeparttable}
\end{center}
\end{table}

\begin{table}
\small
\begin{center}
\caption{\label{3-tab:simresults3} Parameter recovery, simulation study two, 25\% missing data}
\vspace{0.5\baselineskip}
\begin{threeparttable}[t]
\begin{tabular}{r|rllr|rrrrrr}
\toprule
\(N\) & \(J\) & \(K\) & \(L\) & \(\rho\) & \(\gamma\) & \(\eta\) & \(R\) & \(\lambda\) & \(\xi\) \\
\midrule
125 & 15 & 2 & 2 & 0.000 &  & 0.006 & 0.016 & 0.087 & 0.091\\
125 & 15 & 2 & 3 & 0.000 & 0.037 & 0.016 & 0.037 & 0.121 & 0.074\\
125 & 25 & 3 & 2 & 0.000 &  & 0.006 & 0.019 & 0.089 & 0.113\\
125 & 25 & 3 & 3 & 0.000 & 0.026 & 0.007 & 0.015 & 0.078 & 0.060\\
125 & 45 & 4 & 2 & 0.000 &  & 0.008 & 0.029 & 0.099 & 0.173\\
125 & 15 & 2 & 2 & 0.250 &  & 0.006 & 0.014 & 0.081 & 0.080\\
125 & 15 & 2 & 3 & 0.250 & 0.030 & 0.007 & 0.012 & 0.066 & 0.067\\
125 & 25 & 3 & 2 & 0.250 &  & 0.006 & 0.018 & 0.085 & 0.115\\
125 & 25 & 3 & 3 & 0.250 & 0.116 & 0.040 & 0.070 & 0.120 & 0.094\\
125 & 45 & 4 & 2 & 0.250 &  & 0.006 & 0.032 & 0.121 & 0.157\\
125 & 15 & 2 & 2 & 0.500 &  & 0.006 & 0.011 & 0.087 & 0.089\\
125 & 15 & 2 & 3 & 0.500 & 0.026 & 0.007 & 0.010 & 0.070 & 0.061\\
125 & 25 & 3 & 2 & 0.500 &  & 0.029 & 0.138 & 0.171 & 0.132\\
125 & 25 & 3 & 3 & 0.500 & 0.548 & 0.080 & 0.247 & 0.223 & 0.113\\
125 & 45 & 4 & 2 & 0.500 &  & 0.016 & 0.090 & 0.145 & 0.174\\
250 & 15 & 2 & 2 & 0.000 &  & 0.004 & 0.009 & 0.068 & 0.067\\
250 & 15 & 2 & 3 & 0.000 & 0.026 & 0.012 & 0.028 & 0.094 & 0.056\\
250 & 25 & 3 & 2 & 0.000 &  & 0.004 & 0.015 & 0.078 & 0.106\\
250 & 25 & 3 & 3 & 0.000 & 0.020 & 0.005 & 0.011 & 0.061 & 0.044\\
250 & 45 & 4 & 2 & 0.000 &  & 0.007 & 0.024 & 0.095 & 0.168\\
250 & 15 & 2 & 2 & 0.250 &  & 0.004 & 0.011 & 0.066 & 0.063\\
250 & 15 & 2 & 3 & 0.250 & 0.021 & 0.005 & 0.009 & 0.055 & 0.053\\
250 & 25 & 3 & 2 & 0.250 &  & 0.004 & 0.014 & 0.078 & 0.114\\
250 & 25 & 3 & 3 & 0.250 & 0.098 & 0.035 & 0.071 & 0.109 & 0.081\\
250 & 45 & 4 & 2 & 0.250 &  & 0.005 & 0.032 & 0.117 & 0.152\\
250 & 15 & 2 & 2 & 0.500 &  & 0.004 & 0.008 & 0.067 & 0.069\\
250 & 15 & 2 & 3 & 0.500 & 0.019 & 0.005 & 0.006 & 0.054 & 0.045\\
250 & 25 & 3 & 2 & 0.500 &  & 0.024 & 0.117 & 0.160 & 0.131\\
250 & 25 & 3 & 3 & 0.500 & 0.539 & 0.087 & 0.283 & 0.254 & 0.099\\
250 & 45 & 4 & 2 & 0.500 &  & 0.017 & 0.102 & 0.136 & 0.171\\
500 & 15 & 2 & 2 & 0.000 &  & 0.003 & 0.008 & 0.057 & 0.058\\
500 & 15 & 2 & 3 & 0.000 & 0.023 & 0.010 & 0.027 & 0.085 & 0.046\\
500 & 25 & 3 & 2 & 0.000 &  & 0.003 & 0.010 & 0.065 & 0.096\\
500 & 25 & 3 & 3 & 0.000 & 0.016 & 0.003 & 0.009 & 0.051 & 0.035\\
500 & 45 & 4 & 2 & 0.000 &  & 0.010 & 0.030 & 0.100 & 0.166\\
500 & 15 & 2 & 2 & 0.250 &  & 0.003 & 0.007 & 0.057 & 0.053\\
500 & 15 & 2 & 3 & 0.250 & 0.017 & 0.003 & 0.006 & 0.043 & 0.042\\
500 & 25 & 3 & 2 & 0.250 &  & 0.003 & 0.010 & 0.064 & 0.111\\
500 & 25 & 3 & 3 & 0.250 & 0.103 & 0.034 & 0.076 & 0.110 & 0.077\\
500 & 45 & 4 & 2 & 0.250 &  & 0.004 & 0.030 & 0.114 & 0.149\\
500 & 15 & 2 & 2 & 0.500 &  & 0.003 & 0.005 & 0.064 & 0.070\\
500 & 15 & 2 & 3 & 0.500 & 0.013 & 0.004 & 0.005 & 0.047 & 0.034\\
500 & 25 & 3 & 2 & 0.500 &  & 0.030 & 0.150 & 0.168 & 0.135\\
500 & 25 & 3 & 3 & 0.500 & 0.513 & 0.088 & 0.274 & 0.228 & 0.088\\
500 & 45 & 4 & 2 & 0.500 &  & 0.012 & 0.083 & 0.127 & 0.173\\
\bottomrule
\end{tabular}
  \begin{tablenotes}
    \footnotesize
    \item Values displayed for all columns are the average, taken over all elements of the parameter, of the mean absolute error of estimation of that element over all replications.
  \end{tablenotes}
  \end{threeparttable}
\end{center}
\end{table}

\begin{table}
\small
\begin{center}
\caption{\label{3-tab:simresults3contd} Parameter recovery, simulation study two, 25\% missing data (contd.)}
\vspace{0.5\baselineskip}
\begin{threeparttable}[t]
\begin{tabular}{r|rllr|rrrrrr}
\toprule
\(N\) & \(J\) & \(K\) & \(L\) & \(\rho\) & \(\beta\) & \(\delta\) & \(\delta^0\) & \(\delta^1\) & \(\beta^0\) & \(\beta^1\) \\
\midrule
125 & 15 & 2 & 2 & 0.000 & 0.037 & 0.993 & 0.980 & 1.000 & 0.009 & 0.051 \\
125 & 15 & 2 & 3 & 0.000 & 0.127 & 0.980 & 0.978 & 0.981 & 0.031 & 0.203 \\
125 & 25 & 3 & 2 & 0.000 & 0.024 & 0.997 & 0.995 & 1.000 & 0.005 & 0.054 \\
125 & 25 & 3 & 3 & 0.000 & 0.048 & 0.988 & 0.997 & 0.965 & 0.005 & 0.157 \\
125 & 45 & 4 & 2 & 0.000 & 0.028 & 0.991 & 0.989 & 0.998 & 0.011 & 0.072 \\
125 & 15 & 2 & 2 & 0.250 & 0.038 & 0.993 & 0.980 & 1.000 & 0.009 & 0.052 \\
125 & 15 & 2 & 3 & 0.250 & 0.104 & 0.992 & 0.995 & 0.990 & 0.005 & 0.182 \\
125 & 25 & 3 & 2 & 0.250 & 0.024 & 0.996 & 0.994 & 1.000 & 0.004 & 0.052 \\
125 & 25 & 3 & 3 & 0.250 & 0.127 & 0.949 & 0.958 & 0.927 & 0.047 & 0.330 \\
125 & 45 & 4 & 2 & 0.250 & 0.017 & 0.998 & 0.997 & 1.000 & 0.003 & 0.053 \\
125 & 15 & 2 & 2 & 0.500 & 0.038 & 0.990 & 0.970 & 1.000 & 0.012 & 0.052 \\
125 & 15 & 2 & 3 & 0.500 & 0.088 & 0.990 & 0.990 & 0.990 & 0.009 & 0.151 \\
125 & 25 & 3 & 2 & 0.500 & 0.123 & 0.940 & 0.915 & 0.978 & 0.102 & 0.156 \\
125 & 25 & 3 & 3 & 0.500 & 0.224 & 0.889 & 0.883 & 0.907 & 0.148 & 0.416 \\
125 & 45 & 4 & 2 & 0.500 & 0.053 & 0.980 & 0.976 & 0.991 & 0.027 & 0.121 \\
250 & 15 & 2 & 2 & 0.000 & 0.025 & 0.995 & 0.984 & 1.000 & 0.005 & 0.035 \\
250 & 15 & 2 & 3 & 0.000 & 0.099 & 0.985 & 0.985 & 0.986 & 0.021 & 0.162 \\
250 & 25 & 3 & 2 & 0.000 & 0.017 & 0.997 & 0.995 & 1.000 & 0.003 & 0.037 \\
250 & 25 & 3 & 3 & 0.000 & 0.036 & 0.992 & 0.998 & 0.978 & 0.003 & 0.120 \\
250 & 45 & 4 & 2 & 0.000 & 0.024 & 0.992 & 0.991 & 0.996 & 0.010 & 0.061 \\
250 & 15 & 2 & 2 & 0.250 & 0.025 & 0.997 & 0.990 & 1.000 & 0.005 & 0.036 \\
250 & 15 & 2 & 3 & 0.250 & 0.079 & 0.993 & 0.995 & 0.991 & 0.003 & 0.139 \\
250 & 25 & 3 & 2 & 0.250 & 0.018 & 0.996 & 0.993 & 1.000 & 0.005 & 0.038 \\
250 & 25 & 3 & 3 & 0.250 & 0.110 & 0.955 & 0.961 & 0.940 & 0.042 & 0.281 \\
250 & 45 & 4 & 2 & 0.250 & 0.016 & 0.995 & 0.994 & 0.998 & 0.005 & 0.044 \\
250 & 15 & 2 & 2 & 0.500 & 0.027 & 0.996 & 0.987 & 1.000 & 0.005 & 0.037 \\
250 & 15 & 2 & 3 & 0.500 & 0.072 & 0.991 & 0.993 & 0.990 & 0.005 & 0.126 \\
250 & 25 & 3 & 2 & 0.500 & 0.112 & 0.940 & 0.914 & 0.980 & 0.097 & 0.133 \\
250 & 25 & 3 & 3 & 0.500 & 0.252 & 0.857 & 0.837 & 0.908 & 0.186 & 0.419 \\
250 & 45 & 4 & 2 & 0.500 & 0.061 & 0.964 & 0.955 & 0.988 & 0.039 & 0.115 \\
500 & 15 & 2 & 2 & 0.000 & 0.018 & 0.998 & 0.993 & 1.000 & 0.003 & 0.026 \\
500 & 15 & 2 & 3 & 0.000 & 0.086 & 0.988 & 0.987 & 0.989 & 0.020 & 0.140 \\
500 & 25 & 3 & 2 & 0.000 & 0.012 & 0.998 & 0.997 & 1.000 & 0.001 & 0.027 \\
500 & 25 & 3 & 3 & 0.000 & 0.026 & 0.995 & 0.998 & 0.989 & 0.002 & 0.086 \\
500 & 45 & 4 & 2 & 0.000 & 0.034 & 0.982 & 0.977 & 0.994 & 0.022 & 0.065 \\
500 & 15 & 2 & 2 & 0.250 & 0.018 & 0.996 & 0.989 & 1.000 & 0.003 & 0.026 \\
500 & 15 & 2 & 3 & 0.250 & 0.062 & 0.994 & 0.996 & 0.993 & 0.002 & 0.110 \\
500 & 25 & 3 & 2 & 0.250 & 0.011 & 0.999 & 0.998 & 1.000 & 0.001 & 0.026 \\
500 & 25 & 3 & 3 & 0.250 & 0.122 & 0.953 & 0.954 & 0.948 & 0.061 & 0.276 \\
500 & 45 & 4 & 2 & 0.250 & 0.012 & 0.996 & 0.995 & 0.999 & 0.004 & 0.033 \\
500 & 15 & 2 & 2 & 0.500 & 0.018 & 0.996 & 0.989 & 1.000 & 0.003 & 0.026 \\
500 & 15 & 2 & 3 & 0.500 & 0.061 & 0.992 & 0.995 & 0.991 & 0.003 & 0.107 \\
500 & 25 & 3 & 2 & 0.500 & 0.124 & 0.924 & 0.890 & 0.976 & 0.115 & 0.138 \\
500 & 25 & 3 & 3 & 0.500 & 0.236 & 0.854 & 0.831 & 0.910 & 0.164 & 0.417 \\
500 & 45 & 4 & 2 & 0.500 & 0.038 & 0.980 & 0.975 & 0.993 & 0.022 & 0.078 \\
\bottomrule
\end{tabular}
  \begin{tablenotes}
    \footnotesize
    \item Values displayed for \(\beta\) parameters are the average, taken over all elements of the parameter, of the mean absolute error of estimation of each element over all replications. Values displayed for \(\delta\) parameters are the average, taken over all elements of the parameter, of the recovery accuracy of each element over all replications.
  \end{tablenotes}
  \end{threeparttable}
\end{center}
\end{table}

\begin{table}
\small
\begin{center}
\caption{\label{tab:sim1timings} Run time in minutes for one replication, simulation study one}
\vspace{0.5\baselineskip}
\begin{tabular}{lccc|cccc}
\toprule
      &       &       &            & \multicolumn{4}{c}{\(N\)} \\
\(J\) & \(K\) & \(L\) & \(\rho\)   & 250 & 500 & 1500 & 3000 \\
\midrule
15 & 2 & 2 & 0.000 & 0.837 & 1.720 &  7.135 & 12.769 \\
15 & 2 & 3 & 0.000 & 1.093 & 3.108 &  9.050 & 19.434 \\
25 & 3 & 2 & 0.000 & 2.108 & 4.346 & 12.936 & 27.174 \\
25 & 3 & 3 & 0.000 & 4.082 & 6.811 & 20.498 & 35.841 \\
45 & 4 & 2 & 0.000 & 4.510 & 7.853 & 23.255 & 42.831 \\
15 & 2 & 2 & 0.250 & 0.736 & 2.333 &  7.155 & 13.256 \\
15 & 2 & 3 & 0.250 & 1.049 & 3.525 &  8.610 & 19.550 \\
25 & 3 & 2 & 0.250 & 2.302 & 4.586 & 11.971 & 26.146 \\
25 & 3 & 3 & 0.250 & 3.720 & 6.880 & 19.108 & 39.587 \\
45 & 4 & 2 & 0.250 & 4.372 & 8.177 & 21.912 & 44.929 \\
15 & 2 & 2 & 0.500 & 1.251 & 2.382 &  6.690 & 13.461 \\
15 & 2 & 3 & 0.500 & 1.748 & 3.243 &  9.337 & 18.968 \\
25 & 3 & 2 & 0.500 & 2.350 & 4.463 & 13.858 & 26.841 \\
25 & 3 & 3 & 0.500 & 3.554 & 6.747 & 18.498 & 37.894 \\
45 & 4 & 2 & 0.500 & 4.035 & 7.894 & 22.399 & 45.233 \\
\bottomrule
\end{tabular}
\end{center}
\end{table}

\begin{table}
\small
\begin{center}
\caption{\label{tab:sim2timings1} Run time in minutes for one replication, simulation study two, no missing data}
\vspace{0.5\baselineskip}
\begin{tabular}{lccc|rrr}
\toprule
      &       &       &            & \multicolumn{3}{c}{\(N\)} \\
\(J\) & \(K\) & \(L\) & \(\rho\)   & 125 & 250 & 500 \\
\midrule
15 & 2 & 2 & 0.000 &  6.643 & 13.041 & 25.839 \\
15 & 2 & 3 & 0.000 &  9.232 & 17.963 & 35.094 \\
25 & 3 & 2 & 0.000 & 11.431 & 24.459 & 47.083 \\
25 & 3 & 3 & 0.000 & 17.802 & 36.123 & 73.727 \\
45 & 4 & 2 & 0.000 & 19.921 & 43.303 & 83.551 \\
15 & 2 & 2 & 0.250 &  6.472 & 12.851 & 26.553 \\
15 & 2 & 3 & 0.250 &  8.984 & 18.254 & 37.436 \\
25 & 3 & 2 & 0.250 & 11.671 & 24.399 & 47.519 \\
25 & 3 & 3 & 0.250 & 16.617 & 37.069 & 75.249 \\
45 & 4 & 2 & 0.250 & 19.896 & 41.895 & 86.916 \\
15 & 2 & 2 & 0.500 &  6.456 & 13.789 & 27.225 \\
15 & 2 & 3 & 0.500 &  9.575 & 18.843 & 38.113 \\
25 & 3 & 2 & 0.500 & 11.581 & 24.095 & 49.575 \\
25 & 3 & 3 & 0.500 & 17.792 & 35.091 & 72.375 \\
45 & 4 & 2 & 0.500 & 20.874 & 44.594 & 86.447 \\
\bottomrule
\end{tabular}
\end{center}
\end{table}

\begin{table}
\small
\begin{center}
\caption{\label{tab:sim2timings2} Run time in minutes for one replication, simulation study two, 10\% missing data}
\vspace{0.5\baselineskip}
\begin{tabular}{lccc|rrr}
\toprule
      &       &       &            & \multicolumn{3}{c}{\(N\)} \\
\(J\) & \(K\) & \(L\) & \(\rho\)   & 125 & 250 & 500 \\
\midrule
15 & 2 & 2 & 0.000 & 7.063 & 14.389 & 26.246 \\
15 & 2 & 3 & 0.000 & 9.820 & 20.723 & 37.883 \\
25 & 3 & 2 & 0.000 & 13.469 & 26.487 & 51.017 \\
25 & 3 & 3 & 0.000 & 20.010 & 40.209 & 78.351 \\
45 & 4 & 2 & 0.000 & 22.675 & 45.778 & 89.979 \\
15 & 2 & 2 & 0.250 & 6.898 & 14.047 & 27.977 \\
15 & 2 & 3 & 0.250 & 10.050 & 20.306 & 39.686 \\
25 & 3 & 2 & 0.250 & 13.512 & 26.757 & 53.323 \\
25 & 3 & 3 & 0.250 & 20.577 & 41.333 & 80.077 \\
45 & 4 & 2 & 0.250 & 22.775 & 46.790 & 89.617 \\
15 & 2 & 2 & 0.500 & 6.831 & 13.822 & 27.761 \\
15 & 2 & 3 & 0.500 & 9.804 & 21.147 & 39.774 \\
25 & 3 & 2 & 0.500 & 13.147 & 27.063 & 54.766 \\
25 & 3 & 3 & 0.500 & 19.856 & 40.323 & 79.222 \\
45 & 4 & 2 & 0.500 & 22.472 & 45.280 & 90.526 \\
\bottomrule
\end{tabular}
\end{center}
\end{table}

\begin{table}
\small
\begin{center}
\caption{\label{tab:sim2timings3} Run time in minutes for one replication, simulation study two, 25\% missing data}
\vspace{0.5\baselineskip}
\begin{tabular}{lccc|rrr}
\toprule
      &       &       &            & \multicolumn{3}{c}{\(N\)} \\
\(J\) & \(K\) & \(L\) & \(\rho\)   & 125 & 250 & 500 \\
\midrule
15 & 2 & 2 & 0.000 & 7.222 & 15.518 & 28.649 \\
15 & 2 & 3 & 0.000 & 9.949 & 20.752 & 39.413 \\
25 & 3 & 2 & 0.000 & 13.510 & 27.679 & 52.755 \\
25 & 3 & 3 & 0.000 & 20.458 & 40.719 & 80.557 \\
45 & 4 & 2 & 0.000 & 23.333 & 47.644 & 94.864 \\
15 & 2 & 2 & 0.250 & 7.217 & 15.038 & 27.172 \\
15 & 2 & 3 & 0.250 & 10.313 & 21.235 & 37.587 \\
25 & 3 & 2 & 0.250 & 13.769 & 27.716 & 51.922 \\
25 & 3 & 3 & 0.250 & 20.852 & 41.941 & 77.197 \\
45 & 4 & 2 & 0.250 & 23.691 & 47.649 & 88.378 \\
15 & 2 & 2 & 0.500 & 7.213 & 14.959 & 27.272 \\
15 & 2 & 3 & 0.500 & 10.186 & 19.668 & 39.221 \\
25 & 3 & 2 & 0.500 & 13.838 & 27.555 & 51.130 \\
25 & 3 & 3 & 0.500 & 20.264 & 40.087 & 77.899 \\
45 & 4 & 2 & 0.500 & 24.701 & 47.713 & 90.792 \\
\bottomrule
\end{tabular}
\end{center}
\end{table}

\section*{Supplementary Material I}

Supplementary Material I derives the conditional likelihood for this model, which is as follows: \begin{align}
  p(y_n \mid \theta, \alpha_n) & = \prod_{t=1}^T p(y_n^t \mid \theta, \alpha_n^t) \notag\\
  & = \prod_{t=1}^T \prod_{j=1}^J p(y_{nj}^t \mid \theta, \alpha_n^t) \notag\\
  & = \prod_{t=1}^T \prod_{j=1}^J \int p(y_{nj}^{\ast, t}, y_{nj}^t \mid \theta, \alpha_n^t) d y_{nj}^{\ast, t} \notag\\
  & = \prod_{t=1}^T \prod_{j=1}^J \int p(y_{nj}^t \mid y_{nj}^{\ast, t}, \theta, \alpha_n^t) p(y_{nj}^{\ast, t} \mid \theta, \alpha_n^t) d y_{nj}^{\ast, t} \notag\\
  & = \prod_{t=1}^T \prod_{j=1}^J \int p(y_{nj}^t \mid y_{nj}^{\ast, t}, \kappa_j) p(y_{nj}^{\ast, t} \mid \beta_j, \alpha_n^t) d y_{nj}^{\ast, t} \notag\\
  & = \prod_{t=1}^T \prod_{j=1}^J \int I(y_{nj}^{\ast, t} \in [\kappa_{y_{nj}^t}, \kappa_{y_{nj}^t + 1})) \cdot \phi(y_{nj}^{\ast, t}; d_n^t \beta_j, 1) d y_{nj}^{\ast, t} \notag\\
  & = \prod_{t=1}^T \prod_{j=1}^J \left[\Phi(\kappa_{y_{nj}^t + 1} - d_n^t \beta_j) - \Phi(\kappa_{y_{nj}^t} - d_n^t \beta_j)\right]
\end{align}

\section*{Supplementary Material J}

Supplementary Material J displays the 95\% equal-tail credible intervals for coefficients of \(\lambda\) and \(\xi\) for the education application.

\begin{table}[H]
\small
\begin{center}
\caption{95\% equal-tail credible intervals for \(\lambda\) coefficients, education application}
\vspace{0.5\baselineskip}
\begin{tabular}{lrrr}
\toprule
Attribute & Intercept & Diagnosis & Traditional \\
\midrule
1 & (0.57, 1.27)   & (-0.95, -0.36) & (-1.01, -0.39) \\
2 & (-0.48, -0.03) & (0.06, 0.54)   & (0.04, 0.53)   \\
3 & (-2.16, -1.50) & (0.37, 0.94)   & (0.21, 0.79)   \\
4 & (-1.36, -0.76) & (0.09, 0.64)   & (-0.19, 0.37)  \\
5 & (-0.38, 0.06)  & (0.02, 0.56)   & (-0.15, 0.37)  \\
6 & (-0.75, -0.31) & (0.07, 0.58)   & (-0.03, 0.48)  \\
\bottomrule
\end{tabular}
\end{center}
\end{table}

\begin{table}[H]
\small
\begin{center}
\caption{95\% equal-tail credible intervals for \(\xi\) coefficients, education application}
\vspace{0.5\baselineskip}
\begin{tabular}{lcccccc}
\toprule
& \multicolumn{6}{c}{Attributes at time \(t\)} \\
Attr. at \(t - 1\) & 1 & 2 & 3 & 4 & 5 & 6 \\
\midrule
1 & (0.43, 1.18)   & (-0.66, 0.04) & (-0.50, 0.24) & (-0.60, 0.20) & (0.26, 1.07)   & (0.04, 0.74) \\
2 & (-0.53, 0.22)  & (0.95, 1.65)  & (-0.58, 0.10) & (-0.59, 0.16) & (0.43, 1.27)   & (-0.07, 0.68) \\
3 & (-0.04, 0.72)  & (-0.59, 0.09) & (0.86, 1.60)  & (-0.26, 0.52) & (-0.08, 0.82)  & (0.03, 0.85) \\
4 & (-0.57, 0.12)  & (-0.22, 0.53) & (-0.52, 0.23) & (0.78, 1.48)  & (0.05, 0.89)   & (0.14, 0.85) \\
5 & (-0.88, -0.07) & (-0.57, 0.02) & (0.30, 0.93)  & (0.00, 0.80)  & (2.13, 2.95)   & (0.80, 1.46) \\
6 & (-0.88, -0.13) & (-0.16, 0.48) & (-0.36, 0.35) & (0.14, 0.80)  & (0.94, 1.82)   & (1.75, 2.52) \\
Intercept & (-1.15, -0.02) & (-0.36, 0.59) & (0.21, 1.31)  & (-0.71, 0.31) & (-3.27, -2.02) & (-2.47, -1.33) \\
\bottomrule
\end{tabular}
\end{center}
\end{table}

\section*{Supplementary Material K}

Supplementary Material K contains additional information on the emotional state application.

\subsection*{Details on Handling Time Points for Missingness}

Since the particular positions of the missing time points for each respondent for each day were not recorded in the dataset (only the date and time of each request was recorded), users of this dataset can see for each respondent how many responses there were per day but not the indices of missing positions. We thus computed an estimate of these missing time points as follows.

First, we created an empirical distribution of \(10^6\) draws of six time points from a process that approximates the one described above that was used by the researchers who collected the data. We converted the time points to integer values, where rounding to the nearest integer implies a binning of the data. For each respondent, for each day where the respondent had at least one time point missing, we found the subset of vectors from the empirical distribution which contain the observed vector of integer time points, and took one sample from the empirical distribution to get a hypothetical full vector of six integer time points for that respondent for that day. That gave us a vector (hypothesized value) of missing time positions and missing time points for each respondent-day that had missing data.

\subsection*{Initializations of Missing Data Values}

For initial values of missing data rows, since the data was collected over a period of five days with six possible measurements a day (that is, if there were six measurements in a day there was no missing data for that day), if \(t\) is the first possible measurement for a day, we initialized \(Y_n^t\) with the first non-missing value following \(t\) for that day, and then for each \(t \in (t_n^1, \ldots, t_n^i)\), if \(t\) is the second or later possible measurement for a day, we initialized \(Y_n^t\) with the first non-missing value preceding \(t\) within that day.

\subsection*{Further Data Analysis Results}

\begin{table}
\small
\begin{center}
\caption{\label{3-tab:lambda} \(\lambda\) coefficients estimates, emotional state application}
\vspace{0.5\baselineskip}
\begin{tabular}{lrrr}
\toprule
& \multicolumn{3}{c}{Attribute} \\
 & 1 & 2 & 3 \\
\midrule
Intercept &  0.32 &  0.51 &  0.54 \\
Presence  & -0.04 &  0.11 & -0.05 \\
Search    &  0.10 & -0.13 &  0.06 \\
Afternoon &  0.21 & -0.10 & -0.01 \\
Evening   &  0.45 & -0.22 & -0.10 \\
\bottomrule
\end{tabular}
\end{center}
\end{table}

\begin{table}
\small
\begin{center}
\caption{\label{3-tab:lambdacredint} 95\% equal-tail credible intervals for \(\lambda\) coefficients, emotional state application}
\vspace{0.5\baselineskip}
\begin{tabular}{lccc}
\toprule
& \multicolumn{3}{c}{Attribute} \\
 & 1 & 2 & 3 \\
\midrule
Intercept & (0.06, 0.59) & (0.23, 0.79) & (0.31, 0.77) \\
Presence  & (-0.09, 0.01) & (0.06, 0.17) & (-0.1, -0.01) \\
Search    & (0.04, 0.15) & (-0.19, -0.07) & (0.01, 0.11) \\
Afternoon & (0.09, 0.33) & (-0.23, 0.04) & (-0.12, 0.1) \\
Evening   & (0.32, 0.57) & (-0.37, -0.08) & (-0.22, 0.01) \\
\bottomrule
\end{tabular}
\end{center}
\end{table}

\begin{table}
\small
\begin{center}
\caption{\label{3-tab:xi} \(\xi\) coefficients estimates, emotional state application}
\vspace{0.5\baselineskip}
\begin{tabular}{crrr}
\toprule
& \multicolumn{3}{c}{Attribute at time \(t\)} \\
Effect from time \(t - 1\) & 1 & 2 & 3 \\
\midrule
Intercept     & -0.42 & -0.75 & -0.90 \\
\([0, 0, 1]\) & -0.13 & -0.19 & 1.02 \\
\([0, 0, 2]\) & 0.12 & -0.00 & 0.51 \\
\([0, 1, 0]\) & 0.21 & 0.52 & -0.09 \\
\([0, 2, 0]\) & 0.17 & 0.90 & -0.27 \\
\([1, 0, 0]\) & 0.39 & 0.19 & 0.03 \\
\([2, 0, 0]\) & 0.38 & 0.20 & -0.04 \\
\bottomrule
\end{tabular}
\end{center}
\end{table}

\begin{table}
\small
\begin{center}
\caption{\label{3-tab:xicredint} 95\% equal-tail credible intervals for \(\xi\) coefficients, emotional state application}
\vspace{0.5\baselineskip}
\begin{tabular}{cccc}
\toprule
& \multicolumn{3}{c}{Attribute at time \(t\)} \\
Effect from time \(t - 1\) & 1 & 2 & 3 \\
\midrule
Intercept     & (-0.73, -0.12) & (-1.07, -0.42) & (-1.18, -0.62) \\
\([0, 0, 1]\) & (-0.25, -0.01) & (-0.32, -0.07) & (0.91, 1.14) \\
\([0, 0, 2]\) & (-0.02, 0.26) & (-0.16, 0.16) & (0.38, 0.64) \\
\([0, 1, 0]\) & (0.1, 0.32) & (0.39, 0.65) & (-0.2, 0.02) \\
\([0, 2, 0]\) & (-0.0, 0.34) & (0.7, 1.11) & (-0.45, -0.09) \\
\([1, 0, 0]\) & (0.26, 0.52) & (0.05, 0.34) & (-0.08, 0.15) \\
\([2, 0, 0]\) & (0.26, 0.49) & (0.08, 0.32) & (-0.15, 0.07) \\
\bottomrule
\end{tabular}
\end{center}
\end{table}

\begin{table}
\small
\begin{center}
\caption{\label{3-tab:rmat} \(R\) (underlying correlation matrix) estimate, emotional state application}
\vspace{0.5\baselineskip}
\begin{tabular}{lrrr}
\toprule
 & 1 & 2 & 3 \\
\midrule
1 & 1.00 & -0.21 & -0.47 \\
2 & -0.21 & 1.00 & -0.43 \\
3 & -0.47 & -0.43 & 1.00 \\
\bottomrule
\end{tabular}
\end{center}
\end{table}

\end{document}